\documentclass{daj}

\usepackage{booktabs}
\usepackage{titlesec}
\usepackage{amsmath,amsfonts,amsthm,mathrsfs,xspace,graphicx}
\usepackage{nicefrac}
\usepackage{color}
\usepackage{float}
\usepackage{xcolor}
\usepackage{mdframed}
\usepackage{bbm}
\usepackage{suffix} 
\usepackage{times}
\usepackage{tabularx}
\usepackage{makecell}
\usepackage{amssymb,latexsym}
\usepackage[capitalize]{cleveref}
\usepackage{enumitem}
\usepackage{multirow}
\usepackage[section]{placeins}
\usepackage{comment}
\usepackage{footmisc}

\mdfdefinestyle{figstyle}{ %
  linecolor=black!7, %
  backgroundcolor=black!7, %
  innertopmargin=10pt, %
  innerleftmargin=25pt, %
  innerrightmargin=25pt, %
  innerbottommargin=10pt %
}

\newtheorem{theorem}{Theorem}[section]
\newtheorem{theorem*}{Theorem}

\newtheorem{proposition}[theorem]{Proposition}

\newtheorem{lemma}[theorem]{Lemma}
\newtheorem{claim}[theorem]{Claim}

\newtheorem{corollary}[theorem]{Corollary}

\newtheorem{remark}[theorem]{Remark}

\theoremstyle{definition}
\newtheorem{definition}[theorem]{Definition}

\newcommand{\beq}{\begin{eqnarray}}
\newcommand{\eeq}{\end{eqnarray}}

\newcommand{\code}{\mathscr{C}}
\newcommand{\strategy}{\mathscr{S}}

\newcommand{\Tr}{\mbox{\rm Tr}}
\newcommand{\Id}{\ensuremath{I}}
\DeclareMathOperator*{\Expectation}{\mathbb{E}}
\newcommand{\Es}[1]{\Expectation_{#1}}

\newcommand{\field}{\mathbb{F}_2}
\newcommand{\C}{\ensuremath{\mathbb{C}}}
\newcommand{\N}{\ensuremath{\mathbb{N}}}

\newcommand{\complex}{\ensuremath{\mathbb{C}}}

\newcommand{\dlS}{\ensuremath{\rm dlS}}

\newcommand{\F}{\ensuremath{\mathbb{F}}}

\newcommand{\ld}{\textsc{ld}}
\newcommand{\com}{\textsc{com}}
\newcommand{\sq}{\textsc{sq}}

\newcommand{\R}{\ensuremath{\mathbb{R}}}
\newcommand{\Z}{\ensuremath{\mathbb{Z}}}

\newcommand{\mA}{\ensuremath{\mathcal{A}}}
\newcommand{\mB}{\ensuremath{\mathcal{B}}}
\newcommand{\mC}{\ensuremath{\mathcal{C}}}
\newcommand{\mE}{\ensuremath{\mathcal{E}}}

\newcommand{\mF}{\ensuremath{\mathcal{F}}}

\newcommand{\mH}{\ensuremath{\mathcal{H}}}
\newcommand{\mK}{\ensuremath{\mathcal{K}}}
\newcommand{\mM}{\ensuremath{\mathcal{M}}}
\newcommand{\mI}{\ensuremath{\mathcal{I}}}

\newcommand{\cM}{\ensuremath{\mathcal{M}}}
\newcommand{\mP}{\ensuremath{\mathcal{P}}}

\newcommand{\mR}{\ensuremath{\mathcal{R}}}

\newcommand{\mU}{\ensuremath{\mathcal{U}}}
\newcommand{\mX}{\ensuremath{\mathcal{X}}}

\newcommand{\RM}{\ensuremath{\textsc{RM}}}
\newcommand{\RS}{\ensuremath{\textsc{RS}}}
\newcommand{\bRM}{\ensuremath{\textsc{RM2}}}
\newcommand{\Had}{\ensuremath{\textsc{Had}}}

\newcommand{\cc}{\mathrm{com}}
\newcommand{\ac}{\mathrm{ac}}

\DeclareMathOperator{\poly}{poly}

\newcommand{\had}{\textsc{Had}}

\newcommand{\eps}{\varepsilon}

\newcommand{\mN}{\mathcal{N}}

\DeclareMathOperator{\polylog}{polylog}

\DeclareMathOperator{\tr}{tr}

\newcommand{\eq}{\mathrm{eq}}
\newcommand{\var}{\mathrm{var}}

\newcommand{\game}{\mathfrak{G}}

\newcommand{\gamestyle}[1]{\ensuremath{\textsc{#1}}\xspace}

\newcommand{\pauli}{\gamestyle{P}}

\newcommand{\labelstyle}[1]{\ensuremath{\textsc{#1}}\xspace}

\newcommand{\alice}{\labelstyle{A}}
\newcommand{\bob}{\labelstyle{B}}

\newcommand{\class}[1]{\ensuremath{\mathsf{#1}}\xspace}
\newcommand{\QMIP}{\class{QMIP}} %
\WithSuffix\newcommand\QMIP*{\ensuremath{\class{QMIP}^*}} %
\newcommand{\MIP}{\class{MIP}} %
\newcommand{\RE}{\class{RE}} %

\WithSuffix\newcommand\MIP*{\ensuremath{\class{MIP}^*}} %

\renewcommand{\cal}[1]{\mathcal{#1}}

\mathchardef\mhyphen="2D

\newenvironment{gamespec}{
  \begin{mdframed}[style=figstyle]}{
  \end{mdframed}}

\newcommand{\tnote}[1]{}

\dajAUTHORdetails{%
  title = {Efficiently Stable Presentations from Error-Correcting Codes}, 
  author = {Michael Chapman, Thomas Vidick and Henry Yuen},
  plaintextauthor = {Michael Chapman, Thomas Vidick and Henry Yuen},
  plaintexttitle = {Efficiently Stable Presentations from Error-Correcting Codes}, 
  copyrightauthor = {M. Chapman, T. Vidick and H. Yuen},
  keywords = {group stability, error-correcting codes},
}   

\dajEDITORdetails{%
   year={2026},
   number={1},
   received={9 November 2023},   
   revised={21 October 2025},    
   published={16 July 2026},  
   doi={10.19086/da.154623},       
}   

\begin{document}

\begin{frontmatter}[classification=text]

\title{Efficiently Stable Presentations from Error-Correcting Codes} 

\author[mc]{Michael Chapman\thanks{Supported by the National Science Foundation under Grant No. DMS-2424441 and Simons Foundation grant N. 965535}}
\author[tv]{Thomas Vidick\thanks{Supported by the Swiss State Secretariat for Education, Research and Innovation (SERI). Parts of this work was completed while at the Weizmann Institute, supported by a research grant from the Center for New Scientists at the Weizmann Institute of Science, AFOSR Grant No. FA9550-22-1-0391, and ERC Consolidator Grant VerNisQDevS (101086733).}}
\author[hy]{Henry Yuen\thanks{Supported by AFOSR award FA9550-23-1-0363, CCF-2144219, and the Sloan Foundation.}}

\begin{abstract}
We introduce a notion of \emph{efficient stability} for finite presentations of groups. Informally, a finite presentation using generators $S$ and relations $R$ is \emph{stable} if any map from $S$ to unitaries 
that approximately satisfies the relations (in the tracial norm) is close to the restriction of a representation of $G$ to the subset $S$. This notion and variants thereof have been extensively studied in recent years, in part motivated by connections to property testing in computer science. The novelty in our work is the focus on \emph{efficiency}, which, informally, places an onus on small presentations --- in the sense of encoding length.  
The goal in this setup is to achieve non-trivial tradeoffs between the presentation length and its modulus of stability.

With this goal in mind we analyze various natural examples of presentations. We provide a general method for constructing presentations of $\Z_2^k$ from linear error-correcting codes. We observe that the resulting presentation has a weak form of stability exactly when the code is  \emph{testable}. This raises the question of whether testable codes give rise to genuinely stable presentations using this method.  While we cannot show that this is the case in general, we leverage recent results in the study of non-local games in quantum information theory (Ji et al., Discrete Analysis 2021) to show that a specific instantiation of our construction, based on the Reed-Muller family of codes, leads to a stable presentation of $\Z_2^k$ of size $\poly\log(k)$ only. As an application, we combine this result with recent work of de la Salle (arXiv:2204.07084) to re-derive  the quantum low-degree test of Natarajan and Vidick (IEEE FOCS'18), which is a key building block in the recent refutation of Connes' Embedding Problem via complexity theory (Ji et al., arXiv:2001.04383). 
\end{abstract}
\end{frontmatter}

\section{Introduction}

\paragraph{Motivation.}
A linear error-correcting code $\code$ is a $k$-dimensional subspace of the vector space $\F^n$ over a finite field $\F$ that has certain combinatorial properties. The foremost of these is the \emph{minimal distance} $d$, which is defined as the smallest Hamming weight (number of nonzero coordinates) $|c|$ of a nonzero vector  $c\in\code$. In general one would like to design families of codes of increasing length $n$, such that both $k$ and $d$ are bounded below by a positive linear function of $n$. Such codes are referred to as ``good'' codes.

A finer property which concerns us here is the \emph{soundness} of the code, a parameter that is connected to the notion of \emph{testability}. A code can be (non-uniquely) specified through a \emph{parity-check matrix} $h\in\F^{m\times n}$ as $\code = \ker h$. The rows of $h$ are thought of as constraints (``parity checks'') that specify $\code$ as a subspace of $\F^n$. A code is called testable \emph{with soundness $\rho$} if for every $x\in \F^n$, $\frac{1}{m}|hx|\geq \rho \, \frac{1}{n}\, d(x,\code)$, where $d(x,\code)$ denotes the minimum of $|x-c|$ over $c\in \code$.\footnote{In the literature, the notion of \emph{local} testability is emphasized, where in addition the parity-check matrix is required to have rows of low Hamming weight. This requirement is less important for us, and so we de-emphasize it.}  Ideally one would like to design families of good codes such that in addition $\rho$ is bounded below by a constant independent of $n$.\footnote{Note that in principle a code can have a small minimal distance $d$, and still be testable with soundness $\rho>0$. So a family of codes can be ``testable'' without being ``good.''} This is a challenging task, and the construction of families of good (locally) testable codes was only achieved very recently~\cite{LTC_DELLM,LTC_Panteleev_Kalachev}. 

The terminology ``testable'' comes from an interpretation of $\frac{1}{m}|hx|$ as the probability of rejection of a natural ``tester'' for $\code$, i.e.\ an algorithm that on input $x$ checks a randomly chosen row $h_i$ of $h$ and accepts if and only if $h_i\cdot x = \sum_j h_{ij} x_j =0$.  Thus a code is called testable if words that are far from the code have a high probability of being rejected according to this tester. This notion (when accompanied by the locality constraint) plays a central role in applications of codes to complexity theory~\cite{babai1991non,PCP_thm}, and continues to be actively studied. See e.g.~\cite{LTC_DELLM,LTC_Panteleev_Kalachev} for a recent breakthrough on the topic. 

We make an observation that connects the study of testable codes to questions of stability in group theory and motivates our work. Let $\code=\ker h$ be a linear code as above, and suppose for simplicity that $\F=\F_2$ is the binary field. Consider the finitely presented group 
\begin{align}
 G(h) \,=\, \langle S:R\rangle \,=\, \big\langle &x_1,\ldots,x_n \,:\quad  x_j^2=e\quad \forall 1\leq j\leq n\;,\notag\\
& \quad \prod_{1\leq j \leq n} x_j^{h_{ij}} = e\;,\quad [x_{j_1},x_{j_2}]^{h_{ij_1} h_{ij_2}}=e\quad \forall 1\leq i\leq m\,\forall 1\leq j_1,j_2\leq n\,\big\rangle\;.\label{eq:gh-intro}
\end{align}
Here a commutation relation between two generators  $[x_{j_1},x_{j_2}]=e$ is imposed only when needed for the relations $\prod_{1\leq j \leq n} x_j^{h_{ij}} = e$ to make sense, i.e.\ two generators are required to commute only if they both take part in the same equation, which is the case if and only if $h_{ij_1}h_{ij_2}=1$ for some $i$. While one could add all pairwise commutation relations, forcing $G(h)$ to be abelian --- and indeed we will do this in some cases later --- we choose the specific presentation~\eqref{eq:gh-intro} for consistency with the literature on non-local games, and particularly so-called \emph{linear constraint system games}, which we review in more detail below. 

We observe that $1$-dimensional representations, i.e.\ maps from $S=\{x_1,\ldots,x_n\}$ to $\{-1,1\}$ that satisfy all relations $R$, are in one-to-one correspondence with elements of $\code$. Moreover, \emph{approximate} $1$-dimensional representations, i.e.\ maps from  $S=\{x_1,\ldots,x_n\}$ to $\{-1,1\}$ that satisfy a fraction $1-\eps$ of the matrix relations for small $\eps$, can be identified with words $x\in\F_2^n$ such that $\frac{1}{m}|hx|\leq \eps$.\footnote{We conditioned only on the matrix relations since the commutation and involution relations are automatically satisfied by our choice of range $\{\pm1\}$. In principle we could allow the map to range over $U(\C)$ instead of $\{-1,1\}$, and in general we will allow this. But, for the purposes of this introduction, it is simpler to restrict to the $\{-1,1\}$-valued case.} In particular, we notice that $\code$ is testable with soundness $\rho$ if and only if $\eps$-approximate $1$-dimensional representations of $G(h)$ are $\nicefrac{C\eps}{\rho}$-close to genuine $1$-dimensional representations, where $C$ depends on the distribution we choose over the relations. E.g., if we check an involution relation with probability $\nicefrac{1}{3}$, a commutation relation with probability $\nicefrac{1}{3}$, and a matrix relation with probability $\nicefrac{1}{3}$, then $C=3$.

\paragraph{Group stability.}
This observation immediately raises many questions. The problem of relating approximate representations of a group to exact representations of it is termed \emph{stability} in group theory, and has a long history. There are of course many flavors of the problem, depending on how one defines closeness (should it be the operator norm or the Hilbert-Schmidt norm? Should closeness hold for every relation, or every pair of group elements, or is it sufficient that it holds on average? Etc.) We will review some relevant results in this area below. For now, we mention Voiculescu's famous counter-example~\cite{voiculescu1983asymptotically} about approximately commuting unitaries, which was motivated by a question of Halmos on pairs of approximately commuting Hermitian operators~\cite{halmos1976some}. Using the terminology of stability, Voiculescu showed that the presentation $\Z^2 = \langle x,y:[x,y]=e\rangle$ is \emph{not} stable with respect to the operator norm. However, much more recently Glebsky~\cite{glebsky2010almost} showed that the same presentation \emph{is} stable with respect to the normalized Hilbert-Schmidt norm. 

This example and many others show that the notion of stability is, in general, highly sensitive to the notion of closeness considered. Returning to our main concern, so far we have argued that 
stability of approximate \emph{$1$-dimensional} representations of the presentation $G(h)$ is connected to local testability of the code $\ker h$. In the case of $1$-dimensional representations of course the choice of norm does not matter; however, the choice of measuring the error on average over relations, as opposed to e.g.\ taking the maximum, is one in which we depart from most of the literature. We will motivate this choice further below; but before we can continue we must pause to introduce the key definitions that our work builds on. For the purposes of the introduction we focus the discussion on finite-dimensional representations; in the main paper we handle the general case of representations in a tracial von Neumann algebra. Let $\mU(\C^d)$ denote the unitary operators on $\C^d$, and for a set $S$ let $\mF(S)$ denote the free group generated by the elements of $S$. For $X\in \C^{d\times d}$ let $\|X\|_{hs}^2 = \frac{1}{d}\Tr(X^* X)$ denote the (normalized) Hilbert-Schmidt norm.\footnote{This norm is also commonly called the normalized \emph{Frobenius} norm. Following \cite{dechiffre_glebsky_lubotzky_thom_2020}, it became common in stability theory to call the normalized version Hilbert--Schmidt and the un-normalized version Frobenius. In any case, in the main part of this paper, we relate to it as the tracial norm because of the von Neumann algebraic framework we use.}

\begin{definition}[Almost homomorphism]\label{def:approx-hom-intro}
Let $G = \langle S:R\rangle $ be a finitely presented group and $\mu_R$ a distribution on $R$. An $(\eps,\mu_R)$-almost homomorphism of $G$ is a homomorphism $\phi:\mF(S)\to\mU(\C^d)$ for some $d\geq 1$ such that
\[ \Es{r\sim \mu_R} \big\|  \phi(r) - \Id \big \|_{hs}^2 \,\leq\, \eps\;.\]
\end{definition}

As already mentioned this definition makes two important choices: firstly, to measure closeness in the Hilbert-Schmidt norm, and secondly, to measure it on average over the choice of a relation. Next we give our definition for a finitely presented group to be stable; see Definition~\ref{def:eff-stab} for the general setting. 

\begin{definition}[Stability]\label{def:eff-stab-intro}
Let $G = \langle S:R\rangle $ be a finitely presented group, $\mu_S$ a distribution on $S$ and $\mu_R$ a distribution on $R$. For $\delta:[0,1]\to[0,1]$ such that $\lim_{t\to 0}\delta(t)=0$ and an integer $d\geq 1$ we say that the presentation $G=\langle S:R\rangle$ is $(\delta,\mu_S,\mu_R,d)$-stable if for every $(\eps,\mu_R)$-almost homomorphism $\phi: \mF(S) \to \mU(\C^d)$ there is a unitary representation $\psi: G \to \mU(\C^d)$ such that
\[ \Es{s\sim \mu_S} \big\|  \phi(s) - \psi(s) \big \|_{hs}^2 \,\leq\, \delta(\eps)\;.\]
We refer to any function $\delta$ satisfying the above as a ``modulus of stability'' of the presentation.
\end{definition}

\begin{remark}[The $L^\infty$ analogue]\label{rem:L^infty_analogue_defn}
    As mentioned  before, it is common to call a homomorphism $\phi\colon \mF(S)\to \mU(\complex^d)$ an $\eps$-almost homomorphism of $G=\langle S\colon R\rangle,$  if $\max_{r\in R}\Vert \phi(r)-\Id\Vert_{hs}^2\leq \eps$. Furthermore, the distance between two homomorphisms $\phi,\psi\colon \mF(S)\to \mU(\complex^d)$ is ususally taken to be $\max_{s\in S}\Vert \phi(s)-\psi(s)\Vert_{hs}^2$. The notion of stability induced by these definitions of `almost' and `close' is more commonly used. Note that when studying a fixed finitely presented group, and without caring about the exact modulus of stability, there is no difference between the two definitions.  But, when one cares about the modulus of stability, which is the case  when viewing stability as a property testing problem, our framework is the more natural one. 
\end{remark}

In this paper we also consider a version of \emph{flexible} stability, i.e.\ the representation $\psi$ is allowed to range in $U_{d'}(\C)$ for some $d'\neq d$; see Definition~\ref{def:eff-stab} for the general definition. This requires a more careful definition of closeness; for now we restrict our attention to the simpler definition. 

We can now ask the question: is $G(h)$, the group presentation defined in~\eqref{eq:gh-intro}, stable according to Definition~\ref{def:eff-stab-intro}? 
One can verify that with $\mu_R$ which chooses with probability $\nicefrac{1}{3}$ whether to check an involution, a commutation or a matrix row,  and $\mu_S$ the uniform distribution, $G(h)$ is $(\delta,\mu_S,\mu_R,1)$-stable for $\delta=3\eps/\rho$ if and only if $\code = \ker h$ is testable with soundness $\rho>0$. But what about higher-dimensional approximate representations? Are these also stable, or does one need to make further requirements on $\code$ beyond local testability? We do not yet have a comprehensive answer to these questions. However, the answer cannot be straightforward: even deducing basic properties about the group $G(h)$ given its presentation -- let alone its stability properties -- appears to be a challenging task. For example, there are examples of parity check matrices $h$ for which the group $G(h)$ is non-abelian or even non-amenable; see Remark~\ref{rk:non-abelian}. In fact, every finitely generated group can be embedded into $G(h)$ for some $h$ (see \cite{slofstra2019set}).

\paragraph{Efficient stability.}
Faced with the apparent difficulty of studying the general question, it is time to refine our focus and formulate the question which we \emph{do} address. To start, let us explicitly note that 
while stability has previously (for the most part) been studied as a question about a group, our formulation makes it a question about a \emph{presentation} of the group. In particular it is known~\cite{gowers2017inverse,de2019operator} that for finite groups $G$ (which are the only groups we consider in this paper) the multiplication table presentation, which has $|G|$ generators, one  for every group element, and $|G|^2$ relations, one for every pairwise product, is $\delta(\eps)=C\eps$-flexibly stable\footnote{Here we assume that $\mu_R,\mu_S$ are uniform over the relations and generators of the multiplication table presentation, respectively. Furthermore, the results of~\cite{gowers2017inverse,de2019operator} apply to a notion of \emph{flexible} stability, where the nearby exact representation may act on a space of larger, but not too much larger, dimension. See  Definition \ref{def:eff-stab}.} for some constant $C$ that is independent of the group. This holds even for approximate representations in arbitrary tracial von Neumann algebras. 

To simplify the problem let us consider the following presentation for an obviously finite and abelian group
\begin{align}
  \widetilde{G(h)} \,=\, \langle S:R\rangle \,=\, \big\langle &x_1,\ldots,x_n \,:\quad  x_j^2=e\quad \forall 1\leq j\leq n\;,\notag\\
 & \quad \prod_{1\leq j \leq n} x_j^{h_{ij}} = e\;,\quad [x_{j_1},x_{j_2}]=e\quad \forall 1\leq i\leq m\,\forall 1\leq j_1,j_2\leq n\,\big\rangle\;.\label{eq:wtgh-intro}
 \end{align}
This is the same as $G(h)$, except that all pairwise commutations have been added. As we show formally later (see Lemma~\ref{lem:com-code}), it is not hard to check that $\widetilde{G(h)}$ is a presentation of the group $\Z_2^k$, for $k=\dim\ker h$. Our most important results pertain to the stability of the presentation $\widetilde{G(h)}$; although we will later (Section~\ref{sec:braiding}) consider some non-abelian extensions of it. 

\newcommand{\wtG}{\widetilde{G(h)}}

Importantly for us, the results of~\cite{gowers2017inverse,de2019operator} do not imply the same quantitative stability bounds for $\wtG$ with respect to its defining presentation~\eqref{eq:wtgh-intro}. The reason to prefer the presentation~\eqref{eq:wtgh-intro} as opposed to the multiplication table presentation of $\wtG$ is that~\eqref{eq:wtgh-intro} is much more succinct: if $n,k$ and $m$ are linearly related then it has $\textrm{poly}(k)$ generators and relations, as opposed to $2^k$ and $2^{2k}$ respectively. Such a gain is essential when one recalls the interpretation of the presentation $\wtG$ as a ``tester'' for $G$ --- the size of the presentation is then directly related to the amount of \emph{randomness} required by the tester (to sample a random relation); and in computer science applications randomness is seen as an essential resource (we discuss this more below, in the context of quantum computing). 


Can such ``efficient'' presentations of $G(h)$ still be stable? The naive approach, of extending an $\eps$-approximate homomorphism of $G(h)=\langle S:R\rangle$ into an $\eps'$-approximate homomorphism of the multiplication table presentation of $\Z_2^k$, leads to $\eps'=\Omega(k^2\eps)$ and hence a logarithmic dependence of the modulus of stability on the group size. Is it possible to do better? Glebsky's result for the case of $\Z^2$, which is an infinite group, suggests that in some cases the modulus of stability can be independent of the group size. It is therefore not clear what if any of the known group parameters should play a role in it in general.  

The main result of this paper is to exhibit presentations of $\Z_2^k$ (and of slightly more complex $2$-groups built on them, e.g. the \emph{Pauli}\footnote{This group is commonly referred to as the (multi-dimensional) \emph{Heisenberg} group over $\field$, or  the \emph{Weyl--Heisenberg} group.} group ubiquitous in quantum information theory) that are \emph{efficiently stable}: the size of the presentation is quasi-polynomial in $k$ (as opposed to exponential), yet the modulus of stability only depends poly-logarithmically on $k$. These presentations are constructed from specific error-correcting codes, namely the Reed--Muller polynomial codes, which are known to have good local testability properties~\cite{babai1991non}. Our main result on group stability can be stated as follows (see Theorem~\ref{thm:z2-stab} for the precise statement).

\begin{theorem*}[Main, informal]\label{thm:main-inf}
For every integer $k\geq 1$ there is a presentation $\Z_2^k = \langle S_k:R_k\rangle$ such that $|S_k|,|R_k| = 2^{\poly\log(k)}$ and furthermore this presentation is $(\delta,\mu_{S_k},\mu_{R_k},d)$-stable for all $d \in \N$, where $\delta(\eps)=\poly(\log k,\eps)$,\footnote{We use the notation $f(a,b,c,\ldots)=\poly(a,b,c,\ldots)$ to mean that there exists constants $C,c_1,c_2,c_3,\ldots$, all positive, such that $|f(a,b,c,\ldots)|\leq C a^{c_1} b^{c_2} c^{c_3}\cdots$ for all $a,b,c,\ldots$ in their range.\label{ft:poly}} $\mu_{S_k}$ is uniform over the generators and $\mu_{R_k}$ is some distribution over the relations. 
\end{theorem*}

Most of the technical legwork required to prove the theorem is due to prior work in the study of nonlocal games in quantum computing, and in particular~\cite{ji2020quantum} (we explain this connection below). Our contribution in the present work is to make an explicit connection with stability and show how the former results can be ``imported'' to obtain new stability results such as the one stated in our main theorem above. 

We do not know of any other presentation of $\Z_2^k$, arguably one of the simplest groups one could think of, that is stable with similar parameters as the ones stated in the theorem. It would be very interesting to discover different such presentations, built from testable error-correcting codes or not, for this group or others. 

\paragraph{Nonlocal games in quantum information theory.}
To motivate our focus on \emph{efficient} stability we now sketch a connection between our results and problems in quantum complexity theory and in particular the theory of nonlocal games that motivate us. As a result we will recover a key technical result used in the proof of the complexity result $\sf MIP^* = \sf RE$~\cite{ji2020mip} and its corresponding resolution of the Connes' Embedding Problem.

 A \emph{nonlocal game} $\game$ is specified by the following data: finite question and answer sets $\mX$ and $\mA$ respectively, a distribution $\mu$ on $\mX\times \mX$, and a decision predicate $D:\mX\times \mX\times \mA\times \mA\to \{0,1\}$ (see Section~\ref{sec:nl-games} for details). The interpretation of $\game=(\mX,\mu,\mA,D)$ as a game is as follows. A ``referee'' is imagined to sample a pair of ``questions'' $(x,y)\sim \mu$. Each question is sent to a different player, who is tasked with responding with an answer $a,b\in \mA$ respectively. Finally, the referee decides that the players win the game if and only if $D(x,y,a,b)=1$. Interestingly, the maximum success probability of the players in this game, where the probability is over the referee's choice of questions and any randomness in the player's strategy, and the maximum is taken over all allowed strategies, depends on whether one relies on ``classical'' or ``quantum'' interpretations of the game to determine appropriate mathematical formalization of the set of strategies that the players may employ. While a classical viewpoint naturally models a strategy as a pair of functions $f,g:\mX\to\mA$, one for each player, quantum mechanics invites one to consider a broader set of strategies in which an additional form of coordination between the players is allowed in the form of \emph{shared quantum entanglement}. Understanding when there is a gap between the resulting maxima, and how large this gap can be, is of great interest in the foundations of quantum mechanics. To study this question one is drawn to investigate the structure of optimal quantum strategies in a game, and how to design games that enforce a specific structure --- informally, forcing as much ``non-classicality'' in winning strategies as possible. Beyond their foundational appeal, the theory of nonlocal games has had a very large impact in quantum cryptography (such as the analysis of device-independent quantum key distribution protocols~\cite{vazirani2019fully,arnon2019simple}) and quantum complexity, in particular the theory of multiprover interactive proof systems~\cite{cleve2004consequences}.

For concreteness let us focus on a class of games called \emph{linear constraint system} (LCS) games. These games were introduced in~\cite{cleve2014characterization} and their study plays a central role in the celebrated result by Slofstra showing non-closure of the set of quantum correlations~\cite{slofstra2019set}. A linear constraint system game is parametrized by a matrix $h\in \F^{m\times n}$. In the game $\game_h$, the referee selects a pair $(i,j)\in \{1,\ldots,m\}\times\{1,\ldots,n\}$ by first sampling $i$ uniformly at random, and then sampling $j$ uniformly at random conditioned on $h_{ij}=1$. They send $i$ to the first player and $j$ to the second. The first player returns values in $\F$ for \emph{each} $j'$ such that $h_{ij'}\neq 0$, whereas the second players returns a single value in $\F$. The players win if the first player's answers satisfy the parity constraint, and the players' answers are consistent (the first player's answer associated with index $j$ matches the second player's answer). 

Now we see that to each matrix $h$ we have associated a group $G(h)$, and a game $\game_h$.\footnote{A similar correspondence holds between $\wtG$ and a natural associated game $\widetilde{\game}_h$; our main results apply to the latter, but the present discussion is more general and applies to both.} Moreover, and quite interestingly, there is a one-to-one correspondence between representations of $G(h)$ and \emph{perfect strategies} in $\game_h$, i.e.\ strategies that have success probability $1$ in the game.\footnote{This correspondence was established for finite-dimensional strategies, and finite-dimensional representations, in~\cite{cleve2014characterization}. Extensions to infinite-dimensional strategies and representations have appeared in~\cite{cleve2017perfect}. See also e.g.~\cite{kim2018synchronous} for generalizations to the broader class of \emph{synchronous games}.}
This correspondence enables one to ``embed'' group representations into quantum strategies, thereby forcing them to demonstrate a high level of complexity; this is the approach at the heart of~\cite{slofstra2019set}. Going further, for applications one is often required to understand not only optimal but also near-optimal strategies, whose success probability is e.g.\ $1-\eps$. The same correspondence associates to such strategies approximate homomorphisms of $G(h)$~\cite{slofstra2018entanglement}. Here once can see that measuring closeness on average over the choice of relation is a natural choice, which is all but forced by the definition of a game, where the questions are selected according to some pre-specified distribution. 

To summarize, approximate stability results for $G(h)$ enable one to obtain structural results about near-optimal strategies in $\game_h$ --- such a result is known as a \emph{rigidity} result in quantum information. A rigidity result about a game called the quantum low-degree test introduced in~\cite{natarajan2018low} (with a flawed analysis later corrected in~\cite{ji2020quantum,ji2022quantum}), is at the heart of the proof of $\sf MIP^* = \sf RE$. The analysis of this test requires an efficient stability result for the Pauli group. Informally, the reason that the stability result needs to be for an efficient presentation of the Pauli group, as opposed to e.g.\ the multiplication table presentation, is because to obtain the final result it is necessary that the ``complexity'' of the test, or game, is smaller than the ``complexity'' of the object, or group being tested. Of course here we are loose about what we mean by ``complexity,'' and refer to~\cite{ji2021mip,vidickmip} for high-level explanations. The quantum low-degree test is described and analyzed in Section~\ref{sec:pbt}. We end by mentioning an open question: if one was able to obtain $|S_k|,|R_k|=\poly(k)$ while also having $\delta(\eps)=\poly(\eps)$ in the informal theorem stated above, then this result would likely, through the connection we just described, have  consequences for the efficient verification of quantum computations in the framework of interactive proof systems---see e.g.~\cite{coladangelo2019verifier,natarajan2023bounding} for a sample of known results in this direction. 

\paragraph{Related works on stability.} The general question of group stability was first formulated by Ulam~\cite{ulam1960collection}, and later studied  by Kazhdan for the case of the operator norm~\cite{kazhdan1982e}. See the introduction of \cite{ioana2020stability} for a thorough history of these problems through the lens of approximate commutation (cf.  \cite{von1942approximative,voiculescu1983asymptotically,glebsky2010almost}). A major contribution was done by \cite{hadwin2018stability}, in which they characterized the stable amenable groups according to properties of their space of characters. 

Stability of finite groups, with respect to the multiplication table representation, is shown in~\cite{gowers2017inverse}. 
One can also consider representations in permutations, see~\cite{glebsky2009almost,becker2022stability}.  Specifically, the analogous problem of efficient stability in permutations is still open. See the open problems section of \cite{CL_part2}. Both stability in permutations and in unitaries equipped with the Hilbert--Schmidt metric are closely related to the notions of sofic and hyperlinear groups, cf. \cite{glebsky2009almost,becker2020group}.

\paragraph{Open questions.}
We leave many questions open. A natural direction is to determine if there are simple sufficient conditions on the matrix $h$ that guarantee that the presentation $G(h)$ (or $\widetilde{G(h)}$) is stable. As we discussed, the condition that $\ker h$ is a testable code is equivalent to stability for $1$-dimensional permutations. It thus seems likely that testability is not a sufficient condition in general; can testability be determined by a combinatorial parameter of $h$? More generally, finding other examples of efficient stability seems of intrinsic interest, besides potential applications to property testing and the construction of nonlocal games. In a different direction, as mentioned in the previous paragraph an analogue of Theorem~\ref{thm:main-inf} for the case of representations in permutations could have important implications towards showing the existence of non-sofic groups~\cite{CL_part1,CL_part2,BCLV_subgroup_tests}.

\paragraph{Outline.}
We start in Section~\ref{sec:efficient} by giving a precise definition of stability that we work with, and give examples to motivate and illustrate the definition. In Section~\ref{sec:pres-codes} we give a general method for constructing presentations from codes, and apply the method to the spacial case of the Reed-Muller code. This leads in Section~\ref{sec:eff-z2k} to the statement and proof of our main result, an efficiently stable presentation for $\Z_2^k$. Finally in Section~\ref{sec:quantum} we elaborate on the connection with the theory of nonlocal games and detail our applications to this area.

\paragraph{Acknowledgments.} 
We would like to thank John Wright for his remarks on an early draft of this paper. \tnote{added:}We are grateful to an anonymous referee for many constructive comments, including the inclusion of Lemma~\ref{lem:stab-group} and Lemma~\ref{lem:l1}. 
MC acknowledges with gratitude the Simons Society of Fellows and is supported by a grant from the Simons Foundation (N. 965535). 
TV is supported by a research grant from the Center for New Scientists at the Weizmann Institute of Science, a Simons Investigator award, AFOSR Grant No. FA9550-22-1-0391, and ERC Consolidator Grant VerNisQDevS (101086733). HY is supported by AFOSR award FA9550-21-1-0040, NSF CAREER award CCF-2144219, and the Sloan Foundation.

\section{Efficient stability}
\label{sec:efficient}

In this section we give definitions associated with the notion of ``efficient stability'' used in the paper. We reformulate some previously known results in this framework, and give examples that will be used later on.

\subsection{Algebra background and notation}

  Here we call \emph{tracial von Neumann algebra} a pair $(\mM,\tau)$ of a von Neumann algebra $\mM$ together with a normal faithful tracial state $\tau$ on $\mM$, which we often refer to as the \emph{trace}. The main example of interest is $\mM=M_n(\C)$, the algebra of $n\times n$ complex matrices, with $\tau$ the dimension-normalized trace, which we denote $\tr(M)=\frac{1}{n}\Tr(M)$. 	We write $\|x\|_\tau=\tau(x^*x)^{1/2}$ to denote the $2$-norm on $\mM$ with respect to $\tau$ --- which agrees in the case of $M_n(\complex)$ with the Hilbert--Schmidt norm $\|x\|_{hs}$ discussed in the introduction.
	
	Let $B(\ell_2)$ be the von Neumann algebra of bounded operators on $\ell_2$, the Hilbert space of square-convergent sequences in $\C^\Z$ equipped with the usual Euclidean norm (for which we let $(e_i)_{i \in \Z}$ denote the standard basis). We denote $\mM_\infty = \mM \overline{\otimes} B(\ell_2)$, where the overline denotes closure for the weak operator topology. $\mM_\infty$ is a von Neumann algebra equipped with the (infinite) trace $\tau_\infty = \tau \otimes \Tr$, with $\Tr(X)=\sum_{i\in \Z} e_i^T X e_i$ the trace on $B(\ell_2)$. We generally identify $\mM$ with the ``corner'' $\mM\otimes \Id_{1}\subset \mM_\infty$, where $\Id_1$ is the projection on the $1^{\rm st}$ coordinate in $\complex^\Z$.

\subsection{Efficiently stable presentations}

Suppose we are given a finite presentation of a (possibly infinite) group $G$ using generators $S$ and relations $R$. Informally, we say that the presentation is \emph{stable} if any map from $S$ to unitaries that approximately respects the relations $R$ is close, in an appropriate sense, to a representation of $G$. Furthermore, we will say that the presentation is \emph{efficient} if it is stable and provides a good trade-off between its size (the number of relations used and their length) and how the closeness to a representation depends on the error in satisfying the relations $R$. All the notions referred to informally in the preceding sentences --- ``approximately,'' ``close,'' good trade-off,'' etc., can be formalized in a variety of ways, leading to generally incomparable definitions. Here we present the formalization that is most natural to us, and is motivated by applications to quantum information and complexity.

Given a set $S$, we let $\mF(S)$ denote the free group generated by $S$. We identify functions from $S$ to $H$, where $H$ is any group, with homomorphisms from $\mF(S)$ to $H$. If $R$ is a subset of $\mF(S)$ then the quotient of $\mF(S)$ by the normal subgroup generated by $R$ is denoted $\langle S:R\rangle$. 

We start with the notion of an almost-homomorphism, which formalizes what it means for a map defined on $S$ to approximately satisfy the relations $R$. The notion we give is a small variant of the notion of $\eps$-\emph{almost homomorphism} from a finitely presented group to a unital tracial $C^*$-algebra $\mA$ introduced in~\cite[Section 2]{hadwin2018stability}. 
We give a variant of their definition that quantifies the error in an average sense. Below, when $\mu$ is a distribution over a finite set $\mX$ and $f:\mX\to \R$ we write $\Es{x\sim \mu} f(x)$ for the expectation of $f$ under $\mu$. 

\begin{definition}[Almost homomorphism]\label{def:approx-hom}
Let $G = \langle S:R\rangle $ be a finitely presented group, $\mu$ a distribution on $R$, and $(\mM,\tau)$ a tracial von Neumann algebra. An $(\eps,\mu)$-almost homomorphism of $G$ on $(\mM,\tau)$ is a homomorphism $\phi:\mF(S)\to\mU(\mM)$ such that
\[ \Es{r\sim \mu} \big\|  \phi(r) - \Id \big \|_\tau^2 \,\leq\, \eps\;.\]
\end{definition}

We note that this notion depends on the presentation $\langle S:R\rangle$ of $G$, not only on the group itself. 
When the distribution $\mu$ is uniform over the set $R$, we simply write $\eps$-homomorphism. The definition is consistent with the usual notion of a homomorphism that factors through $G$, which is recovered when $\eps=0$ as long as $\mu$ is fully supported. 

A stability result is a statement that $\eps$-homomorphisms are close to homomorphisms. To measure the distance between homomorphisms into different algebras we make the following definition.

\begin{definition}[Closeness for unitaries]\label{def:close}
Let $\{U_i\}\subseteq \mM$ and $\{V_i\}\subseteq \mN$ be two families of unitaries on  tracial algebras $(\mM,\tau^\mM)$ and $(\mN,\tau^\mN)$ respectively, indexed by the same set $\mI$. For $\delta\geq0$ and $\mu$ a measure on $\mI$ we say that $\{U_i\}$ and $\{V_i\}$ are $(\delta,\mu)$-close if there exists a projection $P\in\mM_\infty$ of finite trace such that $\mN=P\mM_\infty P$ and $\tau^\mN=\tau_\infty/\tau_\infty(P)$, and a partial isometry $w\in P \mM_\infty \Id_\mM$ such that 
\[ \Es{i\sim\mu} \big\| U_i - w^* V_i w \big\|_{\tau}^2 \,\leq\,\delta,\ \footnote{We denote the norm by $\|\cdot\|_{\tau}$ and not $\|\cdot\|_{\tau^\mM}$ for clarity of the notations. The von Neumann algebra in which we take the norm should be understood from context, and will usually be $\mM$. }\]
and 
\[\max\big\{ \tau^\mM(\Id_\mM-w^*w)\,,\; \tau^\mN(P-ww^*)\big\} \,\leq\, \delta\;.\]
If $\phi:\mI\to \mU(\mM)$ and $\psi:\mI\to \mU(\mN)$ then we say that $\phi$ and $\psi$ are $(\delta,\mu)$ close if the families $\{\phi(i)\}$ and $\{\psi(i)\}$ are. 
If the measure $\mu$ is omitted then it is understood to be the uniform measure on $\mI$.
\end{definition}

We now give our definition of stability.

\begin{definition}[Stability]\label{def:eff-stab}
Let $G = \langle S:R\rangle $ be a finitely presented group. Let $\mC$ be a class of tracial von Neumann algebras. Let $\mu_S$ be a distribution on $S$ and $\mu_R$ a distribution on $R$. Let $\delta:[0,1]\to[0,1]$ be a function satisfying $\lim_{t\to 0}\delta(t)=0$. The presentation $G=\langle S:R\rangle$ is $(\delta,\mu_S,\mu_R,\mC)$-stable if for every $(\cM,\tau)$ in $\mC$, every $(\eps,\mu_R)$-almost homomorphism of $G$ is $(\delta(\eps),\mu_S)$-close to a unitary representation of $G$ on some $(\mN,\tau^\mN)\in \mC$. We refer to the function $\delta$ as the \emph{modulus of stability} of the presentation.\footnote{Every function that satisfies this condition is a modulus of stability for $G=\langle S\colon R\rangle$, but we occasionally refer that way to the best possible $\delta$ in the definition.} 
\end{definition}

\begin{remark}
    Definition \ref{def:eff-stab} is often referred to in the literature as ``flexible pointwise Hilbert--Schmidt stability (in the class $\mC$)''. There are other versions of stability, such as: the non-flexible one (cf.\ \cite{becker2019stability, hadwin2018stability}), in which one does not allow to compare representations of different dimensions; uniform stability, in which almost representations are with respect to  the (often) infinite multiplication table presentation (cf.\ \cite{kazhdan1982e,becker2022stability}); stability in permutations, where instead of studying approximate representations one studies approximate actions (cf.\ \cite{becker2019stability,becker2020group, ioana2020stability}); and so on. In the cited papers above there are thorough literature surveys that we avoid herein.
\end{remark}

While the quantitative aspects of this definition of stability depend on the choice of presentation, qualitatively a finitely presented group $G$ is stable with respect to some presentation if and only if it is stable with respect to any presentation. This is shown in Lemma~\ref{lem:stab-group} below, which includes some rough quantitative estimates. We start with a simple but useful calculation.
\tnote{Added the lemma. The proof is pretty annoying. Michael, did you say there are references for it?}

For a set of generators $S$ and an element $r\in\mF(S)$, we let $|r|$ denote the length of the word $r\in \mF(S)$ written in the basis $S\cup S^{-1}$.\footnote{Another popular notion of length is the bit length of the encoding of $\langle S\colon R\rangle$, where exponents in $r\in R$ can be written in binary. Up to a factor of $\log |S|$, our notion is stricter than that. See \cite{babai1997short}.}\tnote{previously, we had just $S$. But we also need to allow inverses, right?}

\begin{lemma}\label{lem:g-phi}
Let $G=\langle S:R\rangle$ and $\mu_S$ a distribution on $S$.  Let $\phi:\mathcal{F}(S)\to \mathcal{U}(\cM)$ and $g:G\to \mathcal{U}(\mN)$ a representation such that $g$ is $(\delta,\mu_S)$ close to $\phi$. Then for any $r\in\mathcal{F}(S)$, it holds that 
\[ \|\phi(r)-w^*g(r)w\|_\tau \,\leq\, (|r|+1)\sqrt{\delta} + \sum_{i=1}^{|r|}  \| w^* g(r_i)w -\phi(r_i)\|_\tau\;,\]
where $r_i$ denotes the $i$-the element of $r$ (which is an element of $S\cup S^{-1}$).
\end{lemma}

\begin{proof}
  Let $r$ be a word in the elements of $S$ and their inverses. For $0\leq i \leq |r|$ we write $r_{\leq i}$ for the length-$i$ prefix of $r$ (which is $1_G$ if $i=0$), and $r_{> i}$ for its length-$(|r|-i)$ suffix (which is $1_G$ if $i=|r|$). 
  Using the triangle inequality, 
  \begin{align*}  
    \|\phi(r)-w^*g(r)w\|_\tau &\leq \sum_{i=1}^{|r|} \| w^*{g}(r_{\leq i}) w{\phi}(r_{>i})-w^* {g}(r_{<i})w\phi(r_{\geq i}))\|_\tau  + \|(w^*w-I_\cM)\phi(r)\|_\tau \\
    &\leq \sum_{i=1}^{|r|} \Big(\| w^*g(r_{<i})w(w^* g(r_i)w -\phi(r_i))\phi(r_{>i})\|_\tau \\
    &\qquad + \|w^* g(r_{<i}) (ww^*-P) g(r_i)w \phi(r_{>i})\|_{\tau^\mN} \Big)+ \|w^*w-I_\cM\|_\tau \\ 
    &\leq \sum_{i=1}^{|r|}\| w^* g(r_i)w -\phi(r_i)\|_\tau + |r|\|ww^*-P\|_{\tau^{\mN}} + \|w^*w-I_\mM\|_\tau\;.
    \end{align*}
    Here, for the last inequality we used that $g(r_i)$, $\phi(r_i)$ are unitaries, and $w$ an isometry, so either of them contracts the norm.  
\end{proof}

\begin{lemma}\label{lem:stab-group}
Suppose that $G=\langle S:R\rangle$ is $(\delta,\mu_S,\mu_R,\mC)$-stable as in Definition~\ref{def:eff-stab}. Let $G = \langle S':R'\rangle$ be another presentation. Then for any presentation $G=\langle S':R'\rangle$, and any distribution $\mu_{R'}$ on $R'$ that has full support, $G=\langle S':R'\rangle$ is $(\delta',\mu_{S'},\mu_{R'},\mC)$-stable for any distribution $\mu_{S'}$ and a modulus $\delta'$ that satisfies $\delta'(\eps)\leq C\delta(C'\eps)+C''\eps$ for some constants $C,C',C''$ depending on the presentations.  
\end{lemma}

\begin{proof}
Let $(\cM,\tau)\in \mC$ and $\phi'$ an $(\eps,\mu_{R'})$-almost homomorphism of $\langle S':R'\rangle$. Let $S=\{s_1,\ldots,s_t\}$, where $t=|S|$. For each $i\in\{1,\ldots,t\}$ there is a word $w_i\in \mathcal{F}(S')$ in the generators $S'$ and their inverses such that $s_i=w_i$ in $G$. Define $\phi(s_i)=\phi'(w_i)$ and extend $\phi$ to $\mathcal{F}(S)$ in the obvious way.

Fix an arbitrary $r\in R$. Then $r$ can be expressed as $r=s_{i_1}^{\eps_1}\cdots s_{i_\ell}^{\eps_\ell}=w_{i_1}^{\eps_1}\cdots w_{i_\ell}^{\eps_\ell}$ in $G$, where $i_1,\ldots,i_\ell\in \{1,\ldots,t\}$ and $\eps_1,\ldots,\eps_\ell\in\{\pm 1\}$. Define $w(r) = w_{i_1}^{\eps_1}\cdots w_{i_\ell}^{\eps_\ell}$ as an element of $\mathcal{F}(S')$. Since $r=1$ in $G$, also $w(r)=1$ in $G$. Moreover, we have the estimate 
\[ |w(r)| \leq |w_{i_1}|+\cdots+|w_{i_\ell}|\leq \ell \max_{i\in \{1,\ldots,t\}}|w_i|\;.\]
Let $D'=\max_{i\in \{1,\ldots,t\}}|w_i|$. 

Let $\Delta'$ be the Dehn function of $\langle S':R'\rangle$. Then by definition, since $w(r)=1$ it can be written as a product 
\[ w(r)\,=\, u_1 (r'_{1})^{\eps'_1} u_1^* \cdots u_L (r'_{L})^{\eps'_L}u_L^* \]
of $L\leq \Delta'(|w(r)|) \leq \Delta'(\ell D')$ conjugates of the $r'\in R'$ and their inverses. Here, $u_i\in \mathcal{F}(S')$ and $r'_i\in R'$, for each $1\leq i \leq L$. 

Using the definition of $\phi$, followed by the triangle inequality, 
\begin{align*}
\|\phi(r)-I\|_\tau &= \|\phi'(w(r))-I\|_\tau\\
&= \|\phi'( u_1 (r'_{1})^{\eps'_1}u_1^* \cdots u_L (r'_{L})^{\eps'_L} u_L^* ) - I \|_\tau\\
&\leq  \sum_{j=1}^L \|\phi'( u_j (r'_{j})^{\eps'_j}u_j^* ) - I \|_\tau\\
&=  \sum_{j=1}^L \|\phi'( r'_{j}) - I \|_\tau\\
&\leq L \max_{r'\in R'} \|\phi'(r')-1\|_\tau\;.
\end{align*}
 From the above it follows that $\phi$ is an $(\eps',\mu_R)$-almost homomorphism of $\langle S:R\rangle$ where 
\begin{equation}\label{eq:sg-1a}
\eps' \leq  L^2 \eps (\min_{r'}\mu_{R'}(r'))^{-1} \;.
\end{equation}
Applying the assumption, $\phi$ is $(\delta(\eps'),\mu_S)$-close to a unitary representation $g$ of $G$ on some $(\mathcal{N}, \tau^\mN)\in \mC$. 
For $s'\in S'$, write $s' = w'$ where $w'\in\mathcal{F}(S)$ is a word in the generators from $S$ and their inverses. Further replacing each generator $s_i$ from $S$ by the word $w_i$ identified at the start of the proof gives $s'=w''$ where $w''\in \mathcal{F}(S')$ is a product of the $w_i$ and their inverses. 

Using the same reasoning as in the first part, it follows that 
\begin{equation}\label{eq:sg-1}
   \|\phi'(w''(s')^{-1})-I\|_\tau \leq (LL')^2 \eps (\min_{r'}\mu_{R'}(r'))^{-1}\;,
\end{equation}
where $L'=\Delta(DD')$ with $D=\max_{s'\in S'}|w'|$ and $\Delta$ the Dehn function of $\langle S:R\rangle$.
 But $\phi'(w''(s')^{-1}) = \phi'(w'')\phi'((s')^{-1})=\phi(w')\phi'(s')^{-1}$. Using Lemma~\ref{lem:g-phi} we deduce that 
\begin{align}
 \|\phi(w')-g(w')\| &\leq \sum_{i=1}^{|w'|}  \| w^* g(w'_i)w -\phi(w'_i)\|_\tau + (|w'|+1)\sqrt{\eps'}\notag\\
 &\leq \delta(\eps')^{1/2}  (\min_{s\in S} \mu(s))^{-1} + (|w'|+1)\sqrt{\eps'}\;.\label{eq:sg-2}
\end{align}
It follows from~\eqref{eq:sg-1},~\eqref{eq:sg-2} and~\eqref{eq:sg-1a} that $\phi'$ is $\delta'(\eps)$-close to $g$, where $\delta'(\eps)$ scales as claimed in the lemma. 
\end{proof}

We now give a simple lemma that shows that homomorphisms on $\mathcal{F}(S)$ that are $(\delta,\mu_S)$ close to representations are naturally almost homomorphisms, for any presentation $\langle S:R\rangle$ and a distribution on $R$ that is naturally defined from the distribution $\mu_S$.
\tnote{This is Lemma 1 by the reviewer.}

\begin{lemma}\label{lem:l1}
  Let $G=\langle S:R\rangle$ be a finite presentation, and $\mu_R$ a distribution on $R$. For a relation $r\in R$, let $|r|$ denote its length as a word in $\mathcal{F}(S)$. Let $\mu_S$ be the probability measure on $S$ defined by
  \begin{equation}\label{eq:mus}
     \mu_S \,=\, \frac{1}{\sum_{r\in R} \mu_R(r)|r|^2} \sum_{r\in R} |r|\mu_R(r) \sum_{i=1}^{|r|} \delta_{r_i}\;,
  \end{equation}
  where $\delta_{r_i}$ assigns probability $1$ to the $i$-th element of $r$, or its inverse, depending which lies in $S$. Then for any tracial von Neumann algebra $(\cM,\tau)$ and homomorphism $\phi:\mathcal{F}(S)\to \mathcal{U}(\cM)$, if $\phi$ is $(\delta,\mu_S)$ close to a representation of $G$ then $\phi$ is also an $(O(\delta),\mu_R)$-almost homomorphism. (Here the constant implicit in the $O(\cdot)$ depends linearly on $\Es{r\sim \mu_R}|r|^2$.)
  \end{lemma}
  
  \begin{proof}
    Let $\phi:\mathcal{F}(S)\to \mathcal{U}(\cM)$ be as in the lemma, and let $g$ be a $(\delta,\mu_S)$ close representation of $G$.
   Using the triangle inequality and $g(r)=1_\mN$ and then Lemma~\ref{lem:g-phi}, 
  \begin{align*}
  \Es{r\sim \mu_R}\|\phi(r)-\Id\|_\tau^2 &\leq \Es{r\sim \mu_R} 2\|\phi(r)-w^*g(r)w\|^2_\tau + 2\|w^*w-I_\mM\|^2_\tau\\
  &\leq \Es{r\sim \mu_R}\sum_{i=1}^{|r|}4|r|\| w^* g(r_i)w -\phi(r_i)\|_\tau^2 + O\Big(\Es{r\sim \mu_R}|r|^2\Big)\delta \\
  &=  \sum_{s\in S} \sum_{r\in R} \mu(R) \sum_{i=1}^{|r|} 1_{s\in \{r_i,r_i^{-1}\}} 4|r|\| w^* g(r_i)w -\phi(r_i)\|_\tau^2 +  O\Big(\Es{r\sim \mu_R}|r|^2\Big)\delta \\
  &\leq 4\Big(\sum_{r\in R} \mu_R(r)|r|^2\Big)\Es{s\sim \mu_S } \| w^* g(s)w -\phi(s)\|_\tau^2 +   O\Big(\Es{r\sim \mu_R}|r|^2\Big)\delta \\
  &=  O\Big(\sum_{r\in R} \mu_R(r)|r|^2\Big)\delta\;,
  \end{align*}
  where 
	for we used $\Es{r\sim \mu_R}|r|^2\geq 1$. 
  \end{proof}
  
  \tnote{changed remark below based on referee comment}
\begin{remark}\label{rk-mus}
  Throughout this text, whenever $\mu_R$ is specified by not $\mu_S$ we mean that $\mu_S$ is chosen as the induced distribution defined in~\eqref{eq:mus}.
\end{remark}

Definition~\ref{def:eff-stab} specifies what it means for a presentation to be stable, but not when the presentation is \emph{efficiently} stable. 
As mentioned in the introduction, the goal is to optimise the tradeoff between the encoding length of the presentation and the resulting modulus of stability. We now choose a complexity measurement for presentations:
\begin{definition}[Length of a presentation]\label{def:length_of_pres}
    Let $G=\langle S\colon R\rangle$ be a finite presentation. The \emph{length} of the presentation  will be  $$\ell(G)=|S|+ \sum_{r\in R}|r|\;.$$ 
\end{definition}

The search for the shortest possible presentations of finite groups, and in particular of finite simple groups, contained many twists and turns. Clearly, the shortest presentation of some groups, e.g. $G=\Z_2^k$, needs to be of size at least $\poly\log|G|$. On the other hand, it turns out that there are finite simple groups with minimal presentation length $\poly\log\log|G|$, which was a big surprise (see \cite{guralnick2008presentations} and the references therein). Also, all finite groups (without ${}^2G_2(q)$ composition factors) have a presentation of length $O(\log^3(|G|))$. Thus, we can say that a presentation is ``efficient'', if its length is not much larger than the shortest possible one. This can be phrased in various ways, and certain choices of parameters can be applied in different ways. Our exact choice of tradeoff parameters between length and modulus of stability, which we shall refer to as efficiently stable,  is motivated by natural examples which we outline next.

\subsection{General results}

As was shortly reviewed in the introduction, and specifically in the related works subsection, many results about various notions of stability are known.  Here, we give two results that are most relevant to our work. First of all, 
for a finite group $G$ we can always write $G=\langle S:R\rangle$ where $S = G$ and $R=\{ g\cdot h \cdot (gh)^{-1} =e \}$. We refer to this presentation as the \emph{multiplication table presentation}. If we let $\mu_S$ and $\mu_R$ be the uniform distribution on $S$ and on $R$ respectively then Definition~\ref{def:eff-stab} reduces to a widely used notion of \emph{flexible (Hilbert-Schmidt) stability}. In particular, for finite groups the following result is known~\cite{gowers2017inverse,de2019operator}. We adopt the formulation from~\cite[Theorem 1.4]{de2022spectral}.

\begin{theorem}\label{thm:gh}
Let $G$ be a finite group and $\mC$ the class of all tracial von Neumann algebras. Let $\mu_S$ and $\mu_R$ be the uniform distribution on $S=G$ and $R=\{ g\cdot h \cdot (gh)^{-1}=e \}$ respectively. Then $G=\langle S:R\rangle$ is $(c\eps,\mu_S,\mu_R,\mC)$-stable, where $c>0$ is a universal constant independent of $G$.
\end{theorem}

\begin{remark}\label{rk:linear-modulus}\tnote{adding:}
  We remark that, using Lemma~\ref{lem:stab-group}, Theorem~\ref{thm:gh} implies that any finite presentation $G=\langle S:R\rangle$ of a finite group $G$ is $(O(\eps),\mu_S,\mu_R,\mathcal{C})$-stable, where $\mu_R$ is any distribution with full support and $\mu_S$ is e.g.\ the distribution induced from $\mu_R$ as in Lemma~\ref{lem:l1}. However, the constant implicit in the $O(\eps)$ modulus of stability may of course depend on the presentation, and this will be our focus later on. 
\end{remark}

Next we state for later use a result from~\cite{de2022spectral} which allows us to combine stability results. The results in~\cite{de2022spectral} are rather general and apply to direct products and certain central extensions of a class of finite groups. Here, we will only use the following specialization to the case of the central extension of $\Z_2^k \times \Z_2^k$ by $\{-1,1\}$ given by $\gamma(a,b)=(-1)^{a\cdot b}$, with $a\cdot b$ the inner product modulo $2$. For a measure $\mu$ on $\Z_2^k$, define its inverse spectral gap 
\[ \kappa = \max_{a\neq 0} \frac{1}{1-\Es{b\sim\mu}(-1)^{a\cdot b}}\;.\footnote{Note that this is indeed the inverse of the spectral gap of the random walk operator on $\Z_2^k$ induced by $\mu$. Equivalently, by viewing $\mu$ as an element of the group ring $\complex[\Z_2^k]$, it acts on it from the left, and thus has a spectral decomposition and a spectral gap.}\] 

\begin{theorem}[\cite{de2022spectral} Corollary 2.6]\label{thm:dls-gap}
Let $\mu$ be a measure on $\Z_2^k$ with inverse spectral gap $\kappa$. Let $\phi_X,\phi_Z: \Z_2^k \to \mU(\mM)$ be two homomorphisms such that
\[ \Es{a,b\sim \mu} \big\| \phi_X(a)\phi_Z(b)-(-1)^{a\cdot b} \phi_Z(b)\phi_X(a)\big\|_\tau^2 \,\leq\,\eps\;.\]
Then there is an $\mN=P\mM_\infty P$ and homomorphisms $U_X,U_Z:\Z_2^k\to\mU(\mN)$ and $\delta=O(\kappa^2\eps)$ such that $\phi_X$ and $U_X$ are $(\delta,\mu)$-close, $\phi_Z$ and $U_Z$ are $(\delta,\mu)$-close, and moreover $U_X(a)U_Z(b)=(-1)^{a\cdot b}U_Z(b)U_X(a)$ for all $a,b\in\Z_2^k$.
\end{theorem}
\begin{corollary}
    Theorem \ref{thm:dls-gap} essentially tells us the following: if we are given an almost homomorphism of the Pauli group \eqref{eq:defn_Pauli_as_Heisenberg}, then we can fix it to a homomorphism in two steps. First, fix its restriction to the $X$ and $Z$ observables independently. Then, apply the theorem to get a homomorphism from the whole Pauli group. This idea is spelled out in detail in the proof sketch of Corollary \ref{cor:Pauli-brading_is_stable}.
\end{corollary}

	\subsection{Measurements and orthonormalization}
	\label{sec:measurements}
	
	Before giving some examples, we introduce the notion of a \emph{positive operator-valued measure} (POVM), or more simply 
 a \emph{measurement}. This is a notion that comes from quantum mechanics and will be useful to formulate some of our statements.  
  If $\mM$ is a tracial von Neumann algebra, a measurement on $\mM$ with outcome set $\mA$ is a finite collection of positive semidefinite operators $\{P_a\}_{a\in \mA}$ such that $\sum_a P_a = \Id_\mM$. A measurement is \emph{projective} if for all $a$, $P_a$ is a projection. 
	
	In our results we will make use of the following elementary but powerful result, which allows us to ``pull back'' projective measurements through an isometry. The result is an application of \emph{orthonormalization}, which transforms a nearly-orthogonal measurement to a nearby orthogonal measurement. See e.g.~\cite{kempe2011parallel,ji2020quantum} or~\cite[Theorem 1.2]{de2021orthogonalization} for the version that we use here. 
	
\begin{lemma}\label{lem:pull-back}
Let  $(\mM,\tau^\mM)$ be a tracial von Neumann algebra, $P\in\mM_\infty$ a projection of finite trace, $\mN=P\mM_\infty P$ and $\tau^\mN=\tau_\infty/\tau_\infty(P)$, and $w\in P \mM_\infty \Id_\mM$ a partial isometry. Let 
\[ \eps = \max\big\{ \tau^\mM\big(\Id_\mM - w^* w\big)\,,\;\tau^\mN\big( P- w w^*\big)\big\}\;.\] 
 Then for any projective measurement $\{T_a\}_{a \in \mA}$ on $\mN$, where $\mA$ is a finite set, there is a projective measurement $\{Q_a\}_{a \in \mA}$ on $\mM$ such that 
\begin{equation}
\label{eq:pull-back} \sum_{a \in \mA} \big\| Q_a - w^* T_a w\big\|_{\tau}^2 \,\leq \ 56\eps\;.
\end{equation}
\end{lemma}	

\begin{proof}
If $\eps\geq \frac{1}{2}$ the conclusion is straightforward. This is because, whatever projective measurement $\{Q_a\}_{a\in \mA}$ one chooses, we have
\[
\|Q_a-w^*T_aw\|_\tau^2\leq 2\|Q_a\|_\tau^2+2\|w^*T_aw\|_\tau^2,
\] 
and 
\[
\sum_{a\in \mA} \|Q_q\|_\tau^2,\ \sum_{a\in \mA} \|w^*T_aw\|_\tau^2\leq 1.
\]
So, assume $\eps<\frac{1}{2}$. 
Define 
\[\tilde{Q}_a = w^* T_a w  + \frac{1}{|\mA|}\big(\Id_\mM - w^* w\big) \in \mM\;.\]
Then $\{\tilde{Q}_a\}$ is a POVM on $\mM$. Moreover, 
\begin{align*}
\sum_a \tau^\mM \big( \tilde{Q}_a^2 \big) &\geq \sum_a \tau^\mM \big( \big(w^* T_a w \big)^2 \big) \\
&= \sum_a \tau^\mM \big(  w^* T_a w w^*T_a w \big)\\
&= \sum_a \tau^\mM \big(  w^* T_a  P T_a w \big) - \sum_a \tau^\mM \big( w^* T_a  ( P - w w^*) T_a w \big)\\
&\geq 1 - \eps -  \sum_a \tau_\infty \big( w^* T_a  ( P - w w^*) T_a w \big)\\
&\geq 1 - \eps -  \tau_\infty\Big(\big( P - w w^*\big)\Big(\sum_a  T_a w w^* T_a\Big)\Big)\\ 
&\geq 1- \eps- \tau_\infty\big( P- w w^*\big)\;,
\end{align*}
where the third line uses that $T_aPT_a=T_a$, $\sum_a T_a = \Id_\mN$ and the definition of $\eps$ for the first term, and for the second the fact that for $A\in\mM$, $\tau^\mM(A)=\tau_\infty(A)$ by definition of $\tau_\infty$ and the identification of $\mM$ with a ``corner'' in $\mM_\infty$, the fourth line uses cyclicity of the trace for the second, and the last uses $\|ww^*\|_\infty,\|\sum_a T_a\|_\infty\leq 1$.\footnote{The notation $\|\cdot\|_\infty$ refers to the operator norm.} By assumption, 
\begin{align}
\tau_\infty\big( P- w w^*\big) \,\leq\, \eps\, \tau_\infty(P)\,\leq \frac{\eps}{1-\eps}\;,\label{eq:bound_on_P-ww*_in_trace_infty}
\end{align}
where the last inequality is because by definition, $\tau^N(P)=1$, thus
\[1-\eps \,\leq\, \tau^\mN(ww^*) \,=\, \frac{\tau_\infty(ww^*)}{\tau_\infty(P)}\,=\, \frac{\tau_\infty(w^*w)}{\tau_\infty(P)}  \,\leq\, \frac{1}{\tau_\infty(P)}\]
since $\tau_\infty(w^* w) = \tau^\mM(w^* w)$ and $w^*w\leq I_\mM$. Overall, 
\begin{equation}\label{eq:square_greater_1-3eps}
    \sum_a \tau^\mM \big( \tilde{Q}_a^2 \big) \,\geq\, 1-\eps-\frac{\eps}{1-\eps}\,\geq\, 1-3\eps\;.
\end{equation}
To conclude we apply~\cite[Theorem 1.2]{de2021orthogonalization} to obtain a projective measurement $\{Q_a\}$ on $\mM$ such that 
\begin{equation*}
\sum_a \big\|{Q}_a - \tilde{Q}_a \big\|^2_\tau \,=\, 27\eps\;.
\end{equation*}
Finally,
\begin{align*}
\sum_a \big\|{Q}_a - w^*{T}_a w\big\|^2_\tau &= \sum_a \Big\|{Q}_a - \tilde{Q}_a  + \frac{1}{|\mA|}\big(I_\mM - w^* w\big) \Big\|^2_2\\
&\leq  \sum_a 2\big\|{Q}_a - \tilde{Q}_a\big\|_\tau^2  + 2\frac{1}{|\mA|}\big\|I_\mM - w^* w\big\|_\tau^2 \\
&\leq 54 \eps + 2 \tau^\mM( (I_\mM - w^* w)^2 ) \\
&\leq 54 \eps + 2 \tau^\mM(I_\mM - w^* w ) \\
&\leq 56 \eps\;,
\end{align*}
where the second line is by the triangle inequality, the fourth line is due to the fact that $I_\mM - w^* w$ is positive and has operator norm at most $1$, and the last line is by $\tau^\mM(I_\mM - w^* w ) \leq \eps$.
\end{proof}

\subsection{Examples}
\label{sec:examples}

As a first example we spell out the application of Theorem~\ref{thm:gh} to the case of $G=\Z_2^k$. 
Below, when we write $\Es{i\in \mX}$ where $\mX$ is a finite set, we mean the expectation over $i$ chosen uniformly at random from $\mX$, i.e.\ $\frac{1}{|\mX|} \sum_{i\in \mX}$. 

\begin{corollary}\label{cor:lin-test} 
Let $(\mM,\tau)$ be a tracial von Neumann algebra and $\phi:\Z_2^k \to \mU(\mM)$ such that 
\[ \Es{x,y\in \Z_2^k} \big\| \phi(x)\phi(y)-\phi(x+y) \big\|_{\tau}^2 \,\leq\,\eps\;.\]
Then there is a 
projective measurement $\{P_u\}_{u\in \Z_2^k}$ on $\mM$ such that 
\[ \Es{x\in \Z_2^k} \Big\| \phi(x) -\Big(\sum_u (-1)^{u\cdot x} P_u\Big)  \Big\|_{\tau}^2 \,=\, O(\eps)\;.\]
\end{corollary} 

\begin{proof}
Any $\phi$ as in the corollary statement is an $(\eps,U_R)$-almost homomorphism of $\Z_2^k$ into $(\mM,\tau)$ for the multiplication table presentation. Applying Theorem~\ref{thm:gh}, $\phi$ is $O(\eps)$-close to a homomorphism from $\Z_2^k$ to some $(\mN,\tau^\mN)$. Because $\Z_2^k$ is Abelian, such a homomorphism is given by commuting unitaries $(U_x)_{x\in\Z_2^k}$ on $\mN$. Moreover, since $\Z_2^k$ is a $2$-group, each $U_x$ satisfies $U_x^2=\Id$, hence $U_x=U_x^*$.

For every $u\in  \Z_2^k$ let $Q_u = \Es{x} (-1)^{u\cdot x} U_x$. Then each $Q_u$ is a projection on $\mN$ such that $\sum_u Q_u=\Id$, and $U_x = \sum_u (-1)^{u\cdot x} Q_u$. Furthermore, by the conclusion of Theorem~\ref{thm:gh} it holds that 
\begin{equation}\label{eq:lin-test-1}
   \Es{x\in \Z_2^k} \Big\| \phi(x) - w^*\Big(\sum_u (-1)^{u\cdot x} Q_u\Big)w  \Big\|_{\tau}^2 \,=\, O(\eps)\;,
\end{equation}
for some partial isometry $w\in P\mM_\infty I_\mM$ as in Definition~\ref{def:close}.
Using Lemma~\ref{lem:pull-back}, we find a projective measurement $\{P_u\}$ on $\mM$ that satisfies 
\begin{equation}\label{eq:lin-test-1b}
 \sum_u \big\| P_u - w^* Q_u w \big\|_\tau^2 \,=\, O(\eps)\;.
\end{equation}
Thus
\begin{align}
\Es{x\in \Z_2^k}  \Big\| \sum_u (-1)^{u\cdot x} \big(P_u - w^* Q_u w\big) \Big\|_\tau^2
&=\Es{x\in \Z_2^k} \sum_{u,v} (-1)^{(u+v)\cdot x} \tau\big(\big(P_u - w^* Q_u w\big)\big(P_v - w^* Q_v w\big)\big) \notag\\
&= \sum_u \big\| P_u - w^* Q_u w\big\|_\tau^2\notag\\
&= O(\eps)\;,\label{eq:lin-test-2}
\end{align}
where the second line uses $\Es{x} (-1)^{w\cdot x} = 0$ if $w\neq 0$, and $1$ otherwise, and the last line is by~\eqref{eq:lin-test-1b}. Plugging back into~\eqref{eq:lin-test-1} and using the triangle inequality shows the corollary.  
\end{proof}

The multiplication table presentation of $\Z_2^k$ is quite long,  in the sense of Definition \ref{def:length_of_pres}. It has as many generators as the group size, and quadratically as many relations, which gives a length of $O(2^{2k})$. 
There are much shorter presentations of $\Z_2^k$, for example the straightforward

\begin{equation}\label{eq:z2-efficient}
 \Z_2^k = \langle x_1,\ldots,x_k : [x_i,x_j]=e, x_i^2=e \; \forall i\neq j \rangle\;,
\end{equation}
where $[x_i,x_j]=x_ix_jx_i^{-1}x_j^{-1}$ is the group commutator. 
Its length is $O(k^2)$, which is polylogarithmic in the group size instead compared to the polynomial length of the multiplication table presentation. By \cite{hadwin2018stability}, every presentation of a finite group has \textbf{some} modulus of stability.
The following two lemmas show that, though the presentation \eqref{eq:z2-efficient} has a linear modulus of stability (as expected, see Remark~\ref{rk:linear-modulus}), it deteriorates by a constant factor $k$. In some sense, these two examples are on opposite sides of the efficient stability tradeoof: \eqref{eq:z2-efficient} is short but has a bad modulus of stability, while the multiplication table presentation is very long, while having an essentially optimal modulus of stability. As we keep recalling, the goal of this paper is to provide an example of a somewhat short presentation of $\Z_2^k$, which has a good enough modulus of stability to deduce the main technical results needed for $\MIP^*=\RE$~\cite{ji2020mip}.

\begin{lemma}[Lemma 3.8 in~\cite{slofstra2019set}]\label{lem:eff-z2}
Let $\mC$ be the class of tracial von Neumann algebras. Let $\mu_R$ be the equal mixture of the uniform distribution on all words $[x_i,x_j]$ ($i\neq j$) and the uniform distribution on all words $x_i^2$. Let $\mu_S$ be the uniform distribution on $\{x_1,\ldots,x_k\}$. Then 
  for every $k$, there is a $\delta_k = O_k(\eps)$ such that the presentation~\eqref{eq:z2-efficient} is $(\delta_k,\mu_S,\mu_R,\mC)$-stable. Furthermore, the close representation can be taken on the same algebra.
\end{lemma}

As one would expect from Lemma~\ref{lem:stab-group}, the dependence of $\delta_k$ on $k$ depends on both the diameter of $\Z_2^k$ for the presentation~\eqref{eq:z2-efficient}, and certain values of its \emph{Dehn function} --- namely, the minimal volume of a Van Kampen diagram with perimeters of length $3$. By expressing each element of $\Z_2^k$ as a product of generators in the natural way, and by applying Corollary~\ref{cor:lin-test}, it is possible to show that $\delta_k=O(k^2\eps)$ in Lemma~\ref{lem:eff-z2}. With more work one can get $\delta_k=O(k\eps)$. This was shown in~\cite[Theorem 3.2]{chao2017overlapping} for the case of the operator norm. We show it for the dimension-normalized Frobenius norm in Appendix~\ref{sec:proj}, using an iterative rounding argument that (partially) parallels the one in~\cite{chao2017overlapping}.  This turns out to be  tight, as the next lemma shows. 

\begin{lemma}\label{lem:lower_bound_on_stability_rate_standard_presentation_Z_2^k}
    Let $\mu_R$ and $\mu_S$ be as in Lemma~\ref{lem:eff-z2}. 
    Then, for every $1\leq c\leq \frac{k}{2}$, there is a $(\nicefrac{c}{\binom{k}{2}},\mu_R)$-almost homomorphism of $\langle S\colon R\rangle$ which is at least  $(\nicefrac{c}{16k},\mu_S)$-far from any homomorphism (according to Definition~\ref{def:close}). 
\end{lemma}

Lemma~\ref{lem:lower_bound_on_stability_rate_standard_presentation_Z_2^k} implies that the modulus of stability of \eqref{eq:z2-efficient} is $\Omega(k\eps)$ whenever $\eps\leq\frac{1}{k-1}$. This is because, for any $0<\eps<\nicefrac{1}{k-1}$, one can add trivial representations of $\Z_2^k$ to copies of the approximate representation described in Lemma~\ref{lem:lower_bound_on_stability_rate_standard_presentation_Z_2^k} (choosing $c=\nicefrac{k}{2}$), such that the resulting map is an $(\eps,\mu_R)$-approximate representation which is at least $(\nicefrac{(k-1)\eps}{32},\mu_S)$-away from any actual representation of $\Z_2^k$. We record this observation as the following corollary.

\begin{corollary}\label{cor:delta-z2k}
The presentation of $\Z_2^k$ given in~\eqref{eq:z2-efficient}, together with the distributions $\mu_R$ and $\mu_S$ specified in Lemma~\ref{lem:eff-z2}, has a modulus of stability  $\delta_k(\eps)$ with respect to the class $\mathcal{C}$ of tracial von Neumann algebras that satisfies
\[ \frac{1}{C} \min(k\eps,1)\,\leq\,\delta_k(\eps)\,\leq\, C\min(k\eps,1)\;,\]
where $C>0$ is some constant.
\end{corollary}

The proof of Lemma~\ref{lem:lower_bound_on_stability_rate_standard_presentation_Z_2^k} is given in Appendix \ref{appendix:lower_bounds}, which also includes some discussion on the $L^\infty$ variant of it.

In the next section we will obtain presentations of $\Z_2^k$ that have a much better length/modulus tradeoff than the one given in~\eqref{eq:z2-efficient} or the multiplication table presentation. In the meantime, we give one last example. To formulate it we recall the definition of the Pauli matrices
	\begin{equation}\label{eq:def-pauli-1} \sigma^X = \begin{pmatrix} 0 & 1 \\ 1 & 0 \end{pmatrix}\;,\qquad \sigma^Z = \begin{pmatrix} 1 & 0 \\ 0 & -1\end{pmatrix}\;,
	\end{equation}
	and more generally for $a,b\in \Z_2^k$ let 
	\begin{equation}\label{eq:def-pauli-2}
	\sigma^X(a) = \bigotimes_{i=1}^k (\sigma^X)^{a_i}\qquad\text{and}\qquad\sigma^Z(b) = \bigotimes_{i=1}^k (\sigma^Z)^{b_i}\;.
	\end{equation}
These are self-adjoint unitary operators called Pauli observables. Each observable $\sigma^X(a)$ (resp. $\sigma^Z(b)$) corresponds to the \emph{Pauli measurement} $\{ \sigma^X_a \}_{a \in \Z_2^k}$ (resp. $\{ \sigma^Z_b \}_{b \in \Z_2^k}$) where (in a slight abuse of notation)  
	\[    \sigma^X_a = \Es{\alpha\in\F_2^k} (-1)^{a\cdot \alpha} \sigma^X(\alpha)\qquad\text{and}\qquad\sigma^Z_b = \Es{\beta\in\F_2^k} (-1)^{b\cdot\beta} \sigma^Z(\beta).\]
	It is easy to verify that $\{\sigma^X_a\}_a$ and $\{\sigma^Z_b\}_b$ are projections summing to identity.	
	
For an integer $k\geq 1$,  the Pauli group $\pauli_k$ is the group generated by the Pauli matrices $\sigma^X(a)$, $\sigma^Z(b)$ introduced in~\eqref{eq:def-pauli-2}. It can also be defined more abstractly as follows. Let $\gamma: \Z_2^k\times \Z_2^k \to \{-1,1\}$ be given by $\gamma(a,b)=(-1)^{a\cdot b}$. Then $\pauli_k$ is the central extension of $\Z_2^k\times \Z_2^k$ by $\{-1,1\}$ given by $\gamma$. This group 
is also known as the Heisenberg group 
\begin{equation}\label{eq:defn_Pauli_as_Heisenberg} H_{2k+1} = \left \{ \begin{pmatrix} 1 & a & c \\ 0 & I_{k \times k} & b \\ 0 & 0 & 1 \end{pmatrix} \right\}\subseteq {\mathrm{GL}}_{k+2}(\F_2)\;.
\end{equation}
 
We consider the following presentation for the Pauli group:
\begin{align}
 \pauli_k \,=\, \big\langle \{J\}\cup  \{(a,0),(0,b): a,b\in \Z_2^k\} &: (a,0)^2=(0,b)^2=J^2=e,\; [(a,0),J]=[(0,b),J]=e,\notag\\
 &\quad(a,0)(a',0)=(a+a',0),(0,b)(0,b')=(0,b+b') \notag\\
&\quad  (a,0)(0,b) = J^{a\cdot b} (0,b)(a,0) \; \quad \forall a,b,a',b'\in \Z_2^k\big\rangle\;.\label{eq:almost_mult_table_pres_of_Pauli}
\end{align}
This presentation is not quite the multiplication table presentation (it has $2^{k+1}+1$ generators, whereas $|\pauli_k|=2^{2k+1}$), but it is not far from it.  
Applying Theorem~\ref{thm:gh} we obtain the following consequence.

\begin{corollary}[Pauli braiding test]\label{cor:Pauli-brading_is_stable}
Let $\phi_X,\phi_Z:\Z_2^k \to \mU(\mM)$ be maps such that for all $W\in \{X,Z\}$,
\[ \Es{a,b\in \Z_2^k} \big\| \phi_W(a)\phi_W(b)-\phi_W(a+b) \big\|_\tau^2 \,\leq\,\eps\;,\]
and
\[ \Es{a,b\in \Z_2^k} \big\| \phi_X(a)\phi_Z(b)- (-1)^{a\cdot b} \phi_Z(b)\phi_X(a) \big\|_{\tau}^2 \,\leq\,\eps\;.\]
Then there is a projection $P\in \mM_\infty$ such that if $\mN=P\mM_\infty P$ then $\mN\simeq (M_2(\C))^{\otimes k} \otimes \mN'$ for some $\mN'$, and a
partial isometry $w\in P\mM_\infty I_\mM$ such that for all $W\in \{X,Z\}$, 
\[ \Es{W\in \Z_2^k} \big\| \phi_W(a) - w^* \big(\sigma_W(a)\otimes \Id_{\mN'}\big) w \big\|_{\tau}^2 \,=\, O(\eps)\;.\]
\end{corollary}

Since this statement already appears in~\cite{natarajan2017quantum} (restricted to approximate homomorphisms into the finite-dimensional unitaries, and with slightly worse dependence on $\eps$), and we will not need it here, we only sketch the proof. 

\begin{proof}[Proof sketch]
The idea is the following: The maps $\phi_X,\phi_Z$ induce an $\eps$-approximate representation of the Pauli group with respect to the presentation \eqref{eq:almost_mult_table_pres_of_Pauli}. From that, we can use the canonical form of elements in $\pauli_k$ to create an $O(\eps)$-approximate representation of the multiplication table presentation of $\pauli_k$ --- define $\phi:\pauli_k \to \mU(\mM)$ by $\phi((-1)^c \sigma_X(a)\sigma_Z(b))= (-1)^c \phi_X(a)\phi_Z(b)$, where $c\in \{\pm 1\}$ and $a,b\in \Z_2^k$. Applying Gowers--Hatami (Theorem \ref{thm:gh}), we deduce the corollary.
\end{proof}

Similarly to Lemma~\ref{lem:eff-z2} we can state a short version of the preceding corollary, with a bad modulus of stability, which applies to the presentation
\begin{align}
 \pauli_k = \big\langle x_1,\ldots,x_k,z_1,\ldots,z_k &: x_i^2=z_i^2=J^2=e,\; [x_i,J]=[z_i,J]=e,\notag\\
&\qquad  [x_i,x_j]=[z_i,z_j]=[x_i,z_j]=e, [x_i,z_i]=J \quad \forall i\neq j \big\rangle\;.\label{eq:pauli-efficient}
\end{align}

\begin{lemma}\label{lem:eff-pauli}
  Let $\mC$ be the class of tracial von Neumann algebras. Let $\mu_R$ be the following sampling procedure: With equal probability do one of the following
  \begin{itemize}
      \item Choose $J^2=e$.
      \item Sample $i\in [k]$ uniformly and choose $x_i^2=e$ (respectively, $z_i^2=e$).
      \item Sample $i\in [k]$ uniformly and choose $[x_i,z_i]=J$.
      \item Sample $i\neq j\in [k]$ uniformly at random and choose $[x_i,x_j]=e$ (respectively, $[z_i,z_j]=e$ or $[x_i,z_j]=e$).
  \end{itemize}
  Let $\mu_S$ be the marginal of $\mu_R$, as described in Remark~\ref{rk-mus}. Then 
 the presentation~\eqref{eq:pauli-efficient} is $(O(k\eps),\mu_S,\mu_R,\mC)$-stable.
\end{lemma}

\begin{proof}
This follows from~\cite{chao2017overlapping} and Theorem~\ref{thm:dls-gap}, in a similar manner to our proof sketch of Corollary \ref{cor:Pauli-brading_is_stable}.
\end{proof}

\section{Presentations from codes}
\label{sec:pres-codes}

In this section we lay the groundwork for our main result, an ``efficient'' stable presentation of $\Z_2^k$ that is presented in the next section. Most of the technical work required is done in~\cite{ji2020mip}. In this section we introduce the language required to reformulate their result in the framework of this paper. 

As a first step, we introduce a general method for translating any binary linear error-correcting code into a presentation of $\Z_2^k$. Later we apply this method to the specific case of the Reed-Muller code (composed with the Hadamard code to obtain a binary code). However, the general method may be of independent interest. 

\subsection{The general construction}

\label{sec:pres-code}

For $q$ a prime power we let $\F_q$ denote the finite field with $q$ elements. 
For $n,k,d$ integer, an $[n,k,d]_q$ linear code $\code$ is a $k$-dimensional subspace of $\F_q^n$ such that for all $x\in \code$ such that $x\neq 0$, the Hamming weight $|x|$ (i.e.\ the number of nonzero coordinates of $x$) is at least $d$. The parameter $n$ is called the \emph{length} of the code, $k$ its \emph{dimension} and $d$ its \emph{distance}. A code can be specified by a \emph{parity check matrix} $h\in \F_q^{m\times n}$ such that $\code = \ker h$. 

For the remainder of this section we specialize the discussion to the case where $q=2$. 
We make the simple but key observation that a parity check matrix for a code of dimension $k$ implies a finite group presentation in the following way. Introduce $n$ generators $S=\{x_1,\ldots,x_n\}$. Each of the generators is required to be an involution: $x_i^2=e$. For each row $i\in \{1,\ldots,m\}$ of the parity check matrix $h$, introduce a relation 
\[ R_i\,:\; \prod_{1\leq j \leq n} x_j^{h_{ij}}=e\;, \]
that ``verifies'' the parity check associated with the $i$-th row of $h$. Finally, to guarantee that the choice of order in which the $x_j$ are multiplied in $R_i$ will not affect the resulting group, whenever $j\neq j'$ are such that $h_{ij}$ and $h_{ij'}$ are both nonzero we require that $x_j$ and $x_{j'}$ commute. This can be written succinctly using a relation 
\[ R'_{ijj'}\,:\; [x_j,x_{j'}]^{h_{ij} h_{ij'}}=e\;.\]
The presentation obtained in this way defines a group $G=G(h)$, already introduced in~\eqref{eq:gh-intro} and which with the present notation reads
\begin{equation}\label{eq:def-gh-pres}
 G(h) \,=\, \big\langle x_1,\ldots,x_n \,:\, x_j^2=e\,,\; R_i\,,\; R'_{ijj'}\,,\quad \forall 1\leq i\leq m,\, 1\leq j< j' \leq n\big\rangle\;.
\end{equation}
Note that we made the dependence of $G(h)$ on $h$ explicit. This is because in general, the group defined in this way may depend on $h$, and not only on $\code$. If however we further impose \emph{all} pairwise commutation relations, as in~\eqref{eq:wtgh-intro}, then we obtain the following. 

\begin{lemma}\label{lem:com-code}
Let $R''_{jj'}$ be the commutation relation $[x_j,x_{j'}]=e$. Then $\Z_2^k$ is isomorphic to 
\[ \Z_2^k \,\simeq\, \big\langle x_1,\ldots,x_n \,:\, x_j^2=e\,,\; R_i\,,\; R''_{jj'}\,,\quad \forall 1\leq i\leq m,\, 1\leq j< j' \leq n\big\rangle\;.\]
\end{lemma}

\begin{proof}
The group defined by the right-hand side is obviously abelian and a $2$-group, so it is of the form $\Z_2^{k'}$ for some $k'$. In fact, it is isomorphic to the quotient of $\Z_2^n$ by the subgroup generated by the $\prod_{1\leq j \leq n} x_j^{h_{ij}}$. So it is isomorphic to $\Z_2^k$ where $k=n-\dim\textrm{im}\ h = \dim\ker h=\dim \code$.  
\end{proof}

\begin{remark}\label{rk:non-abelian}
There exists matrices $h$ such that $G(h)$ is not $\Z_2^k$, and in fact is not Abelian. For an example, see~\cite[Example 2.16]{paddock2022arkhipov}. For that example, $n=12$, $k=7$, and the parity check matrix $h$ can be described explicitly as follows: the $7$ rows of $h$ are indexed by the vertices of the complete bipartite graph $K_{3,4}$, the $12$ columns are indexed by the $12$ edges of $K_{3,4}$, and the entry $(i,j)$ of $h$ is $1$ if and only if the edge $j$ is incident on vertex $i$. As shown in~\cite{paddock2022arkhipov}, $G(h)$ is not Abelian. By considering $K_{3,6}$ instead of $K_{3,4}$, one in addition obtains a non-Abelian \emph{infinite} group. Using similar arguments it is possible to construct $h$ such that $G(h)$ is not amenable, etc.; see the discussion in~\cite[Section 6]{paddock2022arkhipov}.
\end{remark}

For readability it is convenient to reformulate the parity check matrix as a \emph{tester} for the code. This allows us to give a more succinct, ``algorithmic'' definition of a parity check matrix for a given code. Informally, the tester takes as input a word $w\in \F_2^n$ and determines if $w\in \code$ by selecting a parity check at random and evaluating it. 
 Specifically we give the following definition. (For the sake of later use, we state the definition for the case of a general prime power $q$.)

\begin{definition}[$r$-local linear tester]\label{def:code-test}
Let $\code$ be an $[n,k,d]_q$ linear code and $r\in \N$.
An \emph{$r$-local linear tester for $\code$} is a pair $M = (h,\nu)$ where $h \in \F_q^{m \times n}$ is a parity check matrix for $\code$, whose every row has Hamming weight at most $r$, and $\nu$ is a distribution over $\{1,\ldots,m\}$. 
\end{definition}

An $r$-local linear tester $M=(h,\nu)$ for $\code$ induces a pair of distributions $(\nu_R,\nu_S)$ on the relations and generators of the presentation $G(h)$~\eqref{eq:def-gh-pres} in a natural way: For the generators, we let $\nu_S$ be induced from $\nu$ by first sampling $j\sim \nu$ and then a uniformly random $i$ such that $h_{ji}\neq 0$. For the relations, we let $\nu_R$ be the uniform mixture of the distribution $\nu_S$ on relations $x_j^2=e$, the distribution $\nu$ on relations $R_i$, and the distribution $\nu\times \nu_S\times \nu_S$ on relations $R'_{jii'}$.

We end this section with an example, the \emph{Hadamard code}. This code can be defined for any  $t\geq 1$ and it is a $[T,t,T/2]_2$ linear code, where $T=2^t$. For simplicity we write  $\code_\had$ to denote this code, omitting $t$. The Hadamard code is the subspace of linear functionals from $\field^t$ to $\field$ out of all such functions. As a linear space, $\code_\had$ can be described as  $(a\cdot b)_{a\in \field^t} \in \field^{\field^t}\cong \field^T$, for all $b\in \field^t$, where $ a\cdot b=\sum_{i=1}^t a_ib_i$ is again the dot product modulo $2$. 

A parity check matrix for $\code_\had$ is the matrix $h_\had\in \F_2^{T^2\times T}$ defined as follows. Identify the rows of $h_\had$ with pairs $(x,y)\in \F_2^t\times \F_2^t$, and the columns of $h_\had$ with $\F_2^t$. Then the $(x,y)$-th row of $h_\had$ has nonzero entries at positions $x,y$ and $x+y$ only. 
The corresponding $3$-local linear tester is $M_\had = (h_\had,\nu)$ where $\nu$ is the uniform distribution over $\F_2^t \times \F_2^t$. This tester can be described algorithmically, see Figure~\ref{fig:test-had}. 

\begin{figure}[!htbp]
  \centering
  \begin{gamespec}
	Given access to some $g\in \F_2^T$, where $T=2^t$, identify $g$ with a function $g:\F_2^t\to\F_2$. Perform the following. 
\begin{enumerate}
\item Select $(x,y)\in \F_2^t \times \F_2^t$ uniformly at random. 
\item Accept if and only if $g(x)+g(y)+g(x+y)=0$.  	
    \end{enumerate}
  \end{gamespec}
  \caption{A $3$-local linear tester for $\code_{\had}$}
  \label{fig:test-had}
\end{figure}

\begin{remark}
Since each pair of coordinates $(x,y)$ appears together in at least one parity check, we can apply Lemma~\ref{lem:com-code} to deduce that $G(h_\had)=\Z_2^t$. 
\end{remark}

We state our first stability result for a code-based presentation, the presentation $G(h_\had)$ defined as~\eqref{eq:def-gh-pres} where $h_\had$ is defined above. To state the result we need to specify distributions $\mu_S$ and $\mu_R$. We let $\mu_S$ be uniform over the $2^t$ generators, and $\mu_R$ the uniform distribution over the relations $R_{i}$. (Here, there is no need to place any weight on the relations $x_i^2=e$, or on the commutation relations $R'_{ijj'}$, because it can be seen that they follow from the other relations.)  

\begin{lemma}\label{lem:had-stab}
Let $\mC$ be the class of all tracial von Neumann algebras. 
The presentation $\Z_2^t= G(h_\had)$, together with the distributions $\mu_S$ and $\mu_R$ defined above, is $(\delta,\mu_S,\mu_R,\mC)$ stable with $\delta(\eps)=O(\eps)$. 
\end{lemma}

\begin{proof}
This an immediate consequence of Theorem~\ref{thm:gh}, because $G(h_\had)$ is the multiplication table presentation for $\Z_2^t$.
\end{proof}

\subsection{The Reed-Muller code over $\F_q$}
\label{sec:rmq}

We introduce a family of codes that will lead to interesting presentations $G(h)$, whose stability we are able to analyze.
Fix integers $m,t \in \N$ and let $q=2^t$ and $M = 2^m$. \tnote{added:}Let $1\leq d < q$. Let $\mP(q,m,d)$ be the vector space over $\F_q$ that consists of all $m$-variate polynomials $f$ over $\F_q$ of individual degree at most $d$, that is all functions of the form
\[
	f(x_1,\ldots,x_m) = \sum_{\alpha \in \{0,1,\ldots,d\}^m} c_\alpha\,
  x_1^{\alpha_1} \cdots x_m^{\alpha_m}\;,
\]
where $\{c_\alpha\}$ is a collection of coefficients in $\F_q$. It is easy to verify that $\mP(q,m,d)$ has dimension $k = (d+1)^m$ over $\F_q$. It follows that the linear span of all $(f(x))_{x\in \F_q^m}$, when ranging over all possible $\{c_\alpha\}$, defines a $[q^m,(d+1)^m,D]_q$ linear code over $\F_q$, where $D\geq (1-md/q)q^m$ follows from the Schwartz-Zippel lemma.

\begin{lemma}[Schwartz-Zippel lemma~\cite{Sch80,Zip79}]
  \label{lem:schwartz-zippel}
  Let $f, g: \F_q^m \to \F_q$ be two unequal polynomials with total degree at most $d$. Then
  \begin{equation*}
    \Pr_{x \sim \F_q^m}\big(f(x) = g(x)\big) \leq \frac{d}{q}\;.
  \end{equation*}
\end{lemma}

The resulting code is called the \emph{Reed-Muller code} $\code_\RM$ with parameters $q,m,d$. For $m=1$, one obtains the \emph{Reed-Solomon code} $\code_\RS$. A useful feature is that $\code_\RM$ can be seen as the $m$-fold tensor product of $\code_\RS$, i.e.\ $\code_\RM = \code_\RS^{\otimes m}$ as vector spaces over $\F_q$. 

We define a local linear tester $M_{\RM}$ for the code $\code_\RM$ over $\F_q$. The tester is described as an algorithmic procedure in Figure~\ref{fig:RM-tester}. 
The description makes use of interpolation coefficients, which are defined as follows. Fix $d+1$ distinct values $t_0,\ldots,t_d \in \F_q$. \tnote{adding:}(These are fixed once and for all, and their choice does not impact any subsequent statement.) Then for all $u,v \in \F_q$ and $i \in \{0,\ldots,d\}$ define the interpolation coefficients
 \begin{equation}\label{eq:interp-coeff}
 \alpha_{u,v,i} = \prod_{\substack{i'=0\\i'\neq i}}^{d}  \frac{v - (u + t_{i'})}{t_i - t_{i'}}~.
 \end{equation}
These are defined so that any polynomial $f:\F_q\to\F_q$ of degree at most $d$ satisfies
that for all $v \in \F_q$, 
\[ f(v)\,=\, \sum_{i=0}^{d} \alpha_{u,v,i} \, f(u+t_i)\;.\]

The tester verifies this relation along a randomly chosen \emph{axis-aligned direction}.  For all points $u \in \F_q^m$ and $j \in \{1,\ldots,m\}$, we say that the line through $u$ parallel to the $j$-th axis is the set of points $\{ u + te_j : t \in \F_q \}$ where $e_j=(0,\ldots,0,1,0,\ldots,0)\in \F_q^m$, where the unique $1$ is in the $j$-th position.

\begin{figure}[!htbp]
  \centering
  \begin{gamespec}
Given access to some $g\in \F_q^n$, where $n=q^m$, identify $g$ with a function $g:\F_q^m\to \F_q$. Perform the following.
\begin{enumerate}
	\item Sample	$u\in \F_q^m$ and $j\in \{1,\ldots,m\}$ uniformly at random. Let $v$ be a uniformly random point on the line through $u$ parallel to the $j$-th axis.
	\item 
	Accept if and only if $g(v) = \sum_{i=0}^{d} \alpha_{u,v,i} g(u+t_i e_j)$. 
    \end{enumerate}
  \end{gamespec}
  \caption{A local test for $\code_{\RM}$}
  \label{fig:RM-tester}
\end{figure}

A parity check matrix $h_{\RM} \in \F_q^{S \times q^m}$ for $\code_{\RM}$, where $S = q^m \times m \times q$, is as follows. Identify the rows with triples $(u,j,t) \in \F_q^m \times \{1,\ldots,m\} \times \F_q$. The $(u,j,t)$-th row of $h_{\RM}$ is the vector in $\F_q^m$ that for $i \in \{0,\ldots,d\}$ has the value $\alpha_{u,v,i}$ for $v = u + t e_j$ in the coordinate indexed by $u + t_i e_j$, the value $-1$ in the coordinate indexed by $u + te_j$, and the value $0$ everywhere else. 

Rubinfeld and Sudan~\cite{rubinfeld1996robust} (building on the work of Babai, Fortnow, and Lund~\cite{babai1991non}) showed that $h_{\RM}$ is indeed a parity check matrix for $\code_{\RM}$. Furthermore, they showed that the parity check matrix gives rise to a $(d+2)$-local tester for $\code_{\RM}$, whose soundness $\rho$ (as defined in the introduction) is at least $\frac{1}{6}$. 



\subsection{Code composition}
\label{sec:code-comp}

 The Reed-Muller code from the previous section is defined over $\F_q$, for $q=2^t$ a power of $2$. We can transform any $q$-ary code, for $q=2^t$, into a binary code using the idea of \emph{code composition} which we now describe. 




\tnote{edited next two paragraph following referee comments:}
Let $q=2^t$ and $\code$ be an $[n,k,d]_q$ linear code. We define an $[qn,tk,d']_2$ linear code as follows. First introduce the $\F_2$-linear map $\varphi:\F_q\to\F_2^q$ defined by $\varphi(x)=(\tr(xa))_{a\in \F_q}$. Here  $\tr(\cdot):\F_q\to\F_2$ denotes the trace over $\F_2$, wich is defined by $\tr(x) = \sum_{j=0}^{t-1} x^{2^i}$ for $x\in \F_q$. The code $\code'$ is defined as the linear span, over $\F_2$, of all $(\varphi(x_1),\ldots,\varphi(x_n))$ for $(x_1,\ldots,x_n)\in \code$. It is apparent from the definition that this code has dimension $tk$, the dimension of $\code$ over $\F_2$, and length $qn$. Furthermore, it is not hard to verify that it has distance $d'\geq dq/2$.

The reason that this construction is referred to as ``composition'' is the following. Observe that the linear span, over $\F_2$, of all $\varphi(x)$, $x\in \F_q$, is (up to reordering of the indices) identical to the Hadamard code $\code_\Had$ over $\F_2^t$ introduced at the end of Section~\ref{sec:pres-code}. Thus each $\varphi(x)$ can be interpreted as an ``encoding'' of the element $x\in \F_q$ (when seen as an element of $\F_2^t$ by fixing a basis); and  $(\varphi(x_1),\ldots,\varphi(x_n))$ is the successive encoding of an element $m\in\F_2^k$ through $\code$ (to obtain $(x_1,\ldots,x_n))$ and then $\code_\Had$.

%

Given an $r$-local tester $M=(h,\nu)$ for $\code$, there is a natural $\max(r,3)$-local tester $M'$ for $\code'$ which can be described as follows. Index coordinates of $\code'$ by pairs $(i,\alpha)\in [n]\times\F_2^t$, fixing a bijection between $[qn]$ and $[n]\times \F_2^t$.  We describe an $\max(r,3)$-local tester $M' = (h',\nu')$ for $\code'$. Informally, $h'$ contains two type of checks. First, the $A$-checks consist of the repetition of $n$ copies of the checks for the Hadamard code, one for each Hadamard-code encoding $\varphi(x)$ of an $\F_q$-symbol $x$ from $\code$. Second, the $B$-checks implement the checks of $\code$ specified by $M$, directly on the Hadamard encoding.

More precisely, define $h'$ to be the block matrix $h'=\begin{pmatrix} A \\ B \end{pmatrix}$ where 
\begin{itemize}
	\item $A \in \F_2^{nq^2 \times nq}$ is itself a block-diagonal matrix where the diagonal blocks are the $q^2 \times q$ parity check matrix for the Hadamard code. 
	In other words, $A$ can be viewed as $I_{n\times n} \otimes h_\had$ where $h_\had \in \F_2^{q^2 \times q}$ is the parity check matrix for the Hadamard code. 
	\item $B \in \F_2^{\ell q \times nq}$, where $\ell$ is the number of rows of $h$, is viewed as having rows indexed by pairs $(p,\gamma) \in \{1,\ldots,\ell\} \times \F_q$ and columns indexed by pairs $(i,x) \in \{1,\ldots,n\} \times \F_2^t$. The entry in row $(p,\gamma)$ and column $(i,x)$ is $1$ if and only if $h_{pi} \neq 0$ and $x = \kappa(\gamma h_{pi})$.
\end{itemize}
Define the distribution $\nu'$ as the uniform mixture of the uniform distribution on the rows of the $A$ block matrix and the uniform distribution on the rows of the $B$ block matrix. 

\begin{claim}
$h'$ is a parity check matrix for $\code'$.
\end{claim}

\begin{proof}
Let $x\in \F_2^{qn}$ be such that $h'x=0$. Since the $A$-checks enforce that each block of $k$ symbols contains the Hadamard-code encoding of an $\F_q$ symbol, $x$ can be decoded to $x'\in \F_q^n$ such that for each $(i,x)\in \{1,\ldots,n\}\times \F_2^t$, $x_{(i,x)} = \kappa(x'_i)\cdot x$. If the $p$-th row of $h$ enforces the check $v_p\cdot x'=0$, where $v_p\in \F_q^n$, then 
the $(p,\gamma)$-th row of $B$ enforces the check 
\begin{align*}
0&=\sum_{j=1}^n \kappa(\gamma (v_p)_j) \cdot \kappa(x'_j)\\
&= \sum_{j=1}^n \tr( \gamma (v_p)_j x'_j) \\
&= \tr(\gamma (v_p\cdot  x'))\;.
\end{align*}
Therefore, $h'x=0$ is equivalent to $\tr(\gamma(v_p \cdot x'))=0$ for all $\gamma$, which is equivalent to $v_p\cdot x'=0$. Thus $\code'=\ker h'$, as desired. 
\end{proof}

\section{An efficient presentation for $\Z_2^k$}
\label{sec:eff-z2k}

Fix integers $m,t,d \in \N$ and let $q=2^t$. Let $\code_{\bRM}$ be the $[q^{m+1},t(d+1)^m,D']$ code obtained by applying the composition procedure from Section~\ref{sec:code-comp} to the $[q^m,(d+1)^m,D]_q$ Reed-Muller code $\code_\RM$ from Section~\ref{sec:rmq}. 

\tnote{moved this here from Section~\ref{sec:code-comp}:}Recall that $\tr(\cdot):\F_q\to\F_2$ denotes the trace over $\F_2$. We often identify elements of $\F_q$ with vectors in $\F_2^t$. To make this identification precise and convenient, we introduce the following notation. First, fix a  self-dual basis $\{e_1,\ldots,e_t\}$  of $\F_q$ over $\F_2$. (A self-dual basis $\{e_1,\ldots,e_t\}$ is one which satisfies $\tr(e_ie_j)=\delta_{ij}$ for all $1\leq i,j\leq t$.)
Now let $\kappa: \F_q \to \F_2^t$ denote the invertible linear map such that $\kappa(a)$ is the vector of coefficients of $a\in \F_q$ in the basis $\{e_1,\ldots,e_t\}$, i.e.\ $\kappa(a)_i=\tr(ae_i)$ for $i\in\{1,\ldots,t\}$.

\begin{figure}[!htbp]
  \centering
  \begin{gamespec}
Given access to some $g\in \F_2^{N}$, where $N=q^{m+1}$, identify $g$ with a function $g:(\F_2^t)^m \times \F_2^t \to \F_2$. Perform one of the following tests with probability~$\tfrac{1}{2}$ each. 
\begin{enumerate}
	\item \textbf{Low-degree test:}
		Let $u \in \F_q^m$ be a uniformly random point and $j\in \{1,\ldots,m\}$ chosen uniformly at random. Let $\ell$ be the line through $u$ in the $j$-th direction. Let $e_j=(0,\ldots,0,1,0,\ldots,0)\in \F_q^m$, where the unique $1$ is in the $j$-th position. Choose a uniformly random $v\in \ell$ and $\gamma\in \F_q$ and check that 
		\[\sum_{i=0}^d g(\kappa(u+t_i e_j),\kappa(\gamma \alpha_{u,v,i})) \,=\, g(\kappa(v),\kappa(\gamma))\;,\]
		where the interpolation points $t_0,\ldots,t_d \in \F_q$ and the $\alpha_{u,v,i}$ are defined in~\eqref{eq:interp-coeff} and the sentence that precedes it.
	\item \textbf{Hadamard test:} Let $u\sim\F_q^m$ be chosen uniformly at random and $\alpha,\beta\in \F_2^t$ chosen uniformly at random. Check that 
	\[g(\kappa(u),\alpha)+g(\kappa(u),\beta)\,=\,g(\kappa(u),\alpha+\beta)\;.\] 	
    \end{enumerate}
  \end{gamespec}
  \caption{A local test for $\code_{\bRM}$}
  \label{fig:bRM-tester}
\end{figure}

This notation is used in Figure~\ref{fig:bRM-tester}, in which we give an algorithmic description of the natural local tester associated with $\code_{\bRM}$. This tester is obtained by composing, following the template described in Section~\ref{sec:code-comp}, the local tester for the Reed-Muller code over $\F_q$ described in Figure~\ref{fig:RM-tester} with the local tester for the Hadamard code.
Let $(h_\bRM,\nu)$ be the tester that is implied by the figure. Then $h_\bRM\in \F_2^{M\times N}$, where $N=q^{m+1}$ and $M=q^m\cdot m \cdot q \cdot q +  q^{m}\cdot q \cdot q = (m+1)q^{m+2}$, and $\nu$ is the distribution on $[M]$ implied by the description. 



Let $G_\bRM = G(h_\bRM)$ be the group that is presented from $h_\bRM$ (recall from \Cref{sec:pres-code} that codes $\code_\bRM=\ker h_\bRM$ give rise to group presentations through the general construction~\eqref{eq:def-gh-pres}). We do not know if $G_\bRM = \Z_2^k$, with $k=t(d+1)^m$. Instead we modify the presentation $G(h_\bRM)$ by adding pairwise commutation relations, as in~\eqref{eq:wtgh-intro}. Let \tnote{Removed $\{R^\sq_k\}$ below following reviewer suggestion}
\[ G(h_\bRM) \,=\, \big\langle x_1,\ldots,x_N \;:\;  \{R^\ld_k\}\,,\; \{R^\had_{k}\} \big\rangle\;,\]
where  $R^\ld_k$ ranges over all relations implied by the ``low-degree test'' in Figure~\ref{fig:bRM-tester} and $R^\had_{k}$ ranges over all the relations implied by the ``Hadamard test.'' For $k=(i,j)\in\{1,\ldots,N\}^2$ such that $i<j$ let $R^\com_k$ be the relation $[x_i,x_j]=e$, where $[a,b]=aba^{-1}b^{-1}$ is the group commutator.
Then we define 
\begin{equation}\label{eq:z2k-eff}
\tilde{G}:=\widetilde{G(h_\bRM) } \,=\,\big\langle x_1,\ldots,x_N \;:\;  \{R^\ld_k\}\,,\; \{R^\had_{k}\}\, , \; \{ R^\com_k\}\big\rangle\;.
\end{equation}
From Lemma~\ref{lem:com-code} it follows that $\tilde{G}$ is isomorphic to $\Z_2^k$. Our main result is an efficient stability result for this presentation. To state this we need to introduce distributions $\mu_S$ and $\mu_R$ on the generators and relations of $\tilde{G}$. The distribution $\mu_R$ is obtained as follows. With probability $1/3$ each, a relation from $\{R^\ld_k\}$ or $\{R^\had_k\}$ is chosen uniformly at random. With probability $1/3$, a random commutation relation from $R^\com_k$ is chosen according to the uniform mixture of the following two distributions:
\begin{enumerate}
\item For the first distribution, we select $u\in \F_q^m$ uniformly at random, $j\in\{1,\ldots,m\}$ uniformly at random, and $i\neq i'\in\{0,\ldots,d\}$ uniformly at random. Then select $\alpha,\beta\in \F_2^t$ uniformly at random and check commutation between $x_{u+t_i e_j,\alpha}$ and $x_{u+t_{i'} e_j,\beta}$, where we interpret the subscripts  $u+t_i e_j,\alpha$ and $u+t_{i'} e_j,\beta$ as corresponding to some integer in $[N]= (\F_2^t)^m \times \F_2^t$ (as sets!). 

\item For the second distribution, we first select $j\in\{1,\ldots,m\}$ and $u_{m-j+2},\ldots,u_m \in \F_q$ uniformly at random. Then select $v,v' \in \F_q^m$ uniformly at random, conditioned on the last $(j-1)$ coordinates of each vector matching $u_{m-j+2},\ldots,u_m$. Finally, select $\alpha,\beta\in \F_2^t$ uniformly at random and check commutation between $x_{v,\alpha}$ and $x_{v',\beta}$. 
\end{enumerate}
Having defined $\mu_R$, we define $\mu_S$ as in Remark~\ref{rk-mus}. It is easy to check that in this way we obtain that $\mu_S$ is the uniform distribution over $[N] $; this is because for any of the relations involved, the marginal distribution on any element appearing in the relation (e.g.\ the first element, the second, etc.) is uniform over $S$. The following is our main technical result. 

\begin{theorem}\label{thm:z2-stab}
Let $\mC$ be the class of all tracial von Neumann algebras. 
The presentation of  $\Z_2^k$ given in~\eqref{eq:z2k-eff}, together with the distributions $\mu_S$ and $\mu_R$ defined above, is $(\delta,\mC)$ stable with 
\begin{equation}\label{eq:th-min}
 \delta(\eps) = \min\Big\{ C (mdt)^{c_1} (\eps^{c_2}+(1/q)^{c_3})\ ,\ C'_{t,d,m}\ \eps\Big\}\;,
\end{equation}
where $C,c_1,c_2,c_3$ are universal positive constants and $C'_{t,d,m}$ depends on $t,d,m$ but not $\eps$.\footnote{We do not attempt to compute the dependence of $C'_{t,d,m}$ on $t,d,m$ because we are more interested in the regime of parameters where the first part of the min in~\eqref{eq:th-min} is relevant. Nevertheless, we include it so that $\delta(\eps)\to_{\eps\to 0} 0$ as required.}   
\end{theorem}

We briefly explain a possible setting of parameters in Theorem~\ref{thm:z2-stab}. For any integer $t\geq 1$, fix $q=2^t$ and let $d=m=c t^c$ for some constant $c>0$. Then $k=t(d+1)^m = 2^{\Theta((\log q)^c \log\log q)}$ and $N=q^{m+1}= 2^{\Theta((\log q)^{c+1} \log\log q)}$. Moreover, the number of relations in~\eqref{eq:z2k-eff} is $O(N^2)$, which scales as $2^{\poly\log k}$; and the maximum length of a relation is $d+2=O(\log k)$. Finally, with this choice of parameters the function $\delta(\eps)$ scales as $\min\{ \poly(\log k)\cdot (\poly(\eps) + \poly(1/k)),C'_k\eps\}$. When $\eps$ is much larger than $1/k$, the first term in the $\min$ scales as $\poly(\log k)\cdot \poly(\eps)$, which gives a favorable tradeoff between the presentation size and the modulus of stability. For very small $\eps$, the bound is what we would generically expect from Theorem~\ref{thm:gh} and the size of the presentation ; in particular the constant $C'_k$ necessarily scales exponentially with $k$.

\begin{proof}
Let $(\mM,\tau)$ be a tracial von Neumann algebra and $\phi$ be an $(\eps,\mu_R)$-homomorphism of $\langle S:R\rangle$ on $(\mM,\tau)$. 
 Here, $S = \{s_{u,a}: u\in (\F_2^t)^m, a\in \F_2^t\}$. We sometimes enumerate the items of $S$ as $S=\{x_i: i\in\{1,\ldots,N\}\}$, where $N=2^{tm+t}$ and we fixed an arbitrary bijection between $(\F_2^t)^m\times \F_2^t$ and $\{1,\ldots,N\}$. Let $R$ be the set of all relations in~\eqref{eq:z2k-eff}, i.e.\ $R=  \{R^\ld_k\}\cup\{R^\had_k\}\cup\{R^\com_k\}$. 

The second bound in the $\min$ in~\eqref{eq:th-min} follows immediately from Theorem~\ref{thm:gh} and Lemma~\ref{lem:stab-group} according to Remark~\ref{rk:linear-modulus}. Therefore, the bulk of the proof focuses on establishing the first bound. 
The proof strategy is to perform a reduction to~\cite[Theorem 4.1]{ji2022quantum}. Towards this, the main technical work in the proof consists in using the unitaries $\phi(s_{u,a})$ in order to define a synchronous strategy in the tensor code test from~\cite{ji2022quantum}, where the underlying code is the Reed-Solomon code with degree $d$ over $\F_q$. To define the synchronous strategy, we need ``points,'' ``lines,'' and ``pair'' measurements (see~\cite{ji2022quantum} for the terminology). Each of these is a family of projective measurements that obey certain constraints. 

The proof consists of a sequence of claims, which examine the constraints imposed on the $\phi(s_{u,a})$ by each of the four collections of relations in~\eqref{eq:z2k-eff} in turn.



We first exploit the relations $\{R^\had_k\}$ to show the following. Recall that $q=2^t$.

\begin{claim}\label{claim:z2-stab-2}
For every $u\in \F_q^m$ there is a projective measurement $\{P^u_\beta\}_{\beta\in \F_{q}}$ on $\mM$ such that 
\begin{equation}\label{eq:z2-stab-2}
 \Es{u\in \F_q^m} \Es{a\in \F_2^t} \Big\| \phi(s_{u,a}) - \sum_{\beta\in\F_q} (-1)^{a \cdot \kappa(\beta)} P^u_\beta \Big\|_\tau^2 \,=\, O(\eps)\;. 
\end{equation}
\end{claim}

\begin{proof}
Since $\mu_R$ places weight $1/4$ on relations $\{R^\had_k\}$ we deduce that 
\begin{equation}\label{eq:stab-rm-1}
\Es{u\in \F_q^m} \Es{a,b\in \F_2^t} \big\|\phi(s_{u,a})\phi(s_{u,b})\phi(s_{u,a+b})-\Id\big\|_\tau^2 \,\leq\, 4\eps\;. 
\end{equation}
\tnote{edited proof starting here:}Fix an $u\in \F_q^m$ and apply Corollary~\ref{cor:lin-test} for that $u$. This gives a projective measurement $\{P^u_\beta\}_{\beta\in \F_q}$ on $\mM$ such that (by the triangle inequality)
\begin{align*}
 \Es{a \in \F_2^t} \Big\| \phi(s_{u,a}) - \sum_{\beta\in\F_q} (-1)^{a \cdot \kappa(\beta)} P^u_\beta \Big\|_\tau^2
&= O(\eps_u)\;,
\end{align*}
where $\eps_u=\Es{a,b\in \F_2^t}\|\phi(s_{u,a})\phi(s_{u,b})\phi(s_{u,a+b})-\Id\|_\tau^2$.
Averaging over $u$ gives the desired result.
\end{proof}

For $u\in \F_q^m$ let $\{P^{u}_\beta\}_{\beta \in \F_q}$ be the projective measurement obtained from Claim~\ref{claim:z2-stab-2}. For $\alpha\in \F_q$, let  
\[ U_{u,\alpha} = \sum_{\beta\in\F_q} (-1)^{\tr(\alpha\beta)} P^u_{\beta}\;.\]
Then $U_{u,\alpha} \in \mU(\mM)$. 


The next claim uses the relations $\{R^\com_k\}$.

\begin{claim}\label{claim:z2-stab-2b}
For $u\in \F_q^m$, $\alpha\in\F_q$, $j\in\{1,\ldots,m\}$ and $i\in \{0,\ldots,d\}$ let $U_{i,\alpha} = U_{u+t_ie_j,\alpha}$. Then 
\begin{equation}\label{eq:z2-stab-2b-0a}
\Es{u\in \F_q^m} \Es{\substack{j\in\{1,\ldots,m\}\\i\neq i' \in \{0,\ldots,d\}}}\Es{\alpha,\alpha'\in \F_q} \big\| \big[ U_{i,\alpha}, U_{i',\alpha'}\big]-\Id\big\|_\tau^2\,=\, O({\eps})\;, 
\end{equation}
and
\begin{equation}\label{eq:z2-stab-2b-0b}
 \Es{u\in \F_q^m} \Es{j\in\{1,\ldots,m\}} \Es{\substack{v\in \F_q^m \\v_{m-j+2}=u_{m-j+2},\ldots,v_m=u_m}}\Es{\alpha,\alpha'\in \F_q} \big\| \big[ U_{u,\alpha}, U_{v,\alpha'}\big]-\Id\big\|_\tau^2\,=\, O({\eps})\;,
\end{equation}
where the expectation is over a uniformly random $u$ and $j$, and a uniformly random $v$ conditioned on its last $(j-1)$ coordinates matching those of $u$. 
\end{claim}

\begin{proof}
Due to the test of the relations $\{R^\com_k\}$, it holds that 
\begin{align}
\Es{u\in \F_q^m} \Es{\substack{j\in\{1,\ldots,m\} \\ i\neq i' \in \{0,\ldots,d\}}} \Es{a,b\in \F_2^t} \big\| [\phi(s_{u+t_i e_j,a}),\phi(s_{u+t_{i'} e_j,b})] - \Id \big\|_\tau^2&\leq 8\eps\; ,\label{eq:z2-stab-2b-1}\\
\Es{u\in \F_q^m} \Es{j\in\{1,\ldots,m\}} \Es{\substack{v\in \F_q^m \\v_{m-j+2}=u_{m-j+2},\ldots,v_m=u_m}}\Es{a,b\in \F_2^t} \big\| [\phi(s_{u,a}),\phi(s_{v,b})]-\Id\big\|_\tau^2&\leq 8\eps\;.\label{eq:z2-stab-2b-1b}
\end{align}
Now, for any $u,v\in \F_q^m$
\begin{align*}
 &\Es{\alpha,\alpha'\in \F_q}\tau\big( U_{u,\alpha} U_{v,\alpha'} U_{u,-\alpha} U_{v,-\alpha'} \big) \\
&= \sum_{\beta,\beta',\gamma,\gamma'\in \F_q}\Es{\alpha\in \F_q} (-1)^{\tr(\alpha(\beta-\beta'))} \Es{\alpha'\in \F_q} (-1)^{\tr(\alpha'(\gamma-\gamma'))}  \tau(P^{u}_\beta P^{v}_\gamma P^{u}_{\beta'} P^{v}_{\gamma'}\big)\\
&=\sum_{\beta,\beta',\gamma,\gamma'\in \F_q}\Es{a\in\F_2^t} (-1)^{a\cdot\kappa(\beta-\beta')} \Es{a'\in\F_2^t} (-1)^{a'\cdot\kappa(\gamma-\gamma')} \tau(P^{u}_\beta P^{v}_\gamma P^{u}_{\beta'} P^{v}_{\gamma'}\big)\\
&=\Es{a,a'\in\F_2^t}  \tau\Big( \Big(\sum_{\beta\in \F_q} (-1)^{a\cdot \kappa(\beta)} P^{u}_\beta\Big) \Big(\sum_{\gamma\in \F_q} (-1)^{a'\cdot \kappa(\gamma)} P^{v}_\gamma\Big)\Big(\sum_{\beta'\in \F_q} (-1)^{a\cdot \kappa(\beta')} P^{u}_{\beta'}\Big)\Big(\sum_{\gamma'\in \F_q} (-1)^{a'\cdot \kappa(\gamma')} P^{v}_{\gamma'}\Big)\Big)\;.
\end{align*}
It follows that, for any distribution on $(u,v)$ such that both marginals are uniform over $\F_q^m$,
\begin{align}
\Es{u,v} \Es{\alpha,\alpha'\in \F_q} \big\| [U_{u,\alpha} ,U_{v,\alpha'} ]-\Id\big\|_\tau^2
&= \Es{u,v}\Es{a,a'\in \F_2^t} \Big\| \Big[\sum_{\beta\in \F_q} (-1)^{a\cdot \kappa(\beta)} P^{u}_\beta ,\sum_{\beta'\in \F_q} (-1)^{a'\cdot \kappa(\beta')} P^{v}_{\beta'} \Big]-\Id\Big\|_\tau^2\notag\\
&= \Es{u,v}\Es{a,a'\in \F_2^t} \big\| \big[ \phi(s_{u,a}),\phi(s_{v,a'})\big]-\Id\big\|_\tau^2+O({\eps})\;,\label{eq:z2-stab-2b-1d}
\end{align}
where the first line is by expanding the square and using~\eqref{eq:z2-stab-2b-1d} and the second line follows from Claim~\ref{claim:z2-stab-2} and the triangle inequality. Thus~\eqref{eq:z2-stab-2b-0a} follows from~\eqref{eq:z2-stab-2b-1}, and~\eqref{eq:z2-stab-2b-0b} follows from~\eqref{eq:z2-stab-2b-1b}.
\end{proof}

Now we show the following, which essentially follows from the previous claim. 

\begin{claim}\label{claim:z2-stab-3}
For every $u\in \F_q^m$ and $j\in\{1,\ldots,m\}$, there is a projective measurement $\{R^{u,j}_\alpha\}_{\alpha\in\F_q^{d+1}}$ on $\mM$ such that 
\[ \Es{u\in \F_q^m} \Es{j\in\{1,\ldots,m\}} \Es{z_0,\ldots,z_{d}\in \F_q} \Big\| U_{0,z_0}\cdots U_{d,z_d} -  \sum_{\alpha} (-1)^{\tr(\sum z_i\alpha_i)} R^{u,j}_\alpha\Big\|_\tau^2 \,=\, \poly(d,\eps)\;.\footnote{Recall footnote~\ref{ft:poly} on page~\pageref{ft:poly} for the use of this notation.}\]
Similarly, for any $u,v\in \F_q^m$ there is a projective measurement $\{R^{u,v}_{\alpha,\beta}\}_{(\alpha,\beta)\in\F_q^{2}}$ on $\mM$ such that 
\[  \Es{u\in \F_q^m} \Es{j\in\{1,\ldots,m\}} \Es{\substack{v\in \F_q^m \\v_{m-j+2}=u_{m-j+2},\ldots,v_m=u_m}} \Es{y,z\in \F_q} \Big\| U_{u,z}  U_{v,y} -  \sum_{\alpha} (-1)^{ \tr(\alpha z + \beta y)} R^{u,v}_{\alpha,\beta}\Big\|_\tau^2 \,=\, \poly(\eps)\;.\]
\end{claim}

\begin{proof}
We show the first part only, as the second part is analogous. Fix an $u\in \F_q^m$ and a direction $j\in \{1,\ldots,m\}$. Define
$\psi: (\Z_2^{t})^{d+1} \to \mU(\mM)$ by 
\[\psi(z_0,\ldots,z_{d}) \,=\, U_{0,z_0} \cdots U_{d,z_{d}}\;,\]
where for each $i\in\{0,\ldots,d\}$ and $z_i\in \F_q$, $U_{i,z_i}$ is the unitary defined in Claim~\ref{claim:z2-stab-2b} and we slightly abused notation to identify $z_i\in \Z_2^t$ with the unique $\tilde{z}_i\in \F_q$ such that $\kappa(\tilde{z}_i)=z_i$. Using Claim~\ref{claim:z2-stab-2b} and the triangle inequality we verify that the map $\psi$ is a $\delta=O(d^4\sqrt{\eps})$-approximate homomorphism of $\Z_2^{t(d+1)}$ on $\mU(\mM)$, with respect to the multiplication table presentation. To verify this first note that for any unitaries $V,W$, and any $i,i'$ we have that on average over $u$ and $j$,
\begin{align*}
\Es{\alpha,\alpha'\in \F_q} \big\| V U_{i,\alpha} U_{i',\alpha'} W^* - V U_{i',\alpha'} U_{i,\alpha}  W^*\big\|_\tau^2 \,=\, O\big(d^2 \sqrt{\eps}\big)\;,
\end{align*}
by Claim~\ref{claim:z2-stab-2b}, where the factor $d^2$ is because we require the relation to hold for all $i,i'$. Moreover, $U_{i,\alpha}U_{i,\alpha'}=U_{i,\alpha+\alpha'}$ by definition. Applying these relation $O(d^2)$ times gives
\begin{align*}
 \Es{z_0,\ldots,z_d\in\F_2^t}\Es{y_0,\ldots,y_d\in\F_2^t}\big\|\psi(z_0,\ldots,z_d)\psi(y_0,\ldots,y_d) -\psi(z_0+y_0,\ldots,z_d+y_d)\big\|_\tau^2 \,=\, O\big(d^4 \sqrt{\eps}\big)\;,
\end{align*}
as desired. Thus we may apply Corollary~\ref{cor:lin-test}
to obtain a projective measurement $\{R^{u,j}_\alpha\}_{\alpha\in \F_q^{d+1}}$  on $\mM$ such that 
on average over $u$ and $j$,
\[  \Es{z_0,\ldots,z_{d}\in\F_q} \Big\| \psi(z_0,\ldots,z_{d}) -  \sum_{\alpha\in \F_q^{d+1}} (-1)^{\tr(\sum z_i\alpha_i)} R_\alpha^{u,j}\Big\|_\tau^2 \,=\, \poly(d,\eps)\;,\]
as desired. 
\end{proof}

Finally we exploit the relations $\{R^\ld_k\}$ to obtain the following. 

\begin{claim}\label{claim:z2-stab-5}
Use the same notation as in Claim~\ref{claim:z2-stab-2b}. Let $u\in \F_q^m$, $\ell$ an axis-parallel line through $u$ and $v\in\ell$. Let $(\alpha_{u,v,i})_{i=0,\ldots,d}$ be the interpolation coefficients defined in~\eqref{eq:interp-coeff}. 
Then 
\begin{equation}\label{eq:z2-stab-5-0}
\Es{u\in\F_q^m} \Es{\ell: u\in \ell} \Es{v\in \ell} \Es{\gamma\in\F_q} \Big\| U_{0,\gamma \alpha_{u,v,0}}\cdots U_{d,\gamma  \alpha_{u,v,d}} - U_{v,\gamma} \big\|_\tau^2 \,=\, O\big(d\sqrt{\eps}\big)\;,
\end{equation}
where the expectation is over a uniformly random axis-parallel line $\ell\subset\F_q^m$ that contains $u$, and a uniformly random $v\in\ell$. 
\end{claim}

\begin{proof}
For $u,v\in \F_q^m$ and $i\in\{0,\ldots,d\}$ we write $\alpha_i$ for $\alpha_{u,v,i}$, and for $u\in\F_q^m$, $\beta\in\F_q$, $i\in\{0,\ldots,d\}$ and $j\in\{1,\ldots,m\}$, we let $P^{(i)}_\beta = P^{u+t_ie_j}_\beta$ (leaving the dependence on $u$ and $j$ implicit).
We observe that 
\begin{align*}
\Es{\gamma\in\F_q} \tau\big( U_{0,\gamma\alpha_{0}}\cdots U_{d,\gamma\alpha_{d}} U_{v,-\gamma} \big)
&=  \sum_{\beta_0,\ldots,\beta_d\in\F_q}\sum_{\beta\in\F_q} \Es{\gamma\in\F_q} (-1)^{\tr(\sum \gamma \alpha_i\beta_i)} (-1)^{\tr(-\gamma\beta)} \tau\big( P^{(0)}_{\beta_0} \cdots P^{(d)}_{\beta_d} P^v_\beta \big)\\
&= \sum_{\beta_0,\ldots,\beta_d\in\F_q}\sum_{\beta\in\F_q}\Es{\gamma\in\F_2^t} (-1)^{\gamma\cdot \kappa(\sum \alpha_i\beta_i-\beta)} \tau\big( P^{(0)}_{\beta_0} \cdots P^{(d)}_{\beta_d} P^v_\beta \big)\\
&=\Es{\gamma\in\F_2^t}  \tau\Big( \Big(\sum_{\beta_0\in \F_q} (-1)^{\gamma\cdot \kappa(\alpha_0\beta_0)} P^{(0)}_{\beta_0} \Big)\cdots\Big(\sum_{\beta\in \F_q} (-1)^{-\gamma\cdot \kappa(\beta)} P^{v}_{\beta} \Big)\Big)\\
&=\Es{\gamma\in\F_q}  \tau\Big( \Big(\sum_{\beta_0\in \F_q} (-1)^{\kappa( \gamma \alpha_0)\cdot\kappa(\beta_0)} P^{(0)}_{\beta_0} \Big)\cdots\Big(\sum_{\beta\in \F_q} (-1)^{\kappa(-\gamma)\cdot \kappa(\beta)} P^{v}_{\beta} \Big)\Big)\;.
\end{align*}
By repeated application of Claim~\ref{claim:z2-stab-2}, the last expression is, on average over $u$, within $O(d\sqrt{\eps})$ of 
\[ \Es{\gamma\in\F_q}  \tau\big( \phi(s_{u+t_ie_j,\kappa(\gamma \alpha_0)})\cdots \phi(s_{v,\kappa(-\gamma)}) \big)\;.\]
The latter expression is a random relation in $R^\ld_k$, and so, on average over $u,j$ and $v$ it is $1-O(\eps)$.
\end{proof}

\begin{claim}\label{claim:z2-stab-4}
For every axis-parallel line $\ell$ there is a projective measurement $\{Q^\ell_g\}$ with outcomes $g \in \cal{P}(q,1,d)$ that range over degree-$d$ polynomials on $\ell$ such that
\[ \Es{\ell\subset \F_q^m} \Es{v\in \ell} \sum_g \tau\big( P^v_{g(v)} Q^\ell_g\big) \,\geq\, 1-\poly(d,\eps)\;, \]
where the expectation is over a uniformly random axis-parallel line $\ell$ and point $v\in \ell$. 
\end{claim}

\begin{proof}
First we show that the conclusion of Claim~\ref{claim:z2-stab-3} can be strengthened to hold for \emph{every} $z_0,\ldots,z_{d}$, instead of on average. To show this, note that for any $y_0,\ldots,y_d\in \F_q$ we can write 
\[ U_{0,z_0}\cdots U_{d,z_d}\,=\, U_{0,y_0}U_{0,z_0-y_0}\cdots U_{d,y_d}U_{d,z_d-y_d}\;.\]
By repeated application of Claim~\ref{claim:z2-stab-2b}, the term on the right-hand side satisfies 
\[ \big\| U_{0,y_0}U_{0,z_0-y_0}\cdots U_{d,y_d}U_{d,z_d-y_d}  - U_{0,y_0}\cdots U_{d,y_d}U_{0,z_0-y_0}\cdots U_{d,z_d-y_d}\big\|_\tau^2 \,=\,\poly(d,\eps)\;.\]
To show this it suffices to verify that we only need to ``commute'' pairs of terms whose exponents are independent and uniformly random. Using Claim~\ref{claim:z2-stab-3} twice, the right-hand side satisfies 
\[ \big\| U_{0,y_0}\cdots U_{d,y_d}U_{0,z_0-y_0}\cdots U_{d,z_d-y_d} - \Big(\sum_{\alpha} \omega^{\tr(\sum y_i\alpha_i)} R^{u,j}_\alpha\Big)\Big(\sum_{\alpha} (-1)^{\tr(\sum (z_i-y_i)\alpha_i)} R^{u,j}_\alpha\Big) \big\|_\tau^2 \,=\,\poly(d,\eps)\;.\]
Since $\{R^{u,j}_\alpha\}$ is a projective measurement, we get the desired conclusion: on average over $u$ and $j$, for any $z_0,\ldots,z_d\in\F_q$, it holds that
\begin{equation}\label{eq:z2-stab-4-1}
 \Big\| U_{0,z_0}\cdots U_{d,z_d} -  \sum_{\alpha} (-1)^{\tr(\sum z_i\alpha_i)} R^{u,j}_\alpha\Big\|_\tau^2 \,=\, \poly(d,\eps)\;.
	\end{equation}
	For any line $\ell \subset \F_q^m$, let $u_\ell\in \ell$ be chosen such that, for a uniformly random $\ell$ and conditioned on that $u_\ell$,~\eqref{eq:z2-stab-4-1} and~\eqref{eq:z2-stab-5-0} both hold, with right-hand side multiplied by a factor at most $2$. For any $\ell$ in the $j$-th direction and degree-$d$ polynomial $g$, define the operator $Q^\ell_g = R^{u_\ell,j}_{g(u_\ell+t_0 e_j),\ldots,g(u_\ell+t_de_j)}$.
	Combining the two equations we deduce
	\begin{equation}\label{eq:z2-stab-4-2}
 \Es{\gamma\in\F_q} \Big\| U_{v,\gamma} -  \sum_{g} (-1)^{\tr(\gamma \sum \alpha_{u_\ell,v,i} g(u_\ell+t_ie_j))} Q^{\ell}_g\Big\|_\tau^2 \,=\, \poly(d,\eps)\;.
	\end{equation}
	 By definition, $\sum \alpha_{u_\ell,v,i} g(u_\ell+t_ie_j) = g(v)$. By Fourier transform, we obtain the desired conclusion. 
\end{proof}

We are now in a position to apply~\cite[Theorem 4.1]{ji2022quantum}. For this we need to define a synchronous strategy in the tensor code test $\code^{\otimes m}$, where $\code$ is the Reed-Solomon code with degree $d$ over $\F_q$, and thus $\code^{\otimes m}$ is the code $\code_\RM$ considered in Section~\ref{sec:rmq} (see the remark right after Lemma~\ref{lem:schwartz-zippel}). For the ``points measurement'' $A^u$ we choose $P^u$. For the ``lines measurement'' $B^\ell$ we choose $Q^\ell$ from Claim~\ref{claim:z2-stab-4}. Finally, for the ``pair measurement'' $P^{u,v}$ we choose $R^{u,v}$ from Claim~\ref{claim:z2-stab-3}. By Claim~\ref{claim:z2-stab-4} this strategy succeeds with probability $1-\poly(d,\eps)$ in the ``axis-parallel lines test'', and by Claim~\ref{claim:z2-stab-3} it succeeds with probability $1-\poly(\eps)$ in the ``subcube commutation test.'' Applying~\cite[Theorem 4.1]{ji2022quantum} we deduce the existence of a projective measurement $\{G_c\}_{c\in\code^{\otimes m}}$ on $\mA$ such that for all integers $r \geq 12mt$,
\begin{equation}\label{eq:41-last}
 \Es{u\in\F_q^m} \sum_{c\in\code^{\otimes m}} \tau\big( G_c P^u_{c(u)}\big) \geq 1-\eta\;,
\end{equation}
where $\eta = \poly(m,d,r) \cdot\poly(\eps,q^{-1},e^{-\Omega(r/m^2)})$. Choosing $r=\Omega(m^2 t)$, since $q=2^{t}$, the bound becomes $\poly(m,d,t) \cdot\poly(\eps,q^{-1})$. The first bound in the $\min$ in~\eqref{eq:th-min} then follows where explicitly, the homomorphism that is close to $\phi$ is given by
\[ g: s_{u,a} \mapsto \sum_{\beta\in\F_q} (-1)^{a\cdot \kappa(\beta)} \sum_{c:\ c(u)=\beta} G_c \;,\]
and closeness follows from~\eqref{eq:41-last} and~\eqref{eq:z2-stab-2} by expanding the square in the definition.
\end{proof}

\section{Testing entanglement}
\label{sec:quantum}

In this section we give an application of our stability results to the problem of entanglement testing in quantum information. We first introduce the language of nonlocal games. Then we associate a nonlocal game to any presentation of $\Z_2^k$. Finally, we show that, if the presentation is stable, then the game is a robust  entanglement test. 
	
\subsection{Nonlocal games}
\label{sec:nl-games}

We give standard definitions on nonlocal games. For background from a computer science point of view, see~\cite{cleve2004consequences}; for the operator algebra perspective, see e.g.~\cite{kim2018synchronous}.

\begin{definition}[Game]
A game is a tuple $(\mX,\mu,\mA,D)$ where $\mX$ is a finite set, $\mu$ a distribution on $\mX\times \mX$, $\mA=(\mA(x))_{x\in\mX}$ a collection of finite sets, and 
\[ D: \big\{ (x,y,a,b) : (x,y)\in\text{supp}(\mu),a\in\mA(x),b\in\mA(y)\big\} \;\to\;\{0,1\}\]
such that $D$ is symmetric, i.e. $D(x,y,a,b)=D(y,x,b,a)$ whenever both terms are defined. We often abuse notation and write $\mu$ for the symmetrized marginal of $\mu$, i.e.\ 
\[\mu(x) := \sum_{x'\in \mX} \frac{1}{2}\big(\mu(x,x')+\mu(x',x)\big)\;.\]
\end{definition}

The interpretation of $G=(\mX,\mu,\mA,D)$ as a nonlocal game is the following. In the ``game,'' a referee is imagined to sample a pair of ``questions'' $(x,y)\sim \mu$. The question $x$ is sent to a first player, ``Alice,'' and the question $y$ is sent to a second player, ``Bob.'' Each player is tasked with responding with an answer, $a\in \mA(x)$ for Alice and $b\in \mA(y)$ for Bob. The referee accepts the players' answers if and only if $D(x,y,a,b)=1$. 

Nonlocal games provide a framework to study different kinds of bipartite correlations: depending on the level of coordination allowed between Alice and Bob, they may have varying chances of success in the game. A ``classical'' strategy consists of functions $f_A:\mX\to\mA$ for Alice and $f_B:\mX\to\mA$ for Bob; any such pair of functions leads to a probability of success in the game which can be computed in the obvious manner. 

An important motivation for studying nonlocal games is that in quantum mechanics, local strategies for the players (meaning strategies that do not require any communication between the players to determine their answer, given their question) are a larger set than the above-described classical strategies. Specifically, a \emph{quantum local} (or \emph{quantum} for short) strategy is specified by the following. 

\begin{definition}[Synchronous strategy]
If $G=(\mX,\mu,\mA,D)$ is a game and $(\mM,\tau)$ a tracial von Neumann algebra, a \emph{synchronous strategy $\strategy$ for $G$ on $(\mM,\tau)$} is, for every $x\in \mX$, a projective measurement $(P^x_a)_{a\in \mA(x)}$ on $\mM$. The value of a strategy $\strategy$ in $G$ is 
\[ \omega(G;\strategy)\,=\, \sum_{(x,y)\in\mX\times\mX}\mu(x,y)\sum_{(a,b)\in\mA(x)\times\mA(y)} D(x,y,a,b)\, \tau\big(P^x_a \,P^y_b\big) \;.
\]
We say that $\strategy$ is \emph{perfect} if $\omega(G;\strategy)=1$.
\end{definition}
	
The name \emph{synchronous} stems from the fact that whenever an identical pair $(x,x)$ is chosen, $\tau(P^x_a P^x_b)=0$ for $a\neq b$ due to the requirement that $\{P^x_a\}_a$ is a projective measurement. Thus a synchronous strategy always returns the same answer to the same question. 

It will be convenient to have a measure of closeness for strategies. The following definition parallels Definition~\ref{def:close}. 

	\begin{definition}[Closeness for strategies]\label{def:close-meas}
Let $\{A^i_a\}\subseteq \mM$ and $\{B^i_a\}\subseteq \mN$ be two families of projective measurements on  tracial algebras $(\mM,\tau^\mM)$ and $(\mN,\tau^\mN)$ respectively, indexed by the same set $i\in \mI$ and with the same sets of outcomes $a,b\in\mA(i)$. For $\delta\geq0$ and $\mu$ a measure on $\mI$ we say that $\{A^i\}$ and $\{B^i\}$ are $(\delta,\mu)$-close if there exists a projection $P\in\mM_\infty$ of finite trace such that $\mN=P\mM_\infty P$ and $\tau^\mN=\tau_\infty/\tau_\infty(P)$, and a partial isometry $w\in P \mM_\infty \Id_\mM$ such that 
\[ \Es{i\sim\mu} \sum_{a\in\mA(i)} \big\| A^i_a - w^* B^i_a w \big\|_{\tau^\cM}^2 \,\leq\,\delta\]
and 
\[\max\big\{ \tau^\mM(\Id_\mM-w^*w)\,,\; \tau^\mN(P-ww^*)\big\} \,\leq\, \delta\;.\]
If the measure $\mu$ is omitted then it is understood to be the uniform measure on $\mI$.
\end{definition}

In Definition~\ref{def:close-meas} closeness is measured in the $L_2$ sense. The following lemma gives a consequence for distance measured in an $L_1$ sense.

\begin{lemma}\label{lem:l1-l2}
Let $\{P_a^i\}$ and $\{Q_a^i\}$ be two families of projective measurements that are $(\eps,\mu)$-close. Then 
\begin{equation}\label{eq:l1-l2}
\Es{i\in\mI} \sum_{a\in \mA(i)} \tau\big(|P_a-w^*Q_aw|\big) \,\leq\, \eps+ 2 \sqrt{\eps}\;.
\end{equation}
\end{lemma}

\begin{proof}
Using the triangle inequality, 
\begin{align}
\Es{i} \sum_a \tau\big(|P^i_a - w^* Q^i_a w|\big) &\leq\, \Es{i}\sum_a \Big(\tau\big(|(P^i_a-(P^i_a)^2)|\big) + \tau\big(|P_a(P^i_a-w^*Q^i_aw)|\big) \notag\\
&\qquad\qquad+ \tau\big(|(P^i_a-w^*Q^i_aw)w^*Q^i_aw|\big) \notag \\
&\qquad \qquad+\tau\big(| (w^*Q^i_aw^*wQ^i_aw-w^*Q^i_aw)|\big)\Big)\;. \label{eq:lem:l1-l2}
\end{align}
The first term on the right-hand side is zero, because $P^i$ is assumed projective. The terms in the middle are bounded using H\"older's inequality:
\begin{align*}
\Es{i}\sum_a  \tau\big(|P^i_a(P^i_a-w^*Q^i_aw)|\big) &\leq \Es{i} \sum_a \|P^i_a\|_\tau \, \, \|P^i_a-w^*Q^i_aw\|_\tau\\
&\leq \Big(\Es{i}\sum_a \|P^i_a\|_\tau^2\Big)^{1/2}\Big( \Es{i}\sum_a  \|P^i_a-w^*Q^i_aw\|_\tau^2 \Big)^{1/2}\\
&\leq \sqrt{\eps}
\end{align*}
by closeness. The third term of~\eqref{eq:lem:l1-l2} is bounded in a similar fashion.
Finally the last term of~\eqref{eq:lem:l1-l2} can be bounded as 
\begin{align*}
\Es{i}\sum_a \tau\big(| w^*Q^i_aw^*wQ^i_aw-w^*Q^i_aw|\big) &= \Es{i}\sum_a \tau\big( (w^*Q^i_a(I - w^*w)Q^i_aw)\big)\\
&=  \tau\Big( (I-w^*w)\Big(\Es{i}\sum_a Q^i_aww^* Q^i_a\Big)\Big)\\
&\leq \tau(I-w^*w) \,\, \Big\|\Es{i}\sum_a Q^i_aww^* Q^i_a\Big\|\\
&\leq \eps\;.
\end{align*}
\end{proof}

We will make use of the following elementary lemma, which shows that two strategies for the same game $G$ that are close according to Definition~\ref{def:close-meas} have a close value. 

\begin{lemma}\label{lem:close-value}
Let $G=(\mX,\mu,\mA,D)$ be a game and $\strategy=\{P^x_a\}$ and $\strategy'=\{Q^x_a\}$ strategies on $\mM$ and $\mN$ respectively such that $\{P^x_a\}$ and $\{Q^x_a\}$ are $(\eps,\mu)$-close. Then 
\[ \big|\omega(G;\strategy) - \omega(G;\strategy')\big|\,\leq\, 8\sqrt{\eps} \;.\]
\end{lemma}

The lemma is well-known, see e.g.~\cite[Lemma 5.28]{ji2020mip}. For convenience we include the proof. 

\begin{proof}
By definition, there exists a projection $P \in \mM_\infty$ and a partial isometry $w \in P \mM_\infty I_\mM$ such that $\mN = P \mM_\infty P$ and
\begin{align}
\omega(G;\strategy') &= \Es{(x,y)\sim\mu} \sum_{a,b} D(a,b,x,y)  \tau\big( Q^x_a \, Q^y_b \big)\notag\\
&=  \Es{(x,y)\sim\mu} \sum_{a,b} D(a,b,x,y)  \tau\big( w^* Q^x_a w\, w^* Q^y_b w\big)\notag\\
&\qquad\qquad+  \Es{(x,y)\sim\mu} \sum_{a,b} D(a,b,x,y)  \tau\big( w^* Q^x_a (P - w\, w^* ) Q^y_b w\big) \notag\\
&\qquad\qquad+ \Es{(x,y)\sim\mu} \sum_{a,b} D(a,b,x,y)  \tau\big(Q^x_a  Q^y_b (P-ww^*)\big)~.\label{eq:cv-0}
\end{align}
Here we used the fact that $Q^x_a \in \mN$ so $P Q^x_a P = Q^x_a$ for all $x,a$. 
For the second term on the right-hand side, we write
\begin{align}
\Big|\Es{(x,y)\sim\mu} \sum_{a,b} D(a,b,x,y) & \tau\big( w^* Q^x_a (P - w\, w^* ) Q^y_b w\big)\Big|
\leq \Es{(x,y)\sim\mu} \Big | \tau\big( \sum_{a,b} D(a,b,x,y)   Q^y_b w w^* Q^x_a (P - w\, w^* ) \big) \Big | \;.\label{eq:XX1}
\end{align}
We then apply the Cauchy-Schwarz inequality
\begin{align}
\eqref{eq:XX1}&\leq \sqrt{ \tau((P - ww^*)^2)} \cdot \sqrt{\Es{(x,y)\sim\mu} \tau\big( \big| \sum_{a,b} D(a,b,x,y)   Q^x_a w w^* Q^y_b \big|^2 \big) } \notag\\
&\leq \sqrt{\eps}\sqrt{ \Es{(x,y)\sim\mu} \tau\big( \sum_{a,b} D(a,b,x,y)   Q^y_b ww^* Q^x_a w w^* Q^y_b \big) } \notag\\
&\leq \sqrt{\eps}\sqrt{ \Es{(x,y)\sim\mu}  \sum_{a,b}  \tau\big(Q^y_b ww^* Q^x_a w w^* Q^y_b \big) } \notag \\
&=\sqrt{\eps}\sqrt{ \Es{(x,y)\sim\mu}  \sum_{b}  \tau\big(Q^y_b ww^* \big) } \notag\\
&\leq \sqrt{\eps}\;.\label{eq:cv-1}
\end{align}
The second line uses that $P - ww^*$ is a positive operator with operator norm at most $1$ for the first term, and that for fixed $x,y$, the measurements $\{ Q^x_a \}_a$ and $\{Q^y_b\}_b$ are projective. The third line is due to $D(x,y,a,b) \in \{0,1\}$. The fourt line is due to $\sum_a Q^x_a = I_\mN$. The fifth line is due to $\sum_a Q^y_b = I_\mN$. The third term on the right-hand side of~\eqref{eq:cv-0} is bounded in the same manner. 

For $x\in \mX$ such that $\mu(x)\neq 0$, denote by $\mu_x$ the (symmetrized) conditional distribution $\mu_x(y)=\frac{1}{2}(\mu(x,y)+\mu(y,x))/\mu(x)$. Then using~\eqref{eq:cv-0} and the preceding bounds,
\begin{align*}
\big|\omega(G;\strategy) - \omega(G;\strategy')\big|
 &\leq 2\sqrt{\eps}+ \Big|\Es{(x,y)\sim\mu} \sum_{a,b} D(a,b,x,y) \big( \tau\big( P^x_a \, P^y_b \big)-\tau\big( w^* Q^x_a w\, w^* Q^y_b w\big)\big)\Big|\\
&\leq  2\sqrt{\eps}+ \Big|\Es{x\sim\mu} \sum_{a}  \tau\Big( \big(P^x_a-w^* Q^x_aw\big) \, \Big( \Es{y\sim \mu_x} \sum_b D(a,b,x,y) P^y_b \Big)\Big)\Big|\\
&\qquad+ \Big|\Es{y\sim\mu} \sum_{b}  \tau\Big(\Big( \Es{x\sim \mu_y} \sum_a D(a,b,x,y) w^* Q^x_a w\Big) \big(P^y_b-w^*Q^y_bw\big) \, \Big)\Big|\\
&\leq 2\sqrt{\eps}+2\Es{x\sim\mu} \sum_{a}  \tau\big(\big| P^x_a-w^*Q^x_aw\big|\big) \\
&\leq 2\sqrt{\eps}+2(\eps + 2\sqrt{\eps}) \leq 8\sqrt{\eps}\;,
\end{align*}
where the third line uses $\tau(AB)\leq\tau(|A|)\|B\|$ and the last line is by Lemma~\ref{lem:l1-l2}.
\end{proof}
	
\subsection{Some simple games}

We introduce games previously used in the literature, that will serve as building blocks. For details on the implementation of these games we refer to~\cite{de2022spectral}.

\paragraph{Commutation game.}
We call \emph{commutation game} the game $G_{\cc}=(\mX_\cc,\mu_\cc,\mA_\cc,D_\cc)$ defined in~\cite[Section 3.1]{de2022spectral}. For convenience we change the notation slightly and denote $x_{X,0}, x_{Z,0} \in \mX_{\cc}$ the two special questions, $x_{\cc,1}$ and $x_{\cc,2}$ respectively. 

\paragraph{Anti-commutation game.}
We call \emph{anti-commutation game} the game $G_\ac=(\mX_\ac,\mu_\ac,\mA_\ac,D_\ac)$ defined in~\cite[Section 3.2]{de2022spectral}. For convenience we change the notation slightly and  denote $x_{X,1}, x_{Z,1} \in \mX_{\ac}$ the two special questions, $x_{\ac,1}$ and $x_{\ac,2}$ respectively. 

\paragraph{Braiding game.}
Finally we introduce a game, which we denote $G_{\dlS}$, that is built from the previous two games and is a based on a more general class of games analyzed in~\cite[Section 3.4]{de2022spectral}.  Using notation from~\cite{de2022spectral}, the game $G_{\dlS}$ is obtained by making the following choices. The game can be constructed from any $E\in\F_2^{k\times n}$. Let 
\begin{equation}\label{eq:dls-sets}
 S_X\,=\,S_Z\,=\,\{E e_i:\,i\in\{1,\ldots,n\}\}\subseteq \field^k\;,
\end{equation}
and let $\mu_{\dlS}$ be the uniform distribution over $\Omega=S_X\times S_Z$. Let $\alpha,\beta$ be the coordinate projections on $\Omega$. Let $\Omega_+ = \{(a,b)\in \Omega:a\cdot b=0\}$ and $\Omega_-=\{(a,b)\in\Omega:a\cdot b=1\}$. Then $G_{\dlS}=(\mX_{\dlS},\mu_{\dlS},\mA_{\dlS},D_{\dlS})$ has question set 
\[ \mX_{\dlS} = \{X,Z\} \cup (\mX_{\cc}\times \Omega_+) \cup (\mX_{\ac} \times \Omega_-) \cup (\{X\}\times S_X) \cup (\{Z\}\times S_Z)\;.\]
The game itself is based on the game described in~\cite[Section 3.4]{de2022spectral}, with a small modification (the addition of the consistency test). For clarity we recall the entire game in Figure~\ref{fig:dlS}. The sets $\mA_{\dlS}$ and the function $D_{\dlS}$ can be inferred from the figure. 

\begin{figure}[!htbp]
  \centering
  \begin{gamespec}
Let $S_X,S_Z\subseteq \field^k$.  Sample $\omega = (\omega_X,\omega_Z)\in S_X \times S_Z $ uniformly at random. Let $\gamma = \omega_X \cdot \omega_Z \in \field$. Execute either of the following tests with probability $1/3$ each. 
    \begin{enumerate}
      \setlength\itemsep{1pt}
    \item \textbf{(Anti-)commutation test:} 
		\begin{enumerate}
		\item If $\gamma=0$ then sample a pair of questions $(x_c,y_c)\sim\mu_\cc$ as in the commutation game. If $x_c = x_{W,0}$ (resp.\ $y_c=x_{W,0}$) for some $W\in \{X,Z\}$ then send $(W,\omega_W)$ to $\alice$ (resp.\ $(W,\omega_W)$ to $\bob$). Otherwise, send $(x_c,\omega)$ to $\alice$ (resp.\ $(y_c,\omega)$ to $\bob$). Accept if and only if their answers are accepted in the commutation game. 
		\item If $\gamma\neq 0$ then do the same but for the anti-commutation game. 
		\end{enumerate} 
		 \item \textbf{Consistency test:} Select $W\in\{X,Z\}$ uniformly at random. Send $W$ to $\alice$ and $(W,\omega_W)$ to $\bob$. Receive $a\in \field^k$ and $b\in \field$ respectively. Accept if and only if $a\cdot \omega_W=b$. 
    \end{enumerate}
  \end{gamespec}
  \caption{The game $G_{\dlS}$ checks (anti)commutation relations between two collections of observables.}
  \label{fig:dlS}
\end{figure}

We recall the following result from~\cite{de2022spectral} about this game. (The result from~\cite{de2022spectral} is more general, and applies to non-uniform measures on $\Omega$. We only need the consequence stated here.) Before stating the result, we introduce the notion of a \emph{qubit test}. 

For a projective measurement $P = \{P_a\}_{a \in \F_2^k}$, we use $\widehat{P}(b)$ to denote the observable
\[\widehat{P}(b) \,=\, \sum_a (-1)^{a\cdot b} P_a\;.\]
For binary outcome measurements $P = \{ P_0, P_1\}$, we write $\widehat{P}$ to denote $P_0 - P_1$. Recall the notation $\sigma^W$ for the Pauli observables introduced in Section~\ref{sec:examples}.

\begin{definition}[Qubit test]
Let $k\in \N$ and $\delta:[0,1]\to\R_+$. 
A \emph{$(k,\delta(\eps))$-qubit test} is a synchronous game $G=(\mX,\mu,\mA,D)$ such there are two sets $S_X,S_Z\subseteq \field^k$ that each span $\field^k$ and an injection $\phi:(\{X\}\times S_X) \cup (\{Z\}\times S_Z) \to \mX$ such that $\mA(\phi({X},a))=\mA(\phi({Z},b))=\field$ for all $a\in S_X$, $b\in S_Z$ and such that the following holds:
\begin{itemize}
\item (Completeness:) There is a synchronous strategy $\strategy = \{P^{x}\}_{x \in \mX}$ for $G$ on $\mM=M_{2^{k}}(\C)\otimes \mH$, for some Hilbert space $\mH$, that succeeds with probability $1$ in $G$ and is such that $\widehat{P}^{\phi({W},a)} = \sigma^W(a)\otimes \Id_\mH$ for every $W\in\{X,Z\}$ and $a\in S_W$. 
\item (Soundness:) Let $\mu'$ denote the (renormalized) restriction of (the marginal of) $\mu$ to the image of $\phi$ in $\mX$. \tnote{Added this condition as otherwise I think that a $k$-qubit test could easily be a $1$-qubit test in diguise:}Then there are bases $\{e_{X,i}\}\subseteq S_X$ and $\{f_{X,i}\}\subseteq{S_Z}$ such that $\mu'(\phi((X,e_{X,i}))),\mu'(\phi((Z,e_{Z,i})))\geq C/k$ for all $i$ and some constant $C>0$, and moreover the following holds. 
Any synchronous strategy  $\strategy = \{P^{x}\}_{x \in \mX}$ on $(\mM,\tau)$ for $G$ that succeeds with probability $1-\eps$ for some $\eps\geq 0$ is $(\delta(\eps),\mu')$-close to a strategy on some algebra $(M_{2^{k}}(\C)\otimes \mN,\tr\otimes \tau')$ where $(\mN,\tau')$ is a tracial sub-algebra of $\mM_\infty$ and such that
\[\widehat{P}^{\phi({W},a)} = \sigma^W(a)\otimes \Id_\mN\;.\]
\end{itemize}
\end{definition}

The terminology ``qubit test'' is motivated by the notion of a \emph{qubit} as introduced in~\cite{chao2017overlapping}. Informally, a qubit is a copy of the space $\C^2$ together with a pair of anticommuting observables $X,Z$ acting on it. A qubit test is then a test that forces any successful strategy for the players in it to ``contain'', as a subset of its measurement operators, a representation of (generators of) the algebra of $k$ qubits. 

\begin{theorem}[Corollary 3.9 in~\cite{de2022spectral}]\label{thm:dls-braid}
Suppose that $E\in \F_2^{k\times n}$ is such that the rows of $E$ span an $[n,k,d]_2$ linear code. The game $G_{\dlS}$ is a $(k,O(\eps))$-qubit test, where the sets $S_X$ and $S_Z$ are as in~\eqref{eq:dls-sets} and $\phi((X,a))=(X,a)$ and $\phi((Z,b))=(Z,b)$ and the $O(\eps)$ hides a (linear) dependence on $k/n$ and a (quadratic) dependence on $d/n$.  
\end{theorem}

The goal in the remaining sections is to design a qubit test where the size of the question set is $\polylog(k)$, as opposed to $\Omega(k^2)$ here. 

\subsection{The code game}
\label{sec:code-game}

In this section we associate a game $G_{\code,M}$ to any $[n,k,d]_2$ code $\code$ and $r$-local tester $M=(h,\nu)$ for it. In the game, one player is asked to provide an assignment to all generators in the support of a randomly chosen row of $h$, such that this assignment satisfies the check enforced by that row. The other player is asked to provide an assignment to a single of these variables, and checked for consistency with the first player. The formal definition follows.  

\begin{definition}
Let $\code$ be an $[n,k,d]_2$ linear code and $M=(h,\nu)$ an $r$-local tester for $\code$ such that $h\in \F_2^{m\times n}$. The game $G_{\code,M}$ is defined as follows. We set 
\[\mX = \big\{ \{\eq\}\times\{1,\ldots,m\} \sqcup\{\var\}\times \{1,\ldots,n\}\big\}\;,\]
and define a distribution $\mu_G$ on $\mX \times \mX$ by 
\[\mu_G((\var,i),(\eq,j))=\mu((\eq,j),(\var,i)) = \frac{1}{2}\nu(j)\frac{1}{|h_{j\cdot}|} 1_{h_{ji}=1}\;,\]
where $|h_{j\cdot}|$ denotes the Hamming weight of the $j$-th row of $h$ (if this is $0$, then $\mu_G((\var,i),(\eq,j))$ is set to $0$ for all $i$). For any $j$, we let $\mA((\eq,j))=\F_2^{|h_{j\cdot}|}$, and for any $i\in\{1,\ldots,n\}$, $\mA((\var,i))=\F_2$. Finally let $D((\eq,j),(\var,i),a,b)=1_{h_{j\cdot} a=0} 1_{a_i=b}$, where $h_{j\cdot}a$ is naturally computed as the sum, in $\F_2$, of all entries of $a$.  Then $G_{\code,M}=(\mX,\mu_G,\mA,D)$. 
\end{definition}

We show the following. For $\code$ an $[n,k,d]_2$ linear code and $M=(h,\nu)$ an $r$-local tester for $\code$, recall the presentation~\eqref{eq:def-gh-pres} of $G(h)$. We define a distribution $\mu$ on the relations defining that presentation as follows. First, sample a random $j\in\{1,\ldots,m\}$ according to $\nu$, and uniformly random $i,i'\in\{1,\ldots,n\}$, conditioned on $h_{ji}=h_{ji'}=1$ and $i\neq i'$. Then, with probability $1/3$ we return the relation $x_i^2=e$, with probability $1/3$ we return the relation $R_i$, and with probability $1/3$ we return $R'_{jii'}$. 

\begin{proposition}\label{prop:rep-game}
Let $\code$ be an $[n,k,d]_2$ linear code, $M=(h,\nu)$ an $r$-local tester for $\code$, and $\mu$ the distribution defined above. Let $\strategy = \{P^{x}\}_{x \in \mX}$ be a strategy for $G_{\code,M}$ on $(\mN,\tau^\mN)$ such that $\omega^*(G_{\code,M};\strategy)\geq 1-\eps$. Then $\phi:x_i\mapsto {\widehat{P}}^{(\var,i)}$, for $i\in\{1,\ldots,n\}$, is an $(O(r\eps),\mu)$-homomorphism of the presentation~\eqref{eq:def-gh-pres} of $G(h)$.
\end{proposition}

\begin{proof}
Let $\strategy$ be a synchronous strategy for $G_{\code,M}$  in  $(\mM,\tau)$ that succeeds with probability at least $1-\eps$. For $i\in\{1,\ldots,n\}$ let 
\[ \phi(x_i)\,=\,{\widehat{P}}^{(\var,i)}\,=\, P^{(\var,i)}_0-P^{(\var,i)}_1\;.\]
Then by definition $\phi(x_i)^2=\Id$ for all $i$. Recalling~\eqref{eq:def-gh-pres}, it remains to verify the relations $R_j$ and $R'_{jii'}$, for $1\leq i<i'\leq n$ and $1\leq j \leq m$. To show these relations, first express the assumption that $\strategy$ succeeds in $G_{\code,M}$ as
\begin{align}
1-\eps &\leq \Es{j\sim \nu} \Es{i: h_{ji}=1} \sum_{a:h_{j\cdot} a=0}\tau\big(P^{(\eq,j)}_a P^{(\var,i)}_{a_i}\big)\label{eq:rep-game-1}
\end{align}
For an equation $j$ and variables $i,i'$, let 
\[ R^{(\eq,j)}_{ii'} = \sum_a (-1)^{a_i+a_{i'}} P^{(\eq,j)}_a\qquad\text{and}\qquad R^{(\eq,j)}_{i} = \sum_a (-1)^{a_i}P^{(\eq,j)}_a\;,\]
so that, using that $\{P^{(\eq,j)}_a\}$ is a projective measurement, 
\[R^{(\eq,j)}_{ii'} \,=\, R^{(\eq,j)}_{i}R^{(\eq,j)}_{i'}\,=\,R^{(\eq,j)}_{i'}R^{(\eq,j)}_{i}\;.\]
Using this we compute
\begin{align*}
\Es{(j,i,i')\sim\mu} \|R'_{jii'}-\Id\|_\tau^2 &= \Es{(j,i,i')} \big\| \widehat{P}^{(\var,i)}\widehat{P}^{(\var,i')} - \widehat{P}^{(\var,i')}\widehat{P}^{(\var,i)}\big\|_\tau^2\\
&= \Es{(j,i,i')} \big\| \big(\widehat{P}^{(\var,i)} - R^{(\eq,j)}_{i}\big)\widehat{P}^{(\var,i')} + R^{(\eq,j)}_{i}\big(\widehat{P}^{(\var,i')}-R^{(\eq,j)}_{i'}\big)  \\
&\hskip3cm - R^{(\eq,j)}_{i'}\big(\widehat{P}^{(\var,i)}-R^{(\eq,j)}_{i}\big)- \big(\widehat{P}^{(\var,i')}- R^{(\eq,j)}_{i'}\big)\widehat{P}^{(\var,i)}  \big\|_\tau^2\\
&\leq 2\Big( \Es{(j,i,i')} \big\| \widehat{P}^{(\var,i)} - R^{(\eq,j)}_{i}\big\|_\tau^2 +\Es{(j,i,i')} \big\| \widehat{P}^{(\var,i')} - R^{(\eq,j)}_{i'}\big\|_\tau^2\Big)\\
&= 2\Big(4-2\Big(\Es{(j,i,i')} \tau\big( \widehat{P}^{(\var,i)}  R^{(\eq,j)}_{i}\big) +\tau\big(\widehat{P}^{(\var,i')}  R^{(\eq,j)}_{i'}\big)\Big)\Big)\\
&=2\Big(8-8 \Es{j\sim \nu} \Es{i:h_{ji}=1}\sum_{a} \tau\big( \widehat{P}^{(\var,i)}_{a_i}  P^{(\eq,j)}_{a}\big)\Big)\;.
\end{align*}
Here we abused notation slightly and denoted $\Es{(j,i,i')}$ the expectation for a relation $R'_{jii'}$ sampled according to $\mu$, conditioned on such a relation being sampled. The third line is the Cauchy-Schwarz inequality and uses that $\widehat{P}^{(\var,i)}$ and $R^{(\eq,j)}_{i}$ are observables, hence square to identity; the fourth line expands the norms and also uses this fact; and the fifth line uses that for observables $A=A0-A_1$ and $B=B_0-B_1$, $AB=\Id-2(A_0B_0+A_1B_1)$. The last line also uses that the marginal of $(j,i,i')\sim\mu$ on $(j,i)$ or $(j,i')$ are identical and match the distribution indicated in the last line. Using~\eqref{eq:rep-game-1} we deduce that  
\[ \Es{(j,i,i')\sim\mu} \|R'_{jii'}-\Id\|_\tau^2 \,\leq\, 16\eps\;.\]
Now we consider the relations $R_j$. For $j\in\{1,\ldots,m\}$ we denote $i_1,\ldots,i_r$ the indices such that $h_{ji}=1$ (assume for simplicity of notation that there are exactly $r$). Then 
\begin{align*}
\Es{j} \|R_{j}-\Id\|_\tau^2
 &= \Es{j} \big\| \widehat{P}^{(\var,i_1)}\cdots \widehat{P}^{(\var,i_r)}-\Id\big\|_\tau^2\\
&\leq (r+1)\Big(\Es{j} \big\| R^{(\eq,j)}_{i_1}\cdots R^{(\eq,j)}_{i_r}-\Id\big\|_\tau^2 + \sum_{t=1}^r \big\|R^{(\eq,j)}_{i_t}-\widehat{P}^{(\var,i_t)}\big\|_\tau^2\Big)\\
&\leq (r+1)\big(2-2\Es{j}\tau\big( R^{(\eq,j)}_{i_1}\cdots R^{(\eq,j)}_{i_r}\big)\big) + O(r\eps)\\
&\leq O(r\eps)\;,
\end{align*}
where the second line follows from the triangle inequality and using a telescoping sum, the third line uses the definition of $R$ and~\eqref{eq:rep-game-1}, and the last line again uses the definition of $R$ and~\eqref{eq:rep-game-1}. This concludes the proof. 
\end{proof}

\subsection{Braiding the code test}
\label{sec:braiding}

\begin{figure}[!htbp]
  \centering
  \begin{gamespec}
Let $M=(h,\nu)$ be an $r$-local tester for the $[n,k,d]_2$ code $\code$.  Execute either of the following tests with probability $1/3$ each. 
    \begin{enumerate}
      \setlength\itemsep{1pt}
    \item \textbf{Code test:} Sample $W\in \{X,Z\}$ uniformly at random. Execute the code game $G_{\code,M}$ with both players, prepending the symbol $W$ to all questions (which now take the form $(W,\var,i)$ or $(W,\eq,j)$). 
		
    \item \textbf{(Anti-)commutation test:} Sample $(i_X,i_Z)\in \{1,\ldots,n\}^2 $ uniformly at random. Let $\omega = (E_\code e_{i_X}, E_\code e_{i_Z})$ and $\gamma =  (E_\code e_{i_X}) \cdot(E_\code e_{i_Z}) \in \field$. 
		\begin{enumerate} 
		\item If $\gamma=0$ then sample a pair of questions $(x_c,y_c)\sim\mu_\cc$ as in the commutation game. Send $(x_c,\omega)$ to $\alice$ and $(y_c,\omega)$ to $\bob$. Accept if and only if their answers are accepted in the commutation game. 
		\item If $\gamma\neq 0$ then do the same but for the anti-commutation game. 
		\end{enumerate} 
		 \item \textbf{Consistency test:} Sample $(i_X,i_Z)\in \{1,\ldots,n\}^2 $ and $W\in \{X,Z\}$ uniformly at random. Let $\omega=(E_\code e_{i_X}, E_\code e_{i_Z})$ and $\gamma = (E_\code e_{i_X}) \cdot(E_\code e_{i_Z}) \in \field$. Send $(W,\var,i_W)$ to $\alice$ and $(x_{W,\gamma},\omega)$ to $\bob$, where $x_{W,\gamma}$ is a question from the (anti-)commutation game. Receive $a\in \field$ and $b\in \field$ respectively. Accept if and only if $a=b$. 
    \end{enumerate}
  \end{gamespec}
  \caption{The braiding test over $\code$ verifies that the players respond consistently with a uniformly random codeword from $\code$. $E_\code \in \F_2^{k\times n}$ is a generating matrix for $\code$, and for $i\in\{1,\ldots,n\}$ we let $e_i$ be the $i$-th canonical basis vector of $\F_2^n$.}
  \label{fig:braiding-test}
\end{figure}

Let $\code$ be an $[n,k,d]_2$ linear code and $M=(h,\nu)$ an $r$-local tester for $\code$. We let $E_\code\in\F_2^{k\times n}$ be a generating matrix for $\code$, i.e.\ $E_\code$ is such that its rows are linearly independent and span the codespace. For convenience we assume throughout that $E_\code$ has no repeated columns. 
 The braiding test constructed from $\code$ and $M$ is a synchronous game described in Figure~\ref{fig:braiding-test}. The test combines two independent copies of the code game from Section~\ref{sec:code-game} with appropriate commutation and anti-commutation sub-tests. The braiding test is designed to force any successful strategy in it to be close, in some sense, to a representation of the Pauli group generated by observables $\sigma^X(a)$ and $\sigma^Z(b)$, $a,b\in \F_2^k$ (recall the notation from~\eqref{eq:def-pauli-2}). This is shown in the following theorem. 

\begin{theorem}\label{thm:braiding}
Let $\mC$ be a class of tracial von Neumann algebras. Let $\code$ be an $[n,k,d]_2$ linear code and $M=(h,\nu)$ an $r$-local tester for $\code$. Suppose that the presentation $G(h)$ in~\eqref{eq:def-gh-pres} is such that $G(h)=\Z_2^k$ and furthermore this presentation is $(\delta,\nu_R,\nu_S,\mC)$-stable.\footnote{The distributions $\nu_R$ and $\nu_S$ are defined from $\nu$ as described right after Definition~\ref{def:code-test}.} Then the braiding test over $\code$ is a $(k,\delta')$-qubit test with sets $S_X=S_Z=\{E_\code e_i:\,i\in\{1,\ldots,n\}\}\subseteq \field^k$, map $\phi(W,E_\code e_i)=(W,\eq,i)$ and error function $\delta' = O(\delta^{1/2}(6\eps))$.
\end{theorem}

We will make use of the following simple fact.

\begin{lemma}[Data-processing]\label{lem:dp}
Let $\{P_a\}$ and $\{Q_a\}$ be two POVMs on $(\mM,\tau)$ with the same outcome set $\mA$. Then for any function $f:\mA\to \mB$ for some finite set $\mB$, 
\begin{equation}\label{eq:dp}
\sum_{b\in \mB} \Big\| \sum_{a\in f^{-1}(b)} (P_a-Q_a) \Big\|_\tau^2 \,\leq\, \sum_{a\in \mA} \big\| P_a-Q_a\big\|_\tau^2\;.
\end{equation}
\end{lemma}

\begin{proof}
This follows by expanding the left-hand side and using $\tau(P_aQ_{a'})\geq 0$ for all $a\neq a'$. 
\end{proof}

\begin{proof}[Proof of Theorem~\ref{thm:braiding}]
\underline{Completeness:} We first verify completeness. For $W\in\{X,Z\}$, $i\in\{1,\ldots,n\}$ and $b\in \F_2$ let $P^{(W,\var,i)}_b = \frac{1}{2}(\Id + (-1)^b\sigma^W(E_\code e_i))$, and for $j\in\{1,\ldots,m\}$, with $m$ the number of rows of $h$, and $a\in \F_2^r$ let $P^{(W,\eq,j)}_a = \prod_{i: h_{ji}=1} P^{(W,\var,i)}_{a_i}$. Writing $(f_0,f_1)$ for the canonical basis of $\C^2$, $P^{(W,\eq,j)}_a$ is the projection on the span of all $\otimes_{i=1}^k f_{u_i}$ where $u=(u_1,\ldots,u_k)$ is such that $(u^T E_\code)_{|S_j}=a$, where $S_j$ denotes the support of $h_{j\cdot}$.   
For $(i_X,i_Z)\in \{1,\ldots,n\}^2 $ let $\omega=(E_\code e_{i_X}, E_\code e_{i_Z})$ and $\gamma =(E_\code e_{i_X}) \cdot(E_\code e_{i_Z}) $. We let $P^{x_{W,\gamma},\omega} = P^{(W,\var,i_W)}$. 

These choices already ensure that the strategy succeeds with probability $1$ in the consistency test. We verify that it succeeds in the code test. Let $j\in\{1,\ldots,m\}$. As observed above, for any $a\in\field^r$ such that $P^{W,\eq,j}_a\neq 0$ there is an $u\in \field^k$ such that $(u^TE_\code)_{|S_j}=a$, which means that $a$ is the restriction of a valid element of $\code$. 
Using the completeness property of $M$ it follows that $M$ must accept any $a$ such that $P^{W,\eq,j}_a\neq 0$, which shows that the strategy succeeds in the code test with probability $1$. 

It remains to verify that the anti-commutation test is passed with probability $1$. For this we observe that 
the binary observables 
\[ U=\widehat{ P}^{x_{X,\gamma},\omega} \quad\text{and}\quad V= \widehat{P}^{x_{Z,\gamma},\omega} \]
commute in case $\gamma=0$ and anti-commute in case $\gamma=1$. This is because by construction $U=\sigma^X(E_\code i_X)$ and $V=\sigma^W(E_\code i_W)$, and because of the definition of $\gamma$. As a result it can be verified that the pair $(U,V)$ can be completed to a perfect strategy for the commutation game (if $\gamma=0)$ or anti-commutation game (if $\gamma=1)$.
\tnote{added:}For the commutation game, this completion is rather trivial and only requires the introduction of one additional projective $4$-outcome measurement, which measures in the joint eigenbasis of $U$ and $V$ --- this can easily be verified from the definition of the game in~\cite[Section 3.1]{de2022spectral}. For the anti-commutation game, the completion is more subtle and requires extending the Hilbert space on which $U,V$ act by tensoring with $\mathbb{C}^2$. The complete strategy involves $9$ binary observables (including $U$, $V$). It can easily be inferred from the description of a perfect strategy for the anti-commutation game in~\cite[Section 3.1 and Figure 10]{coladangelo2017robust}. 

This defines the measurements $P^{(x,\omega)}$ for $x\notin \{x_{W,\gamma}:\,W\in\{X,Z\},\gamma\in\{0,1\}\}$. 

\bigskip 

\underline{Soundness:} Next we show soundness. Let $\strategy$ be a synchronous strategy for the braiding test  in  $(\mM,\tau)$ that succeeds with probability at least $1-\eps$. For $W\in\{X,Z\}$ let $\strategy^W$ be the strategy in $G_{\code,M}$ that is obtained by restricting $\strategy$ to the relevant measurements corresponding to the ``Code test'' part of the braiding test, i.e.\ the $P^{W,\eq,j}$ and $P^{W,\var,i}$.  Then $\strategy^W$ succeeds with probability at least $1-6\eps$ in $G_{\code,M}$. Using Proposition~\ref{prop:rep-game} it follows that $\{\widehat{P}^{W,\var,i}\}$ is an $O(\eps)$-homomorphism of $G(h)$. Under the assumption that $G(h)=\Z_2^k$ is $(\delta,\nu_R,\nu_S,\mC)$-stable, we deduce that there is a $\delta_1 = O({\delta(6\eps)})$ such that for each $W\in\{X,Z\}$, $\strategy^W$ is $(\delta_1,\nu)$-close to a perfect strategy $\tilde{\strategy}^W$ on $(\mN^W,\tau^W)$ for $G_{\code,M}$, where $\nu$ is uniform on $\{1,\ldots,n\}$. This strategy has measurement operators $\{ \tilde{P}^{W,\eq,j}_a\}$ and $\{\tilde{P}^{W,\var,i}_b\}$ which satisfy 
\begin{equation}\label{eq:main-0}
\Es{i\in\{1,\ldots,n\}}\sum_{b \in \F_2} \big\|P^{W,\var,i}_b - (w^W)^* \tilde{P}_b^{W,\var,i} w^W \big\|^2_\tau \,=\, O(\delta_1)\;.
\end{equation}
Furthermore, using that $G(h)=\Z_2^k$ is Abelian, there is a  POVM $\{\tilde{P}^W_u\}_{u\in \field^n}$ such that $\sum_{u\in \code}\tilde{P}^W_u = \Id$ and for each $i\in \{1,\ldots,n\}$, $\tilde{P}^{W,i}_b = \sum_{u:u_i=b} \tilde{P}^W_u$.
 
Applying Lemma~\ref{lem:pull-back}, we obtain projective measurements $\{Q^W_u\}$ on $\mM$ such that 
\begin{equation}\label{eq:main-1}
\sum_u \big\|Q^W_u - (w^W)^* \tilde{P}_u^W w^W \big\|^2_\tau \,=\, O(\delta_1)\;.
\end{equation}
For any $i\in\{1,\ldots,n\}$ let $Q^{W,i}_b = \sum_{u:\,u_i=b}  Q^W_u$. Then by Lemma~\ref{lem:dp},
\begin{align*}
\Es{i}\sum_b \tau\big( Q^{W,i}_b (w^W)^*\tilde{P}^{W,i}_b(w^W) \big)
&\geq \Es{i}\sum_u \tau\big( Q^{W}_u (w^W)^*\tilde{P}^{W}_u(w^W) \big)\\
&\geq 1-O(\delta_1)\;.
\end{align*}
For $W\in\{X,Z\}$ and $b\in \F_2^k$ let $R^W_b = Q^W_{G_\mC^T b}$, where by definition $G_\mC^T b\in \mC$. 

We now define a strategy $\strategy'$ for the game $G_{\dlS}$. On question $W\in \{X,Z\}$ the projective measurement is $\{R^W_b\}$. On question of the form $(x_c,\omega)$ for $x_c$ a question in the commutation game, the projective measurement is $\{P^{x_c,\omega}\}$, i.e.\ the same projective measurement as used in $\strategy$. Similarly, on a question of the form $(x_{ac},\omega)$ for $x_{ac}$ a question in the anti-commutation game, the projective measurement is $\{P^{x_{ac},\omega}\}$.

To conclude we show that this strategy succeeds in the game $G_{\dlS}$ with probability $1-O(\sqrt{\delta_1})$. Assuming that this has been shown, by Theorem~\ref{thm:dls-braid} the strategy $\strategy'$ is $O(\sqrt{\delta_1})$-close to a strategy $\strategy''$ on an algebra of the form $(M_{2^{k}}(\C)\otimes \mN,\tr\otimes \tau')$ such that $P^W_b = \sigma^W_b\otimes I_\mN$. By definition of $R$, 
\begin{equation}\label{eq:main-3}
 Q^{W,i}_b \,=\,  \sum_{a \in \F_2^n:\,a_i=b}  Q^W_a \,=\, \sum_{c \in \F^k_2: (G_\mC^T c)_i=b}  R^W_c \;,
\end{equation}
hence
\begin{equation*}
 \widehat{Q}^{W,i}\,=\, \sum_c (-1)^{c\cdot (G_\mC e_i)} R^W_c \,=\, \widehat{R}^W(G_\mC e_i)\;.
\end{equation*}
Using the definition of the game distribution, closeness of $\strategy'$ and $\strategy''$ thus implies that
\begin{equation*}
\Es{i\in\{1,\ldots,n\}} \big\|\widehat{Q^{W,i}} - (w'')^* {\sigma^W}(G_\mC e_i) (w'') \big\|_\tau^2 \,=\,O(\sqrt{\delta_1})\;.
\end{equation*}
Combining with~\eqref{eq:main-0} and~\eqref{eq:main-1}, this shows the theorem.

It remains to verify that $\strategy'$ succeeds in the game $G_{\dlS}$ with probability $1-O(\sqrt{\delta_1})$. By definition $\strategy'$ succeeds in the (anti)-commutation test with probability $1-O(\eps)$. It remains to check the $W$-consistency test, for $W\in\{X,Z\}$. Because $\strategy$ succeeds with probability $1-O(\eps)$ in the consistency test, 
\begin{equation*}
\Es{i_X,i_Z\in\{1,\ldots,n\}} \sum_b \tau\big( P^{W,\var,i}_b P^{x_{W,\gamma},\omega}_b\big) \,\geq\, 1-O(\eps)\;,\
\end{equation*}
where $\omega$ and $\gamma$ are defined as in Figure~\ref{fig:braiding-test}. Using~\eqref{eq:main-0},~\eqref{eq:main-1}
and Lemma~\ref{lem:close-value} it follows that 
\begin{equation*}
\Es{i_X,i_Z\in\{1,\ldots,n\}} \sum_b \tau\big( Q^{W,i}_b P^{x_{W,\gamma},\omega}_b\big) \,\geq\, 1-O(\sqrt{\delta_1})\;,\
\end{equation*}
Using~\eqref{eq:main-3}, this can be rewritten as 
\begin{equation}\label{eq:main-4}
\Es{i_X,i_Z\in\{1,\ldots,n\}} \sum_{b,c: (G_\mC^T c)_i=b} \tau\big( R^{W}_c P^{x_{W,\gamma},\omega}_b\big) \,\geq\, 1-O(\sqrt{\delta_1})\;.
\end{equation}
Since $\omega_W = G_\mC e_i$, $(G_\mC^T c)_i = c\cdot \omega_W$. Thus~\eqref{eq:main-4} shows that $\strategy'$ succeeds with probability $1-O(\sqrt{\delta_1})$ in the $W$-consistency test, as desired. 
\end{proof}

\subsection{The quantum low-degree test}
\label{sec:pbt}

By instantiating $\code$ using the Reed-Muller code from Section~\ref{sec:eff-z2k} we obtain an efficient qubit test.  Because we do not know if $G(h_\bRM)$ is Abelian, we need to introduce an additional test for the relations $\{R_k^\com\}$ in~\eqref{eq:z2k-eff}. The resulting test is described in Figure~\ref{fig:pbt}. It is a variant of a test first introduced in~\cite{natarajan2018low} (with a flawed analysis). This paper (together with the work on which our analysis relies) corrects this, and Corollary~\ref{cor:qld} below recovers~\cite[Theorem 3.2]{natarajan2018low}, with essentially the same parameters but a slightly different test. (Some of the complexity-theoretic applications of the theorem given in~\cite{natarajan2018low} remain flawed, see~\cite{natarajan2024status} for an explanation.)

In the next section we detail an important application of this test, Proposition~\ref{prop:dim-test}, which is used in the work~\cite{ji2020mip}. The present proof somewhat simplifies the one given in~\cite{ji2020mip}, but it keeps the same overall strategy. Moreover, Corollary~\ref{cor:qld} below is qualitatively stronger than~\cite[Theorem 7.14]{ji2020mip} as the latter applies to a slightly different game which in particular has a pair of questions with long answers, of $O(2^m\log q)$ bits, whereas the tester in Figure~\ref{fig:bRM-tester}, which is analyzed in the corollary, has short answers, $O(d)$ bits long; with a typical parameter setting (see the discussion following Theorem~\ref{thm:z2-stab}) this is exponentially smaller. The presence of long answers makes the analysis easier.

\begin{figure}[!htbp]
  \centering
  \begin{gamespec}
Let $h_\bRM\in \F_2^{M\times N}$ be the parity check matrix for $\code_\bRM$ considered in Section~\ref{sec:eff-z2k}. Here, $N=q^{m+1}$ and $M=q^{m+2}(1+m)$. Let $\mu_R$ be the distribution on relations~\eqref{eq:z2k-eff} described in Section~\ref{sec:eff-z2k}. Execute either of the following tests with probability $1/4$ each. 
    \begin{enumerate}
      \setlength\itemsep{1pt}
    \item \textbf{Code test:} Identical to that in Figure~\ref{fig:braiding-test}.
    \item \textbf{(Anti-)commutation test:} Identical to that in Figure~\ref{fig:braiding-test}.
		 \item \textbf{Consistency test:} Identical to that in Figure~\ref{fig:braiding-test}.
		\item \textbf{Pairwise commutation test:} Sample $(i,i')\in \{1,\ldots,N\}^2$ according to the distribution $\mu_R$, conditioned on choosing a relation from $R_k^\com$, and $W\in \{X,Z\}$ uniformly at random. 	
		Let $\omega=(E_\code e_{i}, E_\code e_{i'})$. 
		Sample a pair of questions $(x_c,y_c)$ as in the commutation game. If either question is $x_{X,0}$ or $x_{Z,0}$, replace it with $(W,\var,i)$ or $(W,\var,i')$ respectively. Otherwise, send the original question together with $\omega$, i.e.\ $(x_c,\omega)$ or $(y_c,\omega)$ respectively. Accept if and only if the players' answers are accepted in the commutation game.  
    \end{enumerate}
  \end{gamespec}
  \caption{The quantum low-degree test, obtained by adapting the braiding test from Figure~\ref{fig:braiding-test} to the $[q^{m+1},t(d+1)^m,D']_2$ code $\code_\bRM$.}
  \label{fig:pbt}
	
\end{figure}

\begin{corollary}[Quantum low-degree test]\label{cor:qld}
Let $\code_\bRM$ be the $[q^{m+1},t(d+1)^m,D']_2$ Reed-Muller code from Section~\ref{sec:eff-z2k}, and $M=(h_\bRM,\nu_\bRM)$ the $(d+2)$-local tester for $\code_\bRM$ described in Figure~\ref{fig:bRM-tester}. Then the associated braiding test (Figure~\ref{fig:pbt}) is a $(k,\delta')$-qubit test with error function $\delta' = \poly(m,d,t)\cdot\poly(\eps,q^{-1})$.
\end{corollary}

\begin{proof}
Let $\strategy$ be a synchronous strategy that succeeds in the braiding test with probability at least $1-\eps$. Then in particular the strategy succeeds with probability at least $1-4\eps/3$ in the braiding test over $\code_\bRM$. Since $(\code_\bRM,M)$ is not known to be abelian, we cannot apply Theorem~\ref{thm:braiding} directly. However, we can follow its proof. 

The completeness part of the proof follows in a straightforward manner, since the measurement operators $P^{(W,\var,i)}$ and $P^{(W,\var,i')}$ defined in the proof pairwise commute, for any pair $(i,i')\in \{1,\ldots,N)^2$. 

For the soundness part, we first define the same pair of strategies $\strategy^X$ and $\strategy^Z$ for $G_{\code_\bRM,M}$. Applying Proposition~\ref{prop:rep-game}, we deduce an $O(d\eps)$-homomorphism of the presentation~\eqref{eq:def-gh-pres} of $G(h_\bRM)$. However, we are interested in constructing an approximate homomorphism of the presentation~\eqref{eq:z2k-eff}, which in addition contains the relations $R_k^\com$. The fact that $x_i \mapsto \widehat{P}^{(W,\var,i)}$ satisfies these relations, on average and according to the distribution $\mu_R$, follows from success in the pairwise commutation test executed as part of the Pauli braiding test (Figure~\ref{fig:pbt}). Thus we obtain that $x_i \mapsto \widehat{P}^{(W,\var,i)}$ is an $(O(d\eps),\mu_R)$-homomorphism of the presentation~\eqref{eq:z2k-eff}. Applying Theorem~\ref{thm:z2-stab}, and similarly to the proof of Theorem~\ref{thm:braiding} we obtain a pair of PVMs $\{\tilde{P}^W_u\}_{u\in\F_2^N}$, for $W\in\{X,Z\}$, such that defining $\tilde{P}^{W,i}_b = \sum_{u:u_i=b} \tilde{P}^W_u$ these operators satisfy~\eqref{eq:main-0}, with right-hand side $\delta_2 = \delta(O(d\eps))$, with $\delta$ the soundness function from Theorem~\ref{thm:z2-stab}. From here on the proof proceeds exactly as the proof of Theorem~\ref{thm:braiding}. 
\end{proof}

\subsection{Dimension bounds}

The next proposition states a simple consequence of a qubit test, which is that strategies with a high enough success probability must have a large dimension. This consequence is used in~\cite{ji2020mip}. 

\begin{proposition}
\label{prop:dim-test}
Let $G = (\mX,\mu,\mA,D)$ denote a $(k,\delta(\eps))$-qubit test. Then all synchronous strategies $\strategy$ in $(\mM,\tau)$ for $G$ that succeed with probability $1 - \eps$ must satisfy 
\[
\dim(\mM) \geq \Big( 1 + O(\sqrt{\delta(\eps)}) + \frac{\delta(\eps)}{1 - \delta(\eps)}\Big)^{-1} 2^k~.
\]
\end{proposition}
\begin{proof}
If $\mM$ is infinite-dimensional, then we are done. Suppose instead it were finite-dimensional. Then $\mM$ must be (isomorphic to) a direct sum of finite-dimensional matrix algebras. Without loss of generality we assume that $\mM = M_{d}(\C)$ with the dimension-normalized trace $\tau = \frac{1}{d} \Tr$. 

By the soundness property of qubit tests, the strategy $\strategy$ is $(\delta(\eps),\tilde{\mu})$-close to a strategy $\strategy'$ on an algebra $(M_{2^k}(\C) \otimes \mN,\tr_{2^k} \otimes \tau')$ for some tracial algebra $(\mN,\tau')$ where $\tr_{2^k} = 2^{-k} \Tr$. For notational brevity we write $\mR = M_{2^k}(\C) \otimes \mN$ and $\tau^{\mR} = \tr_{2^k} \otimes \tau'$. By definition there exists a projection $P \in \mM_\infty$ of finite trace and a partial isometry $w \in P \mM_\infty 1_\mM$ satisfying
\begin{enumerate}
	\item $\mR = P \mM_\infty P$. 
	\item $\max \left \{ \tau(1_\mM - w^* w), \tau^{\mR} (P - w w^*) \right \} \leq \delta(\eps)$.
	\item $\tau^{\mR} = \tau_\infty/\tau_\infty(P)$.
\end{enumerate}
For $u \in \F_2^k$ let $\sigma^Z_u$ denote the projection
\[
	\sigma^Z_u = 2^{-k} \sum_{a \in \F_2^k} (-1)^{a \cdot u} \sigma^Z(a)~.
\]
It is easy to verify that $\{\sigma^Z_u \otimes I_\mN \}_{u \in \F_2^k}$ is a projective measurement in $\mR$ and furthermore $\tau^{\mR}(\sigma^Z_u \otimes I_\mN) = 2^{-k}$. Applying \Cref{lem:pull-back} we get that there exists a projective measurement $\{Q_u\}_{u \in \F_2^k}$ on $\mM$ such that
\[
	\sum_u \| Q_u - w^* (\sigma^Z_u \otimes I_\mN) w \|_2^2 \leq 56\, \delta(\eps)~.
\]
Applying \Cref{lem:l1-l2} we get
\[
	\sum_u \tau \Big ( \Big| Q_u - w^* (\sigma^Z_u \otimes I_\mN) w\Big| \Big) \leq O(\sqrt{\delta(\eps)})~.
\]
Then we have
\begin{align*}
	\sum_u \Big| \tau(Q_u) - 2^{-k} \Big| &\leq \sum_u \Big | \tau(w^* (\sigma^Z_u \otimes I_\mN) w) - 2^{-k} \Big| + \tau \Big ( \Big| Q_u - w^* (\sigma^Z_u \otimes I_\mN) w\Big| \Big) \\
	&= O(\sqrt{\delta(\eps)}) + \sum_u \Big | \tau_\infty (w w^* (\sigma^Z_u \otimes I_\mN)) - 2^{-k} \Big| \\
	&= O(\sqrt{\delta(\eps)}) + \sum_u \Big | \tau_\infty (P (\sigma^Z_u \otimes I_\mN)) - 2^{-k} \Big| + \Big | \tau_\infty((P - ww^*)(\sigma^Z_u \otimes I_\mN)) \Big|
\end{align*}
Notice that $\tau_\infty (P (\sigma^Z_u \otimes I_\mN)) = \tau_\infty (\sigma^Z_u \otimes I_\mN) = 2^{-k}$, and that $ww^* \leq P$ and thus $\tau_\infty((P - ww^*)(\sigma^Z_u \otimes I_\mN))$ is a nonnegative real number. Therefore the sum in the last line simplifies to
\[
\sum_u \tau_\infty((P - ww^*)(\sigma^Z_u \otimes I_\mN)) = \tau_\infty((P - ww^*) P) = \tau_\infty(P - ww^*) \leq \tau_\infty(P) \cdot \delta(\eps)~.
\]
On the other hand the proof of \Cref{lem:pull-back} shows that $\tau_\infty(P) \leq \frac{1}{1 - \delta(\eps)}$, and thus 
\[
\sum_u \Big| \tau(Q_u) - 2^{-k} \Big| \leq O(\sqrt{\delta(\eps)}) + \frac{\delta(\eps)}{1 - \delta(\eps)}~.
\]
By averaging, there exists a $u \in \F_2^k$ such that
\[
\tau(Q_u) = \frac{1}{d} \Tr(Q_u) \leq  \Big( 1 + O(\sqrt{\delta(\eps)}) + \frac{\delta(\eps)}{1 - \delta(\eps)}\Big) 2^{-k}~.
\]
Rearranging, this implies that $d$, the dimension of $\mM$, satisfies
\[
	d \geq \Big( 1 + O(\sqrt{\delta(\eps)}) + \frac{\delta(\eps)}{1 - \delta(\eps)}\Big)^{-1} 2^k
\]
as desired.
\end{proof}

\appendix

\section{Approximately commuting projections}\label{sec:proj}

In this appendix we show a variant of Theorem 3.2 from~\cite{chao2017overlapping} for the case of the Hilbert-Schmidt norm, as opposed to the operator norm in~\cite{chao2017overlapping}. The proof is similar in that it identifies a sequential rounding mechanism for the projection ; but the tracking of errors is somewhat different.  The resulting bound is the same. 

The proof presented in this appendix was prepared with the help of ChatGPT after a reader pointed out the discrepancy in norms; the authors verified and edited the proof. 

\begin{theorem}[Hilbert--Schmidt separation of projections]
\label{thm:main2}
Let $N\geq2$, and let $P_1,\ldots,P_N$ be projections on a
$d$-dimensional Hilbert space $\mH$ equipped with the trace $\tau=d^{-1}\Tr_{\mH}$.  Assume
\[
  \Es{1\leq i<j\leq N} \lVert[P_i,P_j]\rVert_\tau^2\leq\eps^2\;.
\]
Then there are pairwise commuting projections $Q_1,\ldots,Q_N$ on
$\mH$ such that
\[
  \Es{1\leq i\leq N} \lVert P_i-Q_i\rVert_\tau
  \leq 8N\eps\;.
\]
\end{theorem}

The proof of the theorems relies on three linear-algebraic lemmas. These are probably well-known, but we include the statements and proofs for completeness. The proof of the theorem is given in the following section. 

\subsection{Auxiliary lemmas}

\subsubsection{A sequential projection estimate}

\begin{lemma}
\label{lem:sequential}
Let $\Pi_1,\ldots,\Pi_m$ be orthogonal projections on a Hilbert space,
let $x$ be a vector, and set
\[
  y=\Pi_m\cdots\Pi_1x,
  \qquad
  a=\sum_{j=1}^m\lVert(\Id-\Pi_j)x\rVert^2.
\]
Then
\[
  \lVert x-y\rVert^2\leq a,
  \qquad
  \lVert x\rVert^2-\lVert y\rVert^2\leq 4a.
\]
\end{lemma}

\begin{proof}
Put $x_0=x$, $x_j=\Pi_jx_{j-1}$, and
$d_j=x_{j-1}-x_j$.  Since $x_j\perp d_j$,
\[
  L:=\lVert x\rVert^2-\lVert y\rVert^2
    =\sum_{j=1}^m\lVert d_j\rVert^2.
\]
Moreover $d_j\in\ker\Pi_j$, and therefore
\[
  \langle x,d_j\rangle
  =\langle(\Id-\Pi_j)x,d_j\rangle.
\]
Since $x-y=\sum_jd_j$, Cauchy--Schwarz gives
\[
\begin{aligned}
  \lVert x-y\rVert^2
   &=2\operatorname{Re}\langle x,x-y\rangle-L \\
   &\leq 2\sqrt{aL}-L
   \leq a.
\end{aligned}
\]
The same calculation, with the nonnegative term
$\lVert x-y\rVert^2$ discarded, yields
$L\leq 2\sqrt{aL}$, and hence $L\leq 4a$.
\end{proof}

\subsubsection{Rounding an almost invariant subspace}

\begin{lemma}
\label{lem:round-subspace}
Let $\mK$ be finite-dimensional Hilbert space and $\sigma$ a faithful finite trace on $\mathcal B(\mK)$.  Let
$R_1,\ldots,R_N$ be pairwise commuting projections on $\mK$, and let $E$ be a projection on
$\mK$.  Then there is a projection $F$ on $\mK$ that commutes with
every $R_i$, satisfies $\sigma(F)=\sigma(E)$, and obeys
\[
  \lVert E-F\rVert_\sigma^2
  \leq
  2\sum_{i=1}^N\lVert[R_i,E]\rVert_\sigma^2.
\]
\end{lemma}

\begin{proof}
For $x=(x_1,\ldots,x_N)\in\{0,1\}^N$, let
\[
  R_x=\prod_{i=1}^N R_i^{x_i}(\Id-R_i)^{1-x_i}.
\]
The nonzero $R_x$ are the joint spectral projections of the commuting
family.  Define the joint pinching of $E$ by
\[
  B=\sum_x R_xER_x.
\]
The map $X\mapsto\sum_x R_xXR_x$ is the orthogonal projection, for
the $\sigma$-inner product, onto the algebra of operators that are
block diagonal with respect to the joint eigenspace decomposition.  In
particular,
\[
  0\leq B\leq\Id,
  \qquad
  \sigma(B)=\sigma(E).
\]
Furthermore,
\[
  \lVert E-B\rVert_\sigma^2
  =\sum_{x\neq y}\sigma(R_yER_xER_y),
\]
whereas
\[
  \sum_{i=1}^N\lVert[R_i,E]\rVert_\sigma^2
  =\sum_{x,y}
       \left(\sum_{i=1}^N|x_i-y_i|^2\right)
       \sigma(R_yER_xER_y).
\]
The Hamming distance between distinct $x$ and $y$ is at least one, so
\begin{equation}
  \lVert E-B\rVert_\sigma^2
  \leq
  \sum_{i=1}^N\lVert[R_i,E]\rVert_\sigma^2.
  \label{eq:pinching-bound}
\end{equation}

Let $r=\operatorname{rank}(E)$.  Choose an orthonormal basis
consisting simultaneously of eigenvectors of $B$ and of all the
$R_i$.  Write the eigenvalues of $B$ in nonincreasing order,
\[
  1\geq\lambda_1\geq\cdots\geq\lambda_D\geq0,
  \qquad
  \sum_{s=1}^D\lambda_s=r.
\]
Let $F$ project onto $r$ basis vectors associated with the $r$ largest
eigenvalues.  This choice ensures that $F$ commutes with every $R_i$;
since every trace on $\mathcal B(\mK)$ is a positive scalar multiple of
$\Tr_{\mK}$, it also gives $\sigma(F)=\sigma(E)$.

We claim
\begin{equation}
  \lVert B-F\rVert_\sigma^2
  \leq \sigma(B-B^2).
  \label{eq:rank-rounding}
\end{equation}
Since $\sigma$ is a scalar multiple of $\Tr_{\mK}$, it suffices to
prove this after replacing $\sigma$ by $\Tr_{\mK}$.  If $r=0$ or
$r=D$, the claim is immediate.  Otherwise set
\[
  a=\lambda_r,
  \qquad
  t=\sum_{s>r}\lambda_s
    =\sum_{s\leq r}(1-\lambda_s).
\]
Since $\lambda_s\leq a$ for $s>r$ and $\lambda_s\geq a$ for
$s\leq r$,
\[
  \sum_{s>r}\lambda_s^2
  \leq at
  \leq\sum_{s\leq r}\lambda_s(1-\lambda_s).
\]
Consequently
\[
  \Tr_{\mK}(B^2)=\sum_s\lambda_s^2
  \leq\sum_{s\leq r}\lambda_s
  =\Tr_{\mK}(BF).
\]
Using $\Tr_{\mK}(F)=\Tr_{\mK}(B)=r$, this is equivalent to
\eqref{eq:rank-rounding}.

Because $B$ is the $\sigma$-orthogonal projection of $E$ onto the
block-diagonal algebra,
\[
  \sigma(B-B^2)=\lVert E-B\rVert_\sigma^2,
\]
and $E-B$ is $\sigma$-orthogonal to $B-F$.  Hence
\[
\begin{aligned}
  \lVert E-F\rVert_\sigma^2
   &=\lVert E-B\rVert_\sigma^2
     +\lVert B-F\rVert_\sigma^2 \\
   &\leq 2\lVert E-B\rVert_\sigma^2.
\end{aligned}
\]
Combining this with \eqref{eq:pinching-bound} proves the lemma.
\end{proof}

\subsubsection{Matching two subspaces by close isometries}

\begin{lemma}
\label{lem:isometries}
Let $V:\mH\to\mK$ be an isometry, let $E=VV^*$, and let $F$ be a
projection on $\mK$ with $\operatorname{rank}(F)=\dim\mH$.  Let $\tau$
and $\sigma$ be faithful finite traces on $\mathcal B(\mH)$ and
$\mathcal B(\mK)$, respectively.  There is an isometry
$W:\mH\to\mK$ with $WW^*=F$ such that
\[
  \lVert V-W\rVert_\tau
  \leq
  \left(\frac{\tau(\Id_\mH)}{\sigma(E)}\right)^{1/2}
  \lVert E-F\rVert_\sigma.
\]
In particular, if the traces are compatible on the range of $V$, i.e.
$\sigma(VXV^*)=\tau(X)$ for all $X\in\mathcal B(\mH)$, then
$\lVert V-W\rVert_\tau\leq\lVert E-F\rVert_\sigma$.
\end{lemma}

\begin{proof}
Take the polar decomposition
\[
  FV=U|FV|,
\]
and extend the partial isometry $U$ to an isometry
$W$ from $\mH$ to the range of $F$.  Because $W^*=W^*F$ and $W$ extends $U$,
we have $W^*V=W^*FV=|FV|$.  Let $d=\dim\mH$, and let
$s_1,\ldots,s_d\in[0,1]$ be the singular values of $FV$.  Since
$\tau$ and $\sigma$ are scalar multiples of the usual traces, with
multipliers $\tau(\Id_\mH)/d$ and $\sigma(E)/d$, respectively,
\[
  \lVert V-W\rVert_\tau^2
  =\frac{\tau(\Id_\mH)}{d}\,2\sum_{r=1}^d(1-s_r),
\]
while
\[
  \lVert E-F\rVert_\sigma^2
  =\frac{\sigma(E)}{d}\,2\sum_{r=1}^d(1-s_r^2).
\]
Since $1-s_r\leq1-s_r^2$, the stated inequality follows.
\end{proof}

\subsection{Proof of Theorem~\ref{thm:main2}}

Write $P_i^{(1)}=P_i$ and $P_i^{(0)}=\Id-P_i$.  For
$x\in\{0,1\}^N$ let
\[
  P_x=P_N^{(x_N)}\cdots P_1^{(x_1)}.
\]
Set $\mK_N=\mH\otimes(\mathbb C^2)^{\otimes N}$.  Since
$\sum_x P_x^*P_x=\Id$, the map
$V:\mH\to\mK_N$ defined by
\[
  V\xi=\sum_x P_x\xi\otimes|x\rangle
\]
is an isometry.  On $\mathcal B(\mK_N)$ we use the compatible trace
$\sigma=d^{-1}\Tr_{\mK_N}$; thus
$\sigma(VXV^*)=\tau(X)$ for all $X\in\mathcal B(\mH)$.
For $1\leq j\leq N$ let $R_j$ be the projection onto
$|1\rangle_j$ in the $j$th ancillary register, and set
$A_j=V^*R_jV$.  The projections $R_1,\ldots,R_N$ commute exactly.
We first show that
\begin{equation}\label{eq:a-p-1}
  \Es{1\leq j\leq N} \lVert A_j-P_j\rVert_\tau^2
  \leq \frac{N-1}{2}\eps^2
\end{equation}
and furthermore
\begin{equation}\label{eq:a-p-2}
  \Es{1\leq j\leq N} \tau(A_j-A_j^2)
  \leq 2(N-1)\eps^2\;.
\end{equation}
Thus the commuting projections $R_j$ in the dilation have compressions
$A_j$ that are close to the original projections, and the $A_j$ are
close to projections on average.

To prove~\eqref{eq:a-p-1} and~\eqref{eq:a-p-2}, define for any $1\leq i \leq N$
\[
  \mE_i(X)=P_iXP_i+(\Id-P_i) X(\Id-P_i),
  \qquad X\in\mathcal B(\mH).
\]
Each $\mE_i$ is a trace-preserving orthogonal projection on
$L_2(\mathcal B(\mH),\tau)$.  For any $1\leq j\leq N$,
\begin{align*}
  V^*R_jV
   &=\sum_{x:x_j=1} P_x^*P_x \\
   &=\mE_1\bigl(\mE_2(\cdots\mE_{j-1}(P_j)\cdots)\bigr),
\end{align*}
where the empty composition is the identity.  Moreover,
\[
  \lVert(\mathrm{id}-\mE_i)(P_j)\rVert_\tau
  =\lVert[P_i,P_j]\rVert_\tau .
\]
Applying Lemma~\ref{lem:sequential} to the successive projections
$\mE_{j-1},\ldots,\mE_1$ gives
\[
  \lVert A_j-P_j\rVert_\tau^2
  \leq \sum_{i<j}\lVert[P_i,P_j]\rVert_\tau^2
\]
and, since $\tau(A_j)=\tau(P_j)$,
\[
  \tau(A_j-A_j^2)
  \leq 4\sum_{i<j}\lVert[P_i,P_j]\rVert_\tau^2 .
\]
Averaging over $j$ and using the hypothesis proves
\eqref{eq:a-p-1} and~\eqref{eq:a-p-2}.

We now identify commuting projections on the original space $\mH$ by
slightly perturbing the isometry $V$.  Let
\[
  E=VV^*
\]
be the rank-$d$ projection onto the range of $V$.  For any
$1\leq j\leq N$,
\begin{align}
  \lVert[R_j,E]\rVert_\sigma^2
   &=2\sigma\bigl(ER_j(\Id-E)R_jE\bigr) \notag\\
   &=2\tau(A_j-A_j^2).
  \label{eq:comm-defect}
\end{align}
Applying~\eqref{eq:a-p-2},
\begin{equation}
  \Es{1\leq j\leq N}\lVert[R_j,E]\rVert_\sigma^2
  \leq4(N-1)\eps^2\;.
  \label{eq:sum-comm}
\end{equation}
Lemma~\ref{lem:round-subspace} now gives a projection $F$ on
$\mK_N$, commuting with all the $R_j$, such that
$\sigma(F)=\sigma(E)$ and
\begin{equation}
  \lVert E-F\rVert_\sigma
  \leq2\sqrt{2N(N-1)}\,\eps.
  \label{eq:E-F}
\end{equation}
Since $\sigma$ is a scalar multiple of $\Tr_{\mK_N}$ and $E$ has rank
$d$, the equality $\sigma(F)=\sigma(E)$ implies that $F$ has rank $d$.
By Lemma~\ref{lem:isometries}, applied with the compatible traces
$\tau$ and $\sigma$, there is an isometry
$W:\mH\to\mK_N$ satisfying
\[
  WW^*=F,
  \qquad
  \lVert V-W\rVert_\tau
  \leq2\sqrt{2N(N-1)}\,\eps.
\]
Define
\[
  Q_i=W^*R_iW.
\]
Since $F=WW^*$ commutes with every $R_i$,
\[
  Q_iQ_j=W^*R_iFR_jW=W^*R_iR_jW=Q_jQ_i,
\]
and the same calculation with $i=j$ shows $Q_i^2=Q_i$.  Thus the
$Q_i$ are pairwise commuting projections.

It remains to estimate the displacement.  From $A_i=V^*R_iV$ and
$\tau(A_i)=\tau(P_i)$,
\begin{align}
  \lVert R_iV-VP_i\rVert_\tau^2
   &=2\tau(P_i)-2\tau(A_iP_i) \notag\\
   &=\lVert A_i-P_i\rVert_\tau^2+\tau(A_i-A_i^2).
  \label{eq:intertwining}
\end{align}
Equations~\eqref{eq:a-p-1} and~\eqref{eq:a-p-2} imply that
\begin{equation}
 \Es{1\leq i\leq N} \lVert R_iV-VP_i\rVert_\tau^2
  \leq \frac{5}{2}(N-1)\eps^2.
  \label{eq:intertwining-bound}
\end{equation}
Moreover $WQ_i=R_iW$, because $F$ commutes with $R_i$.  Since $W$ is
an isometry,
\begin{align*}
  \lVert Q_i-P_i\rVert_\tau
   &=\lVert R_iW-WP_i\rVert_\tau \\
   &\leq
      2\lVert V-W\rVert_\tau
      +\lVert R_iV-VP_i\rVert_\tau.
\end{align*}
Averaging over $i$ and using Cauchy--Schwarz, \eqref{eq:E-F}, and
\eqref{eq:intertwining-bound}, we obtain
\[
\begin{aligned}
  \Es{1\leq i\leq N}\lVert Q_i-P_i\rVert_\tau
  &\leq 4\sqrt{2N(N-1)}\,\eps
     +\sqrt{\frac{5}{2}(N-1)}\,\eps \\
  &\leq \left(4\sqrt2+\sqrt{\frac52}\right)N\eps
   <8N\eps.
\end{aligned}
\]
This completes the proof.

\section{Lower bounds on the modulus of stability of \eqref{eq:z2-efficient}}\label{appendix:lower_bounds}
In this Appendix we prove Lemma \ref{lem:lower_bound_on_stability_rate_standard_presentation_Z_2^k}, and deduce some corollaries from it. 
Before diving into the proof, we first sketch the idea behind it. 
We provide a method for translating every graph $([k],E\subseteq{\binom{[k]}{2}})$ into a collection of order $2$ unitaries $A_1,...,A_k$ (actually, permutations) that satisfy that
\begin{equation}\label{eq:commuting_or_not}
   \forall ij\in \binom{[k]}{2}\ \colon\ \ \|A_iA_j-A_jA_i\|_\tau^2=\begin{cases}
    2 & ij\in E,\\
    0 & ij\notin E.
\end{cases} 
\end{equation}
Therefore, these unitaries correspond to an $(\nicefrac{|E|}{\binom{k}{2}},\mu_R)$-approximate representation of \eqref{eq:z2-efficient}. For any pair $ij\in E$ we have $\|A_iA_j-A_jA_i\|_\tau^2=2$, and thus the unitaries $A_i$ and $A_j$ would need to be changed  by some constant amount to make them commute. Thus, to fix the $A_i$'s into a genuine representation of $\Z_2^k$, we would need to move them by at least a constant times the proportion of the largest   matching which embeds in $E$. Hence, by choosing a graph which is already a  matching with $c$ edges, we deduce the Lemma. 

We begin by describing our construction. Given a graph $([k],E)$, the permutations $\{A_i\}_{i=1}^k$ act on  the set $\field^{[k]\cup E}$, namely bit strings  indexed by the vertices and edges of the graph. Now, each $A_i$ flips the $i^{\rm th}$ bit. Furthermore, it conditionally flips the  $ij^{\rm th}$  bit for $ij\in E$ where $i<j$, if the $j^{\rm th}$ bit is $1$. Namely, the $A_i$'s act as a NOT gate on the $i^{\rm th}$ bit composed with CNOT gates on all the $ij^{\rm th}$ bits, where $i<j$ and $ij\in E$. Now, all of these permutations are of order $2$. Equation \eqref{eq:commuting_or_not} is deduced by noting that  $A_iA_j=A_jA_i$ when $ij$ is not an edge, and that $A_iA_jA_iA_j$ has no fixed points when $ij$ is an edge --- this is because $A_iA_jA_iA_j$ must flip the $ij^{\rm th}$ bit.  
The von Neumann algebra $\mM$ containing the $A_i$'s is the one acting on $\complex^{\field^{[k]\cup E}}$ with its standard normalized trace $\tau^\cM(X)=\frac{1}{2^{k+|E|}}\Tr(X)$.
\begin{remark}
    A version of this construction (viewed in a different way) was used by Kozlov--Meshulam \cite{kozlov2019quantitative} to upper bound the Cheeger constant of the $k$-dimensional hypercube --- see Section 4.1 therein.
\end{remark}
\begin{claim}\label{claim:appendix1}
    The map $x_i\mapsto A_i\in \cM$ is a $(\nicefrac{|E|}{\binom{k}{2}},\mu_R)$-approximate representation of $\Z_2^k$ with respect to the presentation \eqref{eq:z2-efficient}.
\end{claim}

\begin{proof}
    With probability $\frac{1}{2}$, $\mu_R$ samples an involution relation, which is always satisfied by the  $A_i$'s. Furthermore, 
    \[
\Expectation_{i\neq j\in [k]}\left[\Vert A_iA_jA_iA_j-\Id_\mM\Vert_\tau^2\right]=\frac{1}{\binom{k}{2}}\sum_{ij\in E}\Vert A_iA_jA_iA_j-\Id_\mM\Vert_\tau^2=\nicefrac{2|E|}{\binom{k}{2}}.
\]
By combining these two observations, we deduce the claim.
\end{proof}

\begin{claim}\label{claim:appendix2}
Assume the largest matching in $E$ contains $c$ edges. Then, for every collection $\{B_i\}_{i=1}^k$ of order $2$ unitaries which pairwise commute in $\mN=P\mM_\infty P$, and  every partial isometry $w=PU\Id_\mM\in P\mU(\mM_\infty)\Id_\mM$, we have
\[
\Expectation_{i\in [k]} \Vert A_i-w^* B_iw\Vert_\tau^2 \geq  \frac{c}{16k} 
\]
or 
\[
\tau^\mM(\Id_\mM-w^*w)\geq  \frac{c}{16k}.
\]
In particular, there is no genuine representation of $\Z_2^k$ that is $(\frac{c}{16k},\mu_S)$-close to the $A_i$'s.
\end{claim}
\begin{proof}
Recall that given our von Neumann algebra $\mM$, the algebra $\mM_\infty$ acts on the Hilbert space $\complex^{\field^{[k]\cup E}}\otimes \complex^\Z$. Let $\{e_v\otimes e_t\mid v\in \field^{[k]\cup E},t\in \Z\}$ be the standard basis of this Hilbert space.
Let $\{B_i\}_{i=1}^k$ be order $2$ unitaries which pairwise commute in $\mN=P\mM_\infty P $, and assume  there is a partial isometry $w=PU\Id_\mM\in P\mU(\mM_\infty)\Id_\mM$ and  $0<\eps\leq 3-2\sqrt{2}\approx 0.17$ such that 
\[
\Expectation_{i\in [k]} \Vert A_i-w^* B_iw\Vert_\tau^2\ ,\qquad  \tau^\mM(\Id_\mM-w^*w)\leq  \eps.
\]
This is a slightly weaker condition than for the $B_i$'s to be $(\eps,\mu_S)$-close to the $A_i$'s, as in Definition \ref{def:close}. Furthermore, we can assume without loss of generality that $U=\Id_\infty$, otherwise we replace $\mN$ by $U^*\mN U$ and $P$ by $U^* PU$. Thus, we are given that 
\[
\Expectation_{i\in [k]} \Vert A_i-\Id_\mM  B_i\Id_\mM\Vert_\tau^2\ ,\qquad  \tau^\mM(\Id_\mM-\Id_\mM P\Id_\mM)\leq  \eps,
\]
and
\[
\begin{split}
    \Vert \Id_\mM-\Id_\mM P\Id_\mM\Vert_\tau^2 &=\tau^{\mM}(\Id_\mM\underbrace{-2\Id_\mM P\Id_\mM+\Id_\mM P^2\Id_\mM}_{=-\Id_\mM P\Id_\mM})\leq \eps.
\end{split}
\]
Let $ij\in E$. By \eqref{eq:commuting_or_not}, we have
\begin{align*}
    \sqrt{2}&=\Vert A_iA_jA_iA_j-\Id_\mM\Vert_\tau\\
    &\leq \Vert A_iA_jA_iA_j-\Id_\mM P\Id_\mM\Vert_\tau+\Vert \Id_\mM P\Id_\mM-\Id_\mM\Vert_\tau\\
    &\leq \Vert A_iA_jA_iA_j-\Id_\mM B_iB_jB_iB_j\Id_\mM\Vert_\tau+\sqrt\eps\\
    &\leq \Vert \Id_\mM(A_i-B_i)A_jA_iA_j\Vert_\tau+\Vert \Id_\mM B_i(A_j-B_j)A_iA_j\Vert_\tau\\
    &+\Vert \Id_\mM B_iB_j(A_i-B_i)A_j\Vert_\tau+\Vert \Id_\mM B_iB_jB_i(A_j-B_j)\Id_\mM\Vert_\tau+\sqrt\eps\\
    &=(\heartsuit)+\sqrt\eps.
\end{align*}
Since the $A$'s are unitaries in $\mM$, and by abusing notation and denoting $\Vert X \Vert_\tau=\tau_\infty(X^*X)$, we have 
\begin{align*}
    (\heartsuit)&= \Vert \Id_\mM(A_i-B_i)\Id_\mM\Vert_\tau+\Vert \Id_\mM B_i(A_j-B_j)\Id_\mM\Vert_\tau\\
    &+\Vert \Id_\mM B_iB_j(A_i-B_i)\Id_\mM\Vert_\tau+\Vert \Id_\mM B_iB_jB_i(A_j-B_j)\Id_\mM\Vert_\tau\\
    &\leq \Vert (A_i-B_i)\Id_\mM\Vert_\tau+\Vert (A_j-B_j)\Id_\mM\Vert_\tau
    +\Vert (A_i-B_i)\Id_\mM\Vert_\tau+\Vert (A_j-B_j)\Id_\mM\Vert_\tau\\
    &=(\spadesuit).
\end{align*}
But,
\begin{align*}
     \Vert (A_i-B_i)\Id_\mM\Vert_\tau^2=  \Vert \Id_\mM(A_i-B_i)\Id_\mM\Vert_\tau^2+\Vert (\Id_\infty-\Id_\mM)(A_i-B_i)\Id_\mM\Vert_\tau^2,
\end{align*}
and since $(\Id_\infty -\Id_\mM)A_i=0$ and $\Id_\mM A_i=A_i\Id_\mM$, we have
\begin{align*}
    \Vert (\Id_\infty-\Id_\mM)(A_i-B_i)\Id_\mM\Vert_\tau^2= \Vert (\Id_\infty-\Id_\mM)B_i\Id_\mM\Vert_\tau^2=\Vert (\Id_\infty-\Id_\mM)B_iA_i\Id_\mM\Vert_\tau^2.
\end{align*}
Here the second equality is because by definition, 
\begin{align*}
    \Vert(\Id_\infty-\Id_\mM)B_i\Id_\mM\Vert_\tau^2=\sum_{\substack{v'\in \field^{[k]\cup E} \\ j\neq 1}}\sum_{v\in \field^{[k]\cup E}} |(e_{v'}\otimes e_j)^*B_1e_v\otimes e_1|^2=(\diamondsuit)\;,
\end{align*}
but since $A_i$ permutes $\{e_v\otimes e_1\}_{v\in \field^{[k]\cup E}}$, we have
\begin{align*}
    (\diamondsuit)=\sum_{\substack{v'\in \field^{[k]\cup E} \\ j\neq 1}}\sum_{v\in \field^{[k]\cup E}} |(e_{v'}\otimes e_j)^*B_iA_ie_v\otimes e_1|^2=\Vert (\Id_\infty-\Id_\mM)B_iA_i\Id_\mM\Vert_\tau^2\;.
\end{align*}
Now, for every $v\in \field^{[k]\cup E}$, we have
\[
1-|(e_v\otimes e_1)^*B_iA_ie_v\otimes e_1|^2\leq |1-(e_v\otimes e_1)^*B_iA_ie_v\otimes e_1|^2\leq \Vert \Id_\mM(\Id_\infty-B_iA_i)e_v\otimes e_1\Vert_2^2.
\]
On the other hand, since $B_iA_i$ is a contraction, 
\begin{align*}
    1-|(e_v\otimes e_1)^*B_iA_ie_v\otimes e_1|^2
    &\geq \Vert B_iA_i e_v\otimes e_1\Vert_2^2-|(e_v\otimes e_1)^*B_iA_ie_v\otimes e_1|^2\\
    &=\sum_{(v',j')\neq (v,1)}|(e_{v'}\otimes e_j)^*B_iA_ie_v\otimes e_1|^2\\
    &\geq \Vert (\Id_\infty-\Id_\mM)B_iA_i e_v\otimes e_1\Vert_2^2\;.
\end{align*}
Therefore, by averaging the combined inequalities over $v\in \field^{[k]\cup E}$, we get
\[
\Vert (\Id_\infty -\Id_\mM)B_iA_i\Id_\mM\Vert_\tau^2\leq  \Vert \Id_\mM(\Id_\infty-B_iA_i)\Id_\mM\Vert_\tau^2=\Vert A_i-\Id_\mM B_i\Id_\mM\Vert_\tau^2\;.
\]
Plugging all of this back to $(\spadesuit)$, we get
\begin{align*}
    (\sqrt{2}-\sqrt\eps)^2&\leq (\spadesuit)^2\\
    &\leq 8\Vert(A_i-B_i)\Id_\mM\Vert_\tau^2+8\Vert(A_j-B_j)\Id_\mM\Vert_\tau^2\\
    &\leq 16 \Vert \Id_\mM (A_i-B_i)\Id_\mM\Vert_\tau^2+16 \Vert \Id_\mM (A_j-B_j)\Id_\mM\Vert_\tau^2.
\end{align*}
and since we assumed $\eps<3-2\sqrt{2}$, we have $\sqrt{2}-\sqrt\eps> 1$ and 
\[
 \Vert \Id_\mM (A_i-B_i)\Id_\mM\Vert_\tau^2+ \Vert \Id_\mM (A_j-B_j)\Id_\mM\Vert_\tau^2> \nicefrac{1}{16}.
\]
Now, let $i_1j_1,...,i_cj_c$ be the edges of a maximal  matching in $E$. Then,
\[
k\eps \geq \sum_{i\in [k]}\Vert \Id_\mM (A_i-B_i)\Id_\mM\Vert_\tau^2\geq \sum_{t=1}^c\Vert \Id_\mM (A_{i_t}-B_{i_t})\Id_\mM\Vert_\tau^2+ \Vert \Id_\mM (A_{j_t}-B_{j_t})\Id_\mM\Vert_\tau^2>\nicefrac{c}{16}.
\]
This finishes the proof.
\end{proof}

By combining Claims \ref{claim:appendix1} and \ref{claim:appendix2} applied to a graph which is a matching with $c$ edges, we deduce Lemma \ref{lem:lower_bound_on_stability_rate_standard_presentation_Z_2^k}.

\subsection{The $L^\infty$ analogue}
As discussed in the introduction (see Remark \ref{rem:L^infty_analogue_defn}), it is more common in stability literature to use a $L^\infty$ analogue of Definition~\ref{def:eff-stab}, where the notion of almost-homomorphism and closeness are both measured by taking a supremum over relations and generators respectively, as opposed to averaging according to distributions $\mu_R,\mu_S$. 
Let us recall the exact definition. We say that a homomorphism $\rho\colon \mF(S)\to \mU(\mM)$ is an $(\eps,\infty)$-approximate representation if 
\[
\forall r\in R\ \colon \ \ \|\rho(r)-\Id_\mM\|_\tau^2\leq \eps.
\]
Furthermore, homomorphisms $\rho\colon \mF(S)\to \mU(\mM),\varphi\colon \mF(S)\to \mU(\mN)$ are $(\delta,\infty)$-close if there exists an isometry $w\in P\mM_\infty \Id_\mM$ such that
\[
\forall s\in S\colon\ \ \|\rho(s)-w^*\varphi(s)w\|_\tau^2\leq \delta.
\]
The goal of this subsection of the appendix is to provide a somewhat general procedure to convert lower bounds on the modulus of stability with respect to $\mu_S,\mu_R$ into a lower bound on the $L^\infty$ modulus of stability.

  Let $\langle S\colon R\rangle$ be a presentation of a group $\Gamma$, and let $\mu_S, \mu_R$ be fully supported distributions over the generators and relations respectively.
  Let $\sigma\in \textrm{Sym}(S)$ be a permutation of the generators. 
  Then, $\sigma$ extends (by the universal property of the free group) to an automorphism of $\mF(S)$ which we still denote by $\sigma$ as well. 
  The \emph{automorphism group} of the presentation  $\langle S\colon R\rangle$ is the subgroup of permutations in $\textrm{Sym}(S)$ that preserve $R$. Namely,  
  \[
\Phi=\textrm{Aut}(\langle S\colon R\rangle )=\left\{ \sigma \in \textrm{Sym}(S)\mid R=\sigma(R) \right\}.
  \]
Assume $R=\bigsqcup R_i$ is the decomposition of $R$ into orbits of $\Phi$, and assume $\mu_R$ is uniform over orbits. Let $\rho\colon \mF(S)\to \mU(\mM)$ be an $(\eps,\mu_R)$-approximate representation of $\Gamma$. Define $\rho'\colon \mF(S)\to \mU(\bigoplus_{\alpha\in \Phi}\mM)$ as follows:
    \[
\forall s\in S\ \colon \ \ \rho'(s)=\bigoplus_{\alpha\in \Phi} \rho(\alpha(s)).
    \]
    We denote by $\mM_\Phi=\bigoplus_{\alpha\in \Phi}\mM\subseteq \mM_\infty$ and by $\mM_\alpha$ the copy of $\mM$ at the $\alpha\in \Phi$ position.Note that $\mM_\Phi$ embeds in $\mM_\infty$, e.g. in the coordinates $1,...,|\Phi|$, and inherits a trace from it by defining 
    \[
    \tau^{\mM_\Phi}(\bigoplus A_\alpha)=\frac{\tau_\infty(\bigoplus A_\alpha)}{\tau_\infty(\Id_{\mM_\Phi})}=\frac{1}{|\Phi|}\sum_{\alpha\in \Phi}\tau^\mM(A_\alpha)=\Expectation_{\alpha\in \Phi}\tau^\mM(A_\alpha).
    \]
    To be consistent, whenever we use $\|X\|_\tau^2$ it means $$\tau_\infty(X^*X)=\tau_\mM(X^*X)=|\Phi|\cdot\tau_{\mM_\Phi}(X^*X).$$
    Let $w_i=\mu_R(R_i)$, namely the probability that $\mu_R$ samples a relation from the orbit $R_i$. Since we assumed $\mu_R$ is uniform over orbits, and since $\Phi$ acts transitively on each orbit, we know that 
    \[
    \forall r\in R_i\ \colon \ \ \Expectation_{\alpha\in \Phi}f(\alpha(r))=\Expectation_{r'\in R_i}f(r)
    \]
    for any function $f\colon \mF(S)\to \complex$. Therefore, 
    \[
    \begin{split}
       \forall r\in R_i\ \colon \ \  \frac{w_i}{|\Phi|}\|\rho'(r)-\Id_{\mM_\Phi}\|_\tau^2&=w_i\Expectation_{\alpha\in \Phi} \|\rho(\alpha(r))-\Id_\mM\|_\tau^2\\
        &=w_i\Expectation_{r'\in R_i} \|\rho(r')-\Id_\mM\|_\tau^2\\
        &\leq \sum_j w_j\Expectation_{r'\in R_j} \|\rho(r')-\Id_\mM\|_\tau^2\\
        &=\Expectation_{r\sim \mu_R}\|\rho(r)-\Id_\mM\|_\tau^2\leq \eps.
    \end{split}
    \]
    Hence, 
    \[
    \forall r\in R\ \colon\ \frac{1}{|\Phi|}\|\rho'(r)-\Id_{\mM_\Phi}\|_\tau^2\leq \max({\nicefrac{1}{w_i}})\cdot\eps,
    \]
which in turn means that $\rho'$ is a $(\max_i\{{\nicefrac{1}{w_i}}\}\cdot\eps,\infty)$-approximate representation --- since $$\frac{1}{|\Phi|}\|\cdot\|_\tau^2=\|\cdot\|_{\tau^{\mM_\Phi}}^2,$$
which is the relevant parameter to consider.
\begin{corollary}
    Any $(\eps,\mu_R)$-approximate representation of $\langle S\colon R\rangle$, where $\mu_R$ is uniform over orbits of $\Phi=\textrm{Aut}(\langle S\colon R\rangle)$, can be transformed into a $(\max_i\{{\nicefrac{1}{w_i}}\}\cdot\eps,\infty)$-approximate representation. 
\end{corollary}
\begin{remark}
    By applying this construction on the example from Claim~\ref{claim:appendix1}, the resulting $\rho'$ is a $(\nicefrac{2|E|}{\binom{k}{2}},\infty)$-approximate representation of \eqref{eq:z2-efficient}. This is because $\Phi=\textrm{Sym}(S)$, and there are two  orbits --- the commutation relations and the involutions --- where $\mu_R$ is supprted equally on each of them. Namely, $w_{commutation}=w_{involution}=\nicefrac{1}{2}$.
\end{remark}
Recall that our goal is to translate $L^1$ lower bounds into $L^\infty$ lower bounds. Namely, we would like to deduce that, if every genuine representation of $\Gamma$ is $(\delta,\mu_S)$-far from $\rho$, then every genuine representation of $\Gamma$ is $(\delta',\infty)$-far from $\rho'$, for $\delta'=C\delta$.  Since $\mu_S$ was assumed to be fully supported, the $L^\infty$ distance is lower bounded by the $L^1$-distance, and being $(\delta',\mu_S)$-far from $\rho'$ implies being $(\delta',\infty)$-far from it. Thus, we can forget about the $L^\infty$ notion of distance and lower bound our usual notion of distance.

To that end, assume that every genuine representation of $\Gamma$ is $(\delta,\mu_S)$-far from $\rho$. Namely, for every genuine representation $\varphi\colon \Gamma\to \mU(\mN)$, where $\mN=P\mM_\infty P$, and every isometry $w\in P\mM_\infty \Id_\mM$, we have 
\[
\max\left\{\Expectation_{s\sim \mu_S} \|\rho(s)-w^*\varphi(s)w\|_\tau^2\ ,\ \tau^\mM(\Id_\mM-w^*w)\right\}\geq \delta.
\]
Let $\varphi'\colon \Gamma\to \mU(\mN)$ be a genuine representation of $\Gamma$ which is $(\delta',\mu_S)$-close to $\rho'$, i.e.\ it satisfies 
\[
\frac{1}{|\Phi|}\Expectation_{s\sim \mu_S} \|\rho'(s)-w^*\varphi'(s)w\|_\tau^2\ ,\ \tau^{\mM_\Phi}(\Id_{\mM_\Phi}-w^*w)\leq  \delta',
\]
where  again  $\|X\|_\tau^2=\tau_\infty(X^*X)$.
As before, we can assume $w=PI_{\mM_\Phi}$. Then,
\[
\begin{split}
    \Expectation_{\alpha\in \Phi}\Expectation_{s\sim \mu_S} \|\rho(\alpha(s))-\Id_{\mM_\alpha}\varphi'(s)\Id_{\mM_\alpha}\|_\tau^2&=\frac{1}{|\Phi|}\Expectation_{s\sim \mu_S}\sum_{\alpha\in \Phi}\|\Id_{\mM_\alpha}\rho'(s)\Id_{\mM_\alpha}-\Id_{\mM_\alpha}\varphi'(s)\Id_{\mM_\alpha}\|_\tau^2 \\
    &\leq \frac{1}{|\Phi|}\Expectation_{s\sim \mu_S}\|\rho'(s)-\Id_{\mM_\Phi}\varphi'(s)\Id_{\mM_\Phi}\|_\tau^2\leq \delta'.
\end{split}
\]
In particular, by Markov's inequality, for at least two thirds of the $\alpha\in \Phi$ we have 
\[
\Expectation_{s\sim \mu_S} \|\rho(\alpha(s))-\Id_{\mM_\alpha}\varphi'(s)\Id_{\mM_\alpha}\|_\tau^2\leq 3\delta'.
\]
Similarly,
\[
\begin{split}
    \delta'&\geq \tau^{\mM_\Phi}(\Id_{\mM_\Phi}-\Id_{\mM_\Phi}P\Id_{\mM_\Phi})\\
    &=\Expectation_{\alpha\in\Phi}\tau_\infty(\Id_{\mM_\alpha}-\Id_{\mM_\alpha}P\Id_{\mM_\alpha})
\end{split}\;,
\]
and for at least two thirds of the $\alpha\in \Phi$ we have 
\[
\tau^{\mM_\alpha}(\Id_{\mM_\alpha}-\Id_{\mM_\alpha}P\Id_{\mM_\alpha})\leq 3\delta'.
\]
Hence, $\delta'\geq \frac{\delta}{3}$, and we deduce that every genuine representation of $\Gamma\cong \langle S\colon R\rangle$ is at least $(\nicefrac{\delta}{3},\infty)$-away from $\rho'$.
\begin{corollary}
    The transformation $\rho\mapsto\rho'$ we described  translates a $(\eps,\mu_R)$-approximate representation into a $((\max{\nicefrac{1}{w_i}})\cdot\eps,\infty)$-approximate representaion, and if every genuine representation of $\Gamma$ is $(\delta,\mu_S)$-far from $\rho$, then every genuine representation is $(\nicefrac{\delta}{3},\infty)$-far from $\rho'$.
\end{corollary}
\begin{remark}
    By applying this corollary to the construction from the beginning of the appendix, we conclude that the $L^\infty$ modulus of stability of the presentation \eqref{eq:z2-efficient} is $\Omega(k\eps)$.
\end{remark}



%
\bibliographystyle{amsplain}
\newcommand{\etalchar}[1]{$^{#1}$}


\begin{dajauthors}
\begin{authorinfo}[mc]
  Michael Chapman\\
	Member of the school of Mathematics\\
  Institute for Advanced Study\\
  Princeton, New Jersey, USA\\
  mchapman\imageat{}ias\imagedot{}edu \\
\end{authorinfo}
\begin{authorinfo}[tv]
  Thomas Vidick\\
  Professor\\
  École Polytechnique Fédérale de Lausanne\\
  Lausanne, Switzerland\\
  thomas\imagedot{}vidick\imageat{}epfl\imagedot{}ch \\
\end{authorinfo}
\begin{authorinfo}[hy]
  Henry Yuen\\
  Associate Professor\\
   Columbia University\\
  New York, New York, USA\\
  hyuen\imageat{}cs\imagedot{}columbia\imagedot{}edu\\
\end{authorinfo}
\end{dajauthors}

\end{document}